\documentclass[aps,pre,
amsmath,amssymb,nofootinbib]{revtex4-2}
\usepackage[utf8]{inputenc}
\usepackage{natbib}
\usepackage{algorithm}
\usepackage{algorithmic}
\usepackage{amsmath}
\usepackage{amsthm}
\usepackage{amsfonts}
\usepackage{geometry}
\usepackage{verbatim}
\usepackage{tabularx}
\usepackage{graphicx}
\usepackage{hyperref}
\usepackage{cleveref}
\usepackage{todonotes}

\newtheorem{definition}{Definition}
\newtheorem{lemma}{Lemma}
\newtheorem{corollary}{Corollary}
\newtheorem{theorem}{Theorem}

\newcommand\numberthis{\addtocounter{equation}{1}\tag{\theequation}}
\newcommand{\powerset}[1]{2^{#1}}
\newcommand{\kl}[2]{D_\mathrm{KL}\left(#1\mathrel\Vert#2\right)}
\newcommand{\prob}{\mathbb P}
\newcommand{\expect}{\mathbb E}
\newcommand{\meas}{\mathcal M_+}
\newcommand{\pmeas}{\mathcal M_1}

\begin{document}
	
	\title{Information dynamics and the arrow of time}
	
	\author{Aram Ebtekar}
	\email{aramebtech@gmail.com}
	\affiliation{Vancouver, BC, Canada
	}
	
	
	\begin{abstract}
		Why does time appear to pass irreversibly? To investigate, we introduce a class of partitioned cellular automata (PCAs) whose cellwise evolution is based on the chaotic baker's map. After imposing a suitable initial condition and restricting to a macroscopic view, we are left with a stochastic PCA (SPCA). When the underlying PCA's dynamics are reversible, the corresponding SPCA serves as a model of emergent time-reversal asymmetry. Specifically, we prove that its transition probabilities are homogeneous in space and time, as well as Markov relative to a Pearlean causal graph with timelike future-directed edges. Consequently, SPCAs satisfy generalizations of the second law of thermodynamics, which we term the Resource and Memory Laws. By subjecting information-processing agents (e.g., human experimenters) to these laws, we clarify issues regarding the Past Hypothesis, Landauer's principle, Boltzmann brains, scientific induction, and the so-called psychological arrow of time. Finally, by describing a theoretical agent powered by data compression, we argue that the algorithmic entropy takes conceptual precedence over both the Shannon-Gibbs and Boltzmann entropies.
	\end{abstract}
	
	\maketitle
	
	\section{Introduction}
	
	The complete trajectory of a dynamical system can be specified in two parts: (1) an \textbf{initial condition} that specifies its configuration at the initial time, and (2) \textbf{dynamics} that specify how the configuration evolves over time. It's widely believed that the Universe exhibits \textbf{CPT symmetry}: under a simultaneous reversal in charge (C) and parity (P), its dynamics are symmetric in time (T). In other words, every video recording remains physically valid when played in rewind \footnote{With minor caveats: the fact that it's CPT instead of just T means that every particle in the rewinded video behaves like the mirror image of its respective antiparticle.}. CPT is a rigorous theorem in axiomatic quantum field theory \citep{jost1957bemerkung}, standing in stark contrast to our everyday experience with irreversible phenomena. Indeed, fragile items drop and shatter, but don't unshatter and undrop; we \emph{remember} yesterday's events, but \emph{plan} tomorrow's agenda; causes precede effects. If the dynamics are truly symmetric, then, by a process of elimination, only a special choice of initial condition can break the symmetry.
	
	In particular, the initial condition must be set in such a way as to imply the \textbf{second law of thermodynamics}: a general principle of physics that forbids the \textbf{entropy} of any closed system from decreasing. Much work has gone into justifying various formal definitions of entropy, along with conditions that would imply the second law. Even if the second law is taken for granted, relating it to causal and decision-theoretic concepts remains a challenge.
	
	From Maxwell's demon to Poincar\'e recurrence to Landauer's principle, the subject of irreversibility has a history fraught with debates and misunderstandings \citep{bennett1982thermodynamics}. To increase confidence in our main results, we take a mathematically rigorous approach in which all modeling assumptions are laid bare. Modeling choices are made in an effort to balance conceptual clarity, generality, and applicability to real physics. We abstract away what we can from physical field theories, substituting a generic class of cellular automata in their place.
	
	Like our Universe, these automata have dynamics that are reversible microscopically, but not macroscopically; as such, they model the emergence of asymmetries with respect to time reversal, i.e., \textbf{arrows of time}. Despite deterministic dynamics, we discover precisely how Markovian probabilities emerge as a macroscopic behavior. By applying ideas from information theory and Pearlean causality to these automata, we collect insights on how the arrow of time operates. In particular, we learn why memories only inform us about the past, and why causes precede effects.
	
	Following the traditional sequence of physics pedagogy, this article only models classical information, leaving extensions to the quantum setting for future work. Our classical approach stands as a rebuttal against suggestions that the psychological arrow of time is a distinctly quantum phenomenon \citep{maccone2009quantum,zurek2018quantum}. The remainder of this introduction develops the background, motivation, and general structure of our approach.
	
	\subsection{Discrete-State Entropy}
	\label{sec:hdiscrete}
	
	Our models are inspired by classical mechanics, which operates on a continuous phase space. Before addressing the latter, we start with the discrete setting, where entropy has a clearer meaning. Even there, entropy can be understood in a number of ways. Intuitively, it's often thought of as a measure of combinatorial complexity, or ``disorder''. Conversely, a \emph{lack} of entropy is often thought of as a resource, akin to free energy; or as knowledge, narrowing down the set of possible states. Loosely speaking, if the object $s$ is a ``typical'' element of some finite set $\mathcal S$, then we'd like to say its entropy is $\log |\mathcal S|$.
	
	To make this precise, we can define entropy as a function of a probability measure $\mu$ on some countable set $\mathcal S$. The \textbf{Shannon entropy} of $\mu$ (also known as its \textbf{Gibbs entropy}) is given by
	
	\begin{align}
		\label{eq:hshannon}
		H(\mu) := \sum_{s\in\mathcal S} \mu(s)\log\frac 1{\mu(s)}.
	\end{align}
	
	The base of the logarithm is not important: since it only changes the entropy by a constant factor, a change of base is effectively just a change of units. For example, when the base is $2$, the entropy is given in \textbf{bits}. Note that when $\mu$ is uniformly distributed over some subset $A\subset\mathcal S$, we indeed have $H(\mu) = \sum_{s\in A}\frac{1}{|A|}\log|A|=\log|A|$.
	
	For non-uniform $\mu$, the Shannon entropy is motivated by supposing that $\mu$ is the joint distribution over a long sequence $s=(s_1,\ldots,s_n)$ from a stationary ergodic process. As a special case, the sequence can be i.i.d. (independent and identically distributed), representing, for instance, the state of $n$ particles in an ideal gas. As $n$ gets large, there is a sense in which $\mu$ approximates the uniform distribution on a ``typical set'' $\mathcal S_\mathrm{typ}$ of about $2^{H(\mu)}$ elements (where $H$ is given in bits). In particular, $H(\mu) \approx \log |\mathcal S_\mathrm{typ}|\approx nh$, where $h$ is the \textbf{entropy rate} of the process \citep[\S16.8]{thomas2006elements}.
	
	Note that $H$ is a function of $\mu$, not of $s$. In other words, the Shannon entropy appears to depend on some subjective measure of our knowledge about $s$, but not on its actual realization. Since we do not yet have the tools to reason about subjective knowledge (see \Cref{sec:appinduction}), we need a notion that's a direct function of $s$.
	
	In statistical mechanics, the most common approach is to \emph{coarse-grain} the state space into a union of \textbf{Boltzmann macrostates}: $\mathcal S = \cup_i\mathcal S_i$. The macrostates $\mathcal S_i$ are typically chosen to be disjoint; if not, identify each \textbf{Boltzmann microstate} $s\in\mathcal S$ with the smallest macrostate containing it. The \textbf{Boltzmann entropy} of $s$ is given by
	
	\begin{align}
		\label{eq:hboltzmann}
		B(s) := \min\{\log\mathcal |\mathcal S_i|: \mathcal S_i\ni s \}.
	\end{align}
	
	Typically, the partition index $i$ corresponds to a collection of macroscopic parameters, such as temperature, pressure, and chemical composition, specified to a reasonable degree of precision. There is debate over the validity of the coarse-graining procedure, since the Boltzmann entropy is very sensitive to the choice of partition.
	
	The theory of computation offers a third definition, which unifies and yields additional insights into the preceding approaches. To apply it, we must fix a universal Turing machine $U$, as well as a digital encoding $\langle s\rangle$ for each element $s\in\mathcal S$. The \textbf{algorithmic $U$-entropy} of $s\in\mathcal S$ (also known as its \textbf{Kolmogorov complexity}) is then defined to be the length of the shortest program $p$ \footnote{In the literature, $p$ is usually restricted to a prefix-free language. \citet[\S2-3]{li2019introduction} discuss in depth the relationship between plain and prefix-free versions of the algorithmic entropy. For the purposes of this article, the two notions are close enough that we can ignore the distinction.}, whose output is $\langle s\rangle$:
	
	\begin{align}
		\label{eq:halgorithmic}
		K(s) := \min\{|p|: U(p)=\langle s\rangle\}.
	\end{align}
	
	If $p$ is written in binary, then $K$ is given in units of bits. Once again, there is some room to fudge the definition, by changing either the universal machine or the encoding. Fortunately, for ``reasonable'' choices, the changes would be very small. For instance, let's see if some other machine $T$ can yield substantially smaller complexities than $U$. The universality of $U$ means that it encodes $T$ as some string $\langle T\rangle$, satisfying $U(\langle T\rangle,p)=T(p)$ for all programs $p$. Thus, the $U$-entropy of a system cannot exceed its $T$-entropy by more than the constant $|\langle T\rangle|$.
	
	For a Universe with an effective minimum length scale (i.e., the Planck length), it's plausible that there would only exist finitely many functionally distinct CPU designs (or Turing machine heads) with diameter smaller than, say, a millimeter. Thus, $U$ can be chosen in such a way that every machine $T$ with a small physical implementation also has a short encoding $\langle T\rangle$ \footnote{If a small physical implementation is discovered for such a $U$, then all candidates for $U$ would have short encodings for one another, so $K$ would be determined to within a small additive term.}. For large complex systems described by very large $|p|$, the contribution $|\langle T\rangle|$ of any small machine becomes negligible in comparison. As for smaller systems, in practice they are always embedded within a larger context: \emph{conditional} forms of the algorithmic entropy or universal prior can be used to quantify a subsystem's contribution to a larger system's entropy \citep[\S3.8 and \S5.2]{li2019introduction}.
	
	Thus, we argue that $K$ is the most objective and fundamental definition of entropy. This perspective finds additional support in \Cref{sec:appcompression}, where compressible strings are shown to serve as a resource. In practice, $H$ and $B$ are often easier to estimate; we now review conditions under which they closely approximate $K$. Obviously, as long as such conditions hold, $H$ and $B$ are approximately invariant to specific details of the chosen distribution or coarse-graining partition.
	
	Let's first look at the Boltzmann entropy $B$. Any coarse-graining yields a \textbf{two-part code} $p$ for each $s\in\mathcal S$, consisting of (1) the macrostate $\mathcal S_i$'s description, and (2) an index of length $B(s)=\log |\mathcal S_i|$, identifying $s$ within $\mathcal S_i$. Suppose that (1) macrostates have short descriptions
	\footnote{The number of macrostates with short descriptions is necessarily limited. This suggests an interesting coarse-graining: we can admit as macrostates \emph{all} subsets of $\mathcal S$ whose descriptions are shorter than some fixed bound. If we do, the Boltzmann entropy becomes the \textbf{Kolmogorov structure function} \citep[\S5.5]{li2019introduction}.
		In practice, macrostates are usually specified in terms of a handful of parameters such as temperature, pressure, and chemical composition. Describing these to a satisfactory degree of precision takes on the order of $10^3$ bits. This is negligible in comparison to the entropy of a macroscopic system, which typically exceeds Avogadro's number ($>10^{23}$ bits). Thus, the code's second part indeed dominates.}
	, and (2) $s$ is \textbf{algorithmically typical} (i.e., Martin-L\"of random \citep[\S2.4]{li2019introduction}) with respect to the uniform distribution on $\mathcal S_i$. Then, the two-part code is approximately optimal, and its length is dominated by the index part \citep[\S5.5]{li2019introduction}. Therefore, $K(s) \approx B(s)$.
	
	Now, let's look at the Shannon entropy $H$. As we saw previously, when $\mu$ describes a sufficiently long sequence from a stationary ergodic process, $H(\mu)\approx\log|\mathcal S_\mathrm{typ}|$. This looks like a Boltzmann entropy with macrostate $\mathcal S_\mathrm{typ}$, this set being determined by the measure $\mu$, which in turn is a function of the underlying process and sample size. To satisfy condition (1) from the previous paragraph, the process should have a short description. With high probability, samples $s$ drawn from $\mu$ satisfy condition (2). When both conditions hold, it follows that $K(s)\approx H(\mu)$. Intuitively, this works because $H$ measures the length of an encoding that's designed to be shortest \emph{on expectation} and, by the law of large numbers, empirical averages converge to expectations. For analyses of the rate of convergence between $H$ and $K$, see \citet{austern2019gaussianity}.
	
	Depending on which entropy is considered, the second law of thermodynamics may be formulated to say that a decrease in entropy is strictly forbidden, or merely unlikely. The Poincar\'e recurrence theorem says that ergodic dynamical systems must visit every state $s\in\mathcal S$ infinitely often. Thus, $B$ and $K$ will eventually be high, and then low again. In contrast, for many of the same systems, $H$ can never decrease \citep[\S4.4]{thomas2006elements}. To reconcile this difference, recall that  Shannon entropy concerns a probability distribution $\mu$, and increases as $\mu$ becomes more uniform. Typical draws from $\mu$ have Boltzmann and algorithmic entropies close to $H(\mu)$. However, on rare occasions, a distribution of high (Shannon) entropy produces a lucky draw of low (Boltzmann or algorithmic) entropy.
	
	Aside from these overwhelmingly unlikely draws, we've seen that in realistic physical settings, the entropies approximate one another. Therefore, we're free to choose whichever we find convenient, with the assurance that our results will transfer. In this article, we choose the Shannon entropy (and the KL divergence that generalizes it; see \Cref{sec:klinfo}), because it offers the conveniences of probability theory. Nonetheless, whenever we formally discuss collections of random variables, the reader is encouraged to informally think of algorithmically typical draws from their distributions. In this manner, we can make sense of such modeling choices as random initial conditions, despite the observable Universe only having one initial state.
	
	To mimic a symmetry of classical mechanics, our model Universe should apply a deterministic and reversible state transition at each time step. This amounts to a fixed permutation on $\mathcal S$. Clearly, the Shannon entropy is invariant to such permutations. Similarly, the algorithmic entropy changes extremely slowly under repeated permutations, since any future state can be computed as a function of the initial state, the permutation, and the number of repetitions. \emph{Substantial} changes in entropy only arise by restricting attention to an incomplete view of $\mathcal S$, allowing entropy to pop in or out of our view. For example, views may be limited in any of the following:
	
	\begin{enumerate}
		\item Numerical precision. The entropy would then be insensitive to microscopic variations, until they are expanded by \emph{chaotic} dynamics.
		\item Spatial extent. Interactions between systems may induce correlations, so that each individual system's entropy increases even if the aggregate entropy does not.
		\item Computational resources. If repetitions of the dynamics are slow to compute, they may increase a system's \emph{time-bounded} algorithmic entropy, defined by restricting \Cref{eq:halgorithmic} to programs $p$ that terminate within a specified time limit.
	\end{enumerate}
	
	We mention the third mechanism only for completeness, leaving its rigorous study to future work. The second mechanism is related to locality, and will be quantified in \Cref{eq:spcadecomp}. However, we focus our modeling efforts on the first mechanism, in which chaos draws out the entropy hidden within microscopic variations, gradually pushing it into the macroscopic world. To ensure a limitless supply of hidden entropy, we now extend the state space to a continuum.
	
	\subsection{Continuous-State Entropy}
	\label{sec:hcontinuous}
	
	In classical Hamiltonian mechanics, the discrete state space $\mathcal S$ is replaced by a continuous \textbf{phase space}. For example, a system of $n$ free particles has its state described by a point in $\mathbb R^{6n}$, representing their 3D positions and momenta. Given a probability measure $\mu$ over phase space, the Shannon (or Gibbs) \textbf{differential entropy} is defined by replacing the sum in \Cref{eq:hshannon} by an integral. A defining characteristic of Hamiltonian evolution is \textbf{Liouville's theorem}: it states that phase volumes, and hence the differential entropy, do not change with time \citep{carcassi2020hamiltonian}.
	
	On the other hand, the continuous setting is conducive to views with limited numerical precision. For concreteness, fix a resolution $\varepsilon>0$. To each point in $\mathbb R^{6n}$, apply the coordinate-wise discretization $x\mapsto\lfloor \frac x\varepsilon \rfloor\varepsilon$. It partitions the phase space into axis-aligned hypercubes of width $\varepsilon$, identified with elements of a discretized state space $\mathcal S:=(\varepsilon\mathbb Z)^{6n}$. Since $\mathcal S$ is countable, all of the discrete-state entropies $H,B,K$ apply. Despite Liouville's theorem, the discretized entropies can be altered by changing the shape of a phase volume, so that it occupies different numbers of $\varepsilon$-cubes.
	
	Note that the computation of $B$ now involves two levels of partitioning: (1) the \emph{fine partition} of continuous phase space into its discretized version, and (2) the \emph{coarse partition} of discretized phase space into Boltzmann macrostates. For the remainder of this article, the word \textbf{microscopic} refers to the full continuous description, whereas \textbf{macroscopic} refers to the discretization into fine macrostates. Although the latter are considered \emph{microstates} in the Boltzmann picture, there will be no ambiguity as we've chosen to focus on the Shannon entropy. In light of the correspondence between $B$ and $H$ in \Cref{sec:hdiscrete}, Boltzmann's coarse macrostates are replaced with probability distributions over fine macrostates.
	
	For certain volume-preserving dynamical systems, we'll find initial conditions that ensure the discretized Shannon entropy never decreases. The idea is to employ a particularly convenient chaotic map.
	
	\subsection{The Baker's Map and the Shift Map}
	\label{sec:baker}
	
	We consider a discrete-time map that shares many properties with Hamiltonian mechanics, but is more tractable. It substitutes the unit square $\mathcal S_\mathrm{baker} := [0,1)^2$ for phase space. Fix $m\in\mathbb N$. The $m$-fold \textbf{baker's map} $\beta:\mathcal S_\mathrm{baker}\rightarrow\mathcal S_\mathrm{baker}$ acts on the square by stretching its $x$-axis while squeezing its $y$-axis, each by a factor of $m$, then cutting and re-arranging to get back the unit square:
	\[\beta(x,y) = \left(mx - \lfloor mx\rfloor,\;
	\frac{y + \lfloor mx\rfloor}m\right).\]
	
	The combination of stretching and squeezing makes $\beta$ volume-preserving and chaotic. It also satisfies a CPT-like symmetry, since
	\begin{equation}
		\label{eq:bakercpt}
		\beta^{-1} = \rho\circ\beta\circ\rho,
	\end{equation}
	
	where $\rho(x,y):=(y,x)$ is a reflection about the line $y=x$. Since $\beta$ is invertible, any (possibly random) initial state $C_0=(X_0,Y_0)\in\mathcal S_\mathrm{baker}$ determines the state at all times $t\in\mathbb Z$ by iterating
	\[C_{t+1} = \beta(C_t),\]
	
	forward and backward. Given the distribution of $C_0$ and the resolution $\varepsilon$, we are interested in the Shannon entropy of $C_t$'s discretization at various $t$. In light of the symmetry in \Cref{eq:bakercpt}, we cannot preferentially guarantee higher entropy at higher $t$.
	
	Instead, provided that the initial distribution is absolutely continuous, the entropy tends toward a maximum in a symmetric fashion, as $t\rightarrow\pm\infty$. Some intuition is offered by the special case in which $C_0$ is uniformly distributed on a union of axis-aligned squares, whose vertex coordinates are multiples of $\delta$, with each of $\varepsilon,\delta$ being integer powers of $\frac 1m$. The reader may verify that the entropy is non-decreasing for $|t| \ge \log_m \frac 1\varepsilon - \log_m \frac 1\delta$, and is maximal at $|t| \ge \log_m \frac 1\varepsilon + \log_m \frac 1\delta$. In a much more general setting, \Cref{thm:markovapprox} will imply an asymptotic version of the second law of thermodynamics as $t\rightarrow\pm\infty$.
	
	$\beta$ is a prototypical example of a chaotic map: it's physically realizable \citep{pitowsky1996laplace} and guides our geometric intuition. It also has a convenient symbolic representation. To that end, let $\mathbb Z_m := \{0,1,\ldots,m-1\}$. The $m$-symbol \textbf{shift map} $\sigma:\mathcal S_\mathrm{shift}\rightarrow\mathcal S_\mathrm{shift}$ is defined on the set of bi-infinite sequences $\mathcal S_\mathrm{shift} := (\mathbb Z_m)^\mathbb Z$, by simply shifting the entire sequence $r\in\mathcal S_\mathrm{shift}$ over by one position:
	\begin{equation}
		\label{eq:shift}
		\sigma(r)_i := r_{i+1}.
	\end{equation}
	
	Despite their different appearances, the maps $\beta$ and $\sigma$ are essentially equivalent, in the sense that the correspondence $\Psi:\mathcal S_\mathrm{shift}\rightarrow\mathcal S_\mathrm{baker}$, defined by
	\begin{equation}
		\label{eq:bakershift}
		\Psi(r) := \left(\sum_{i=1}^\infty r_{i-1} m^{-i} \!\!\!\!\mod 1,\;
		\sum_{i=1}^\infty r_{-i} m^{-i}  \!\!\!\!\mod 1 \right),
	\end{equation}
	
	is \emph{almost} bijective. $\Psi$ interprets the sequence's two halves $(r_{\ge 0},\,r_{<0})$ as representations, in the base-$m$ positional number system, of the $(x,y)$ coordinates of a point in $\mathcal S_\mathrm{baker}$. The only trouble is that points of $\mathcal S_\mathrm{baker}$ with $m$-adic rational coordinates have multiple preimages, due to non-unique representations such as $0.5 = 0.4\overline 9$. Fortunately, these points collectively have measure zero, so we would not expect to encounter them on a random draw. On their complement, one can check that the dynamics are equivalent:
	\[(\beta\circ\Psi)(r) = (\Psi\circ\sigma)(r).\]
	
	In summary, the baker's map $\beta$ serves as a convenient substitute for classical Hamiltonian dynamics, sharing key features such as determinism, reversibility, volume-preservation, chaos, and the second law of thermodynamics. Motivated by the correspondence $\Psi$, we will base our models on the shift map $\sigma$. If the initial condition is suitably randomized, each term in the sequence $r\in\mathcal S_\mathrm{shift}$ acts as a random seed, which can be used to emulate stochastic transitions. We will design macroscopic transitions that depend on only one term, say $r_0$. $\sigma$ shifts these terms like a conveyor belt, providing a fresh $r_0$ at every time step.
	
	As a result, our models' stochastic macroscopic descriptions are \emph{exactly} observationally equivalent to their deterministic microscopic descriptions. We prove it formally in \Cref{thm:equivalence2}. This appears to be not quite true of the real world, where microscopic theories have additional predictive power. Nonetheless, \citet{werndl2009deterministic} argues that such equivalences do hold, either exactly or approximately, in many settings of practical interest.
	
	\subsection{The Psychological Arrow of Time}
	\label{sec:psychologicalarrow}
	
	Had we only wanted a deterministic reversible model with a second law of thermodynamics, then either of the maps $\beta$ or $\sigma$, suitably initialized, would have sufficed. However, the second law is insufficient for explaining the psychological arrow of time. To begin with, it cannot explain the time-reversal asymmetry of memories.
	
	To see why, we must understand what memories are. Intuitively, they are \emph{records} of past events. They contain information about the event in question, but that's not all: while we can also collect information that's predictive of \emph{future} events, we would not call them records. The distinction is clearest with events that depend on highly chaotic interactions, making them hard to predict. The weather is a good example: our records span much of human history (and prehistory, if we include geological evidence of climate variations). In contrast to how much we know about past weather, we cannot reliably forecast even one week into the future. To give another example, all manner of unpredictable scenes can be recorded in a photograph, but only after they occur. In this sense, memories need not even be associated with living beings; rather, their defining feature is the ability to stably capture the information content of an event, no matter how ephemeral.
	
	Just as memories pertain to the past, counterfactuals pertain to the future. Whereas the existence of records suggests that the past is set in stone, the future appears malleable: we imagine various counterfactual possibilities, and plan our lives accordingly. Outside the purview of fundamental physics, nearly all academic disciplines take the arrow of time for granted. This is reflected in their mathematical tools: in particular, Pearl's \textbf{structural causal models} \citep{pearl2009causality} can model both memories and counterfactuals.
	
	Recall that a \textbf{directed graph} is a set of vertices, along with a set of edges (arrows) that point from one vertex to another. A \textbf{causal graph} is a directed graph with no cycles, whose vertices correspond to random variables. Its edges specify causal relationships, pointing from cause to effect. More precisely, they encode two kinds of information \citep{pearl2013structural}:
	\begin{enumerate}
		\item A set of \textbf{structural independence} constraints that the random variables must satisfy.
		\item Rules for modifying the variables' joint distribution under \textbf{structural interventions}.
	\end{enumerate}
	
	A graph's structural independence constraints are formally specified by Pearl's \textbf{$d$-separation criterion} \citep[\S1.2]{pearl2009causality}. A joint probability distribution over a causal graph's random variables is said to be \textbf{Markov}, relative to the given graph, if it satisfies these constraints.
	
	A structural causal model is a causal graph equipped with conditional probability distributions, one for each variable as a function of its parents. These distributions, together with the graph's structural independence constraints, fully determine the variables' joint probability distribution \citep[\S1.2]{pearl2009causality}. In addition to this joint distribution, a causal model also identifies alternative distributions that would arise, if external influences were to intervene on some of its variables.
	
	For an illustrative example, imagine a collection of physical systems that are usually isolated from one another, but occasionally experience pairwise interactions. Let's zoom in on two such systems $A$ and $B$, whose sole interaction within the time interval $[0,3]$ occurs at $t=1.5$. Their causal graph is as follows:
	\begin{align*}
		\cdots\rightarrow A_0\rightarrow A_1&\rightarrow A_2\rightarrow A_3\rightarrow\cdots
		\\&\searrow
		\\\cdots\rightarrow B_0\rightarrow B_1&\rightarrow B_2\rightarrow B_3\rightarrow\cdots
	\end{align*}
	
	$A_t$ is a random variable representing the state of system $A$ at time $t$, and likewise for $B_t$. At non-interaction times, the two systems are modeled as independent Markov chains. During the interaction, $A_1$ is said to \emph{causally influence} $B_2$. Two interpretations of this graph offer different insights: either (1) we focus on $A$, a system whose data is \emph{recorded} by the memory $B$; or (2) we focus on $B$, a system whose dynamics are \emph{intervened} upon by the causal actor $A$. Let's discuss each interpretation in turn.
	
	\subsubsection{Memories}
	\label{sec:intromemory}
	
	As a consequence of the $d$-separation criterion, any pair of random variables without a common ancestor (i.e., a common cause) must be independent. Thus, dependent random variables must carry information about some common cause in the past. In our example, $A_t$ and $B_u$ share the common ancestor $A_1$, for $t\ge 1$ and $u>1$. Earlier states, such $A_0$ and $B_0$, can only be dependent if they share a common ancestor in the more distant past (implied by the ellipsis).
	
	Motivated by our observation that memories can capture ephemeral events, let's imagine that $A$ is a rapidly mixing system, so that the sequence $(A_t)_{t\in\mathbb Z}$ is i.i.d. All directed edges (arrows) in the first row are effectively erased, leaving $A_1$ as the only possible common ancestor. By the $d$-separation criterion, $A_t$ and $B_u$ can only be dependent if $t=1<u$. Thus, memories are indeed of the past.
	
	Now let's drop the i.i.d. assumption. The dependence between $A_t$ and $B_u$ can be quantified by their \textbf{mutual information} $I(A_t;B_u)$ (see \Cref{def:mutinf}). The data processing inequality \citep[\S2.8]{thomas2006elements} implies $I(A_{t'};B_{u'}) \ge I(A_t;B_u)$ for $t'\le t$ and $u'\le u$, provided that no interactions occur between the times involved. In other words, information loss is unrestricted, but information \emph{gain} requires interaction. Thus, all dependences can be traced backward through time to some past interaction (or to dependent initial conditions).
	
	The $d$-separation criterion and the data processing inequality, much like the second law of thermodynamics, impose asymmetries with respect to time reversal. To see how they work in concert together, let's decompose the systems' joint entropy (dropping the subscripts for convenience):
	
	\begin{equation}
		\label{eq:mutinfsimple}
		H(A,\, B) = H(A) + H(B) - I(A;\, B).
	\end{equation}
	
	The composite system $(A,B)$ is closed throughout the time interval $[0,3]$, so the second law implies that $H(A,B)$ is non-decreasing. In general, when $A$ and $B$ interact, their joint entropy may arbitrarily redistribute among the three right-hand terms. However, when $A$ and $B$ are closed off from each other, the second law applies to each system individually, so that $H(A)$ and $H(B)$ are also non-decreasing.
	
	Now we see why the second law on its own is insufficient: it permits an increase in $I(A; B)$, provided that either $H(A)$ or $H(B)$ simultaneously undergo an equal or greater increase. In other words, the second law does not rule out spontaneous information gain between isolated systems. Instead, the required time-reversal asymmetry is given by the data processing inequality. Together with the second law, it ensures that all three right-hand contributions in \Cref{eq:mutinfsimple} are non-decreasing.
	
	An interesting consequence is that when $A$ and $B$ are closed off from each other, an increase in any one of the right-hand terms cannot be compensated for by a decrease in another; therefore, it necessarily induces an equal or greater increase in $H(A,B)$. This is an extension of Landauer's principle, detailing the irrevocable thermodynamic cost of erasing information; see \Cref{sec:applandauer}.
	
	\subsubsection{Counterfactuals}
	\label{sec:introcounterfactual}
	
	Let's examine the relationship between cause and effect. We say that rain causes wet grass, whereas wet grass does not cause rain. The inelastic collisions between rain drops and grass increase entropy, rendering them thermodynamically irreversible. Note that while a \emph{typical} configuration of wet grass does not evolve into rain, our wet grass is decidedly \emph{not} typical, because evolving it backward reveals a past in which it did rain. The forward evolution is distinguished in that it can be predicted by treating the configuration \emph{as if} it were typical, using homogeneous transition probabilities. In contrast, the backward evolution seemingly conspires to extract the grass's moisture and blast it up into the sky. The emergence of homogeneous forward dynamics will be a key feature of our models.
	
	However, our causal intuitions are not confined to processes that are intrinsically irreversible. For example, consider a frictionless pendulum with a period of two seconds. Since its trajectory is periodic, any point on the trajectory causes all others; nonetheless, we prefer to think of earlier states as causing later states. This makes sense in the context of exogenous influences: if we place the pendulum to the left of its pivot, then we would find it on the right a second later, whereas the placement has no retroactive effect a second prior.
	
	More generally, we imagine being able to perform an \textbf{intervention} on any random variable of a causal model, assigning it some value and detaching it from its parents in the graph. Effects then propagate to the variable's descendants, yielding a \textbf{counterfactual} outcome. This counterfactual presents an alternative trajectory for the Universe, in apparent violation of determinism. Moreover, since the edges are directed forward in time, the intervention changes its future while holding its past fixed, in apparent violation of CPT symmetry. Despite violating fundamental physics principles, counterfactual reasoning is ubiquitous in real-life applications \citep{pearl2009causality}.
	
	At this point, we remark that our investigation inevitably touches upon philosophical themes; nonetheless, all of our questions are scientifically motivated. For instance, while we do not concern ourselves with the question of whether free will is ``real'', we must reckon with the empirical fact that intelligent life displays the behaviors associated with free will. These include reasoning in terms of counterfactuals, in order to plan for the future (but not for the past). How can physics predict this phenomenon?
	
	Consider that practical models are necessarily incomplete. Returning to our casual model above, rather than draw it in full, we might only model the system $B$:
	
	\begin{align*}
		\cdots\rightarrow B_0\rightarrow B_1&\rightarrow B_2\rightarrow B_3\rightarrow\cdots
	\end{align*}
	
	Viewed in isolation, $B$ is a Markov chain, following its own dynamics in the absence of external influence. For example, $B$ might be our frictionless pendulum. The interaction, then, appears as an exogenous intervention on $B_2$, modifying its $B_1$-conditioned probability distribution. The unmodeled variable $A_1$ encodes the nature of the intervention. As we saw when discussing memory, information about $A_1$ only appears in $B$'s future, explaining why counterfactuals do not alter the past. If some intelligent agent has partial or total control of $A_1$'s value, then it makes sense for the agent to consider each counterfactual possibility, exercising its ``will'' to choose $A_1$ according to the desired effect on $B$.
	
	In general, the behavior of agents is determined by Darwinian evolution. In \Cref{sec:appparadox}, we propose an explanation for evolution's arrow of time; thus, agents seek to optimize their future. If an agent possesses sufficient cognitive resources, it may employ a model-based planning algorithm \cite{doll2012ubiquity}. From the algorithm's point of view, choices are made ``freely'', in the following sense: the algorithm can simulate the counterfactual outcomes of a variety of different choices, and realize whichever one it prefers. This is true despite the algorithm having a mathematically determined (or possibly randomized) output.
	
	This discussion leaves us in an awkward position: causal models beautifully describe the cognitive world of cause, effect, memory, and planning, i.e., the psychological arrow of time; however, they seem quite far removed from fundamental physics, with its determinism and CPT symmetry. In this article, we extend the shift map, in two stages, toward the ultimate aim of deriving causality as an emergent structure. The result is a type of cellular automaton that, even more so than the baker's and shift maps, exhibits key features of classical mechanics.
	
	By studying these automata, we find that the fundamental time-reversal asymmetry takes one of two forms, depending on whether we take a microscopic or macroscopic view. At the microscopic level, symmetry is broken by a Past Hypothesis, consisting of an absolutely continuous initial condition. At the macroscopic level, this corresponds to dynamics that are homogeneous in space and time, and Markov relative to a causal graph who edges are timelike and future-directed. As a consequence, we can prove generalizations of the second law of thermodynamics. These lead to insights on a variety of topics related to the arrow of time.
	
	\subsection{Article Outline}
	\label{sec:outline}
	
	\Cref{sec:related} reviews some literature that provides context for our contributions. Then, \Cref{sec:prelim} sets notational conventions. It also introduces useful lemmas from information theory \citep{thomas2006elements}, reformulated to suit the analysis of general Markovian dynamics. Our core mathematical contributions are split into two stages, developed in \Cref{sec:markov,sec:spca} respectively.
	
	\Cref{sec:markov} presents the first stage, extending the shift map in such a way as to emulate arbitrary time-homogeneous Markov chains. We discuss the extent to which Markov chains exhibit or violate time-reversal symmetry, and prove a precise correspondence between three presentations of the dynamics. Two of these are stochastic, corresponding to macroscopic views of the third, deterministic presentation. We apply this theory to generalize the Past Hypothesis, and then ask what happens at negative times, before the initial condition.
	
	\Cref{sec:spca} presents the second stage, enriching the Markov chains with spatial structure that turns them into \textbf{stochastic partitioned cellular automata} (SPCAs). The local dynamics of an SPCA makes it possible to discuss multiple systems, each having their own entropy, and correlations between them. We find that SPCAs are Markov relative to a timelike future-directed causal graph. Moreover, they satisfy the Resource and Memory Laws, which correspond to the second law of thermodynamics and the data processing inequality, respectively.
	
	\Cref{sec:discussion} explores applications of our work to a variety of topics, in and out of physics. In some cases, we obtain novel extensions of existing theory; in others, we simply demonstrate the increased clarity and precision offered by SPCAs and their information-theoretic dynamics.
	
	Finally, \Cref{sec:conclusions} concludes with some suggestions for future work.
	
	\section{Related Works}
	\label{sec:related}
	
	The arrow of time is too broad a topic to comprehensively review in a few pages. We make our best effort to emphasize the most pertinent literature that provides context for this article.
	
	\subsection{Literature on Dynamical Systems}
	\label{sec:relateddyn}
	
	For simplicity's sake, we focus on the discrete-time (but continuous-state) setting, to which \citet{devaney2021introduction} provides an accessible introduction. In practice, we cannot observe a system's state precisely, so let's imagine we instead observe a discrete-valued function of the state; i.e., the set of possible observations is countable. Then, for all practical purposes, a state may be identified with its trajectory's bi-infinite sequence of past, present, and future observations. The shift map describes the system's dynamics in this sequence representation, the study of which is called symbolic dynamics \citep{lind2021introduction}.
	
	Despite the system's determinism, if its initial state is randomly sampled, the corresponding sequence of observations becomes a discrete-state stochastic process. \citet{werndl2009deterministic} examines deterministic and stochastic systems that are \emph{observationally equivalent}, in the sense that the former's observations are distributed like the latter. Under reasonable hypotheses, she proves that Bernoulli dynamical systems are approximately observationally equivalent, at every resolution, to irreducible and aperiodic Markov chains. In \Cref{thm:equivalence2}, we prove exact observational equivalence, between general Markov chains and an extension of the shift map. Under more general initial conditions, \Cref{thm:markovapprox} turns this into an approximate equivalence. We also relate properties between the two systems, showing how measure-preservation and time-reversal (on the deterministic side) map to stationarity and duality (on the stochastic side).
	
	The literature includes several dynamical systems modeling the second law of thermodynamics. \citet{schack1996chaos} interleave a chaotic system's dynamics with random measure-preserving perturbations. Unfortunately, their perturbations sacrifice determinism and CPT symmetry. For some simple dynamical systems, \citet{nicolis1988master} identify suitable initial conditions and coarse-grainings that ensure the Markov property. \citet{pitowsky1996laplace} also studies generalized shift maps, but for the purposes of modeling deterministic Turing machines. More recently, researchers such as \citet{parrondo2015thermodynamics} are looking more closely at the thermodynamics of correlated systems.
	
	Perhaps the closest existing model to ours is the multibaker chain of \citet{gaspard1992diffusion}. It employs the baker's map to gradually extract a source of randomness that's set at initialization. By doing so, it deterministically and reversibly emulates a random walk on $\mathbb Z$, i.e., a discrete diffusion process. Our first-stage model in \Cref{sec:markov} generalizes the multibaker chain, emulating not only random walks, but arbitrary Markov chains.
	
	Continuous-state analogues of the entropy rate \citep[\S4.2]{thomas2006elements} are given by the Kolmogorov-Sinai entropy \citep{entropybook} and the topological entropy \citep{adler1965topological}. These quantities measure the rate at which a system's trajectories acquire new information, as a result of expanding once-hidden microscopic variations. For the $m$-symbol shift map, they equal $\log m$. \citet{latora1999kolmogorov} describe some circumstances in which the Kolmogorov-Sinai entropy approximates the rate of entropy growth in dynamical systems.
	
	Some concepts from the computer science literature are helpful. \citet{bennett1982thermodynamics} reviews connections between logical and physical descriptions of entropy and reversibility. We draw upon the design of reversible and partitioned cellular automata (PCAs) \citep{morita1989computation,kari2018reversible}, essentially fusing them with causal models to create SPCAs. In the context of reversible computing, \citet{frank2018physical} anticipates thermodynamic consequences for the data processing inequality.
	
	Subtle issues come up when comparing Gibbs and Boltzmann entropies, largely revolving around ambiguous coarse-grainings and subjective probabilities \citep{frigg2019statistical}. Algorithmic entropy provides a firm theoretical grounding that minimizes these ambiguities. Its foundations date back to Occam's razor, and it appears in some ambitious accounts of inductive inference and artificial intelligence \citep{solomonoff1964formal,hutter2004universal,rathmanner2011philosophical}. \citet{zurek1989algorithmic} introduces the algorithmic entropy to physics. \citet{wallace2005statistical} describes two-part codes in depth, and adds some speculations on the arrow of time. Last, but not least, \citet{li2019introduction} provide the standard reference on algorithmic information theory, surveying applications to physics in their final chapter.
	
	\subsection{Literature on the Psychological Arrow of Time}
	\label{sec:relatedtime}
	
	The psychological arrow of time is intimately tied to the sense that causes precede effects. Historically, before we had a mathematical language in which to discuss them, causal concepts were subject to misunderstanding, controversy, and even outright dismissal in the scientific community. That changed with the introduction of Pearl's structural causal models, a powerful methodology for causal inference, applicable to a wide range of disciplines \citep{pearl2009causality}.
	
	In this article, we understand causality exactly as Pearl does. In particular, we discuss counterfactuals using his structural semantics; though, see \citet[\S7.4.1]{pearl2009causality} for a comparison against Lewis' closest-world semantics, and \citet{loewer2020mentaculus} for semantics based on probabilistic conditioning. One advantage of the structural account is that, because it need not alter the past, it deals more cleanly with implausible (yet still useful) counterfactuals.
	
	Ties between Pearlean causality and dynamical systems are sparser in the literature. The numerical study by \citet{paluvs2018causality} presents insights on some causality detection methods.
	
	\citet{janzing2016algorithmic} link causality to the algorithmic independence of initial conditions and dynamics. In the deterministic setting, this approach predicts the algorithmic entropy to grow unrealistically slowly: the length of a later state's description, in terms of only the starting state, dynamics, and time elapsed, grows at most logarithmically in the time elapsed. On the other hand, their result makes a lot more sense in the context of random dynamics, where the description length can grow linearly in the time elapsed. Using \Cref{def:markovr}, we offer the same conclusions in probabilistic, rather than algorithmic, terms.
	
	In a deeply thought-provoking treatise, \citet{albert2001time} elucidates the philosophical issues in reconciling statistical physics with the psychological arrow of time. He argues, alongside \citet{loewer2020mentaculus}, that the arrow of time emerges from the \emph{Mentaculus}. This is a probability distribution over Universes, defined by extrapolating a \emph{Past Hypothesis} in which the Big Bang is sampled uniformly from a low-entropy macrostate. Our article supports this conclusion with rigorous theorems. We also extend the Past Hypothesis to non-uniform continuous distributions, yielding a generalized Mentaculus in \Cref{eq:markovapprox1,eq:markovapprox2}.
	
	In the search for rigorous explanations, a common line of attack takes the second law of thermodynamics as given, and focuses on what appears to be a basic component of the psychological arrow: memory. After proposing a definition for memory, some authors proceed to argue that its operation must align with the second law's arrow of time. In light of our observation following \Cref{eq:mutinfsimple}, that the second law alone is insufficient, this approach is unlikely to succeed. Even so, it's helpful to examine some of these attempts with care.
	
	\citet{wolpert1992memory} distinguishes between two types of memory systems: \emph{computer-type} and \emph{photograph-type}. Both types aim to possess information at some time $t_r$, about an event occurring at some other time $t_e$. A computer-type memory has access to so much state information at $t_r$, that it can deduce the event at $t_e$ by directly simulating the dynamical evolution of the Universe. Such ``memories'' have no arrow of time, with both $t_e<t_r$ and $t_e>t_r$ being admissible, but they may require unrealistic levels of precision. As such, they are more like weather predictions rather than reliable records.
	
	Our notion of a memory aligns better with Wolpert's photograph-type which, while less demanding in terms of precision and computation, requires initialization. Wolpert argues that real-world initialization procedures result in a net increase in entropy, forcing the memory to align with the thermodynamic arrow. Unfortunately, this is begging the question: the thermodynamic arrow makes it so nearly \emph{all} real-world processes increase entropy. If anything, the many-to-one mapping that defines an initialization seems predisposed to \emph{lose} entropy, which can only be done in reverse; in order to initialize forward, entropy must be released into the environment just to stay in accordance with the second law. A perfectly efficient initialization procedure would be reversible, with the entropy gains and losses exactly canceling \citep{bennett1982thermodynamics}.
	
	\citet{mlodinow2014relation} argue that even a reversibly implemented memory can only function in the direction of the thermodynamic arrow. Their thought experiment consists of a pair of connected chambers containing elastic particles, and a counter that tracks the net flow of particles from one chamber to the other. Instead of treating initialization explicitly, the counter is assumed to have a known value at $t_e$. Thus, reading the counter's value at $t_r$ is enough to infer the net flow of particles during the time interval between $t_e$ and $t_r$, acting as a memory of this interval. In order for this memory to be useful, it should be capable of recording more than one possible answer, which the authors ensure by applying a small random perturbation to the particles at $t_e$. This directs the thermodynamic arrow of time away from $t_e$, presumably toward $t_r$, thus aligning the two arrows.
	
	Their argument has a few loose ends. Firstly, it doesn't apply if the particles are at thermodynamic equilibrium, as no thermodynamic arrow would be discernible from their evolution. Secondly, it's unclear how one arrives at knowledge of the counter's value at $t_e$. And finally, a random perturbation at $t_e$ cannot be a realistic model of uncertainty: it would direct the thermodynamic arrow away from $t_e$ in \emph{both} directions, violating the Universe's low-entropy initial condition. While the argument captures a valid intuition, that memories should be robust to counterfactual alternatives, such a principle demands careful justification.
	
	\citet{rovelli2020memory} presents a physically plausible model consisting of two near-equilibrium systems with long thermalization times: a cooler memory and a hotter environment. In a follow-up paper, \citet{rovelli2020agency} views the environment as an agent, and its random interactions with the memory as choices. Thus, he concludes that the exercise of free will increases entropy, and must therefore align with the thermodynamic arrow. His account of free will in terms of randomness disagrees with our account in \Cref{sec:introcounterfactual}, where we equated free will with counterfactual-based planning. Our account is agnostic to whether the output is deterministic or randomized; thus, it also stands in contrast to accounts of free will that require quantum or Knightian forms of uncertainty, including the freebit picture of \citet{Aaronson2016TheGI}.
	
	On the other hand, since causal models are free to invoke Markovian randomness (e.g., sampling independently of the past), we also question \emph{superdeterminism} of the sort suggested by \citet{t2016cellular} in his cellular automaton interpretation of quantum mechanics. Superdeterminism proposes that our actions must somehow conspire to meet certain constraints on future outcomes. The trouble with conspiracies is that one must make them sufficiently powerful to support their proponent theories, while retaining the appearance of ``free'' Markovian dynamics \citep{wharton2020colloquium}.
	
	In unfinished work independent of our own, \citet{carroll2022causality} shares our goal of deriving Pearlean causality from a suitable Past Hypothesis. While it's too early for a technical comparison, we disagree with one of his premises. Carroll writes: \emph{``Increasing entropy is the sole reason for the asymmetry between past and future states in the macroscopic world.''} In \Cref{sec:intromemory}, we saw that increasing entropy does \emph{not} imply the data processing inequality. On the other hand, we will find that all of these asymmetries follow from a suitable causal model, which in turn emerges from a suitable pairing of dynamics with a Past Hypothesis.
	
	While entropy's monotonicity makes it convenient to study, \citet{bennett1988logical} (see also \citet[\S7.7]{li2019introduction}) argues that since both low-entropy (e.g., all zeros) and high-entropy (e.g., uniformly random) strings can be uninteresting, another quantity is needed to capture the ``interestingness'' of life's structures, such as its DNA. For Bennett, the interesting strings are those whose most plausible explanations involve long and difficult computations, such as Darwinian evolution on geological scales. Thus, he defines an object's \emph{logical depth} to be the time it takes to execute its shortest description. Stated differently, a logically deep object is one with low algorithmic entropy, but high time-bounded algorithmic entropy.
	
	\citet{aaronson2014quantifying} observe a variety of systems in which logical depth rises during early stages of dynamical evolution, and then falls during the approach to thermodynamic equilibrium. While their evidence is largely empirical, we can understand this phenomenon in terms of Bennett's slow growth law \citep{bennett1988logical}. Under Markovian dynamics, the law states that logical depth cannot (with substantial probability) rise quickly. If the system begins in a logically shallow state, then high logical depth can only arise from a long, slow computation \footnote{This argument parallels \Cref{sec:intromemory} where, instead of the slow-growth law for logical depth, the data processing inequality acted as a ``no-growth law'' for mutual information. There, we concluded that all correlations can be traced back in time to a common cause.}. Thus, logically deep objects cannot exist too early in the system's lifetime. They also cannot exist too late, when the system settles in a very shallow (e.g., uniform) equilibrium distribution. It follows that logically deep objects, if they appear at all, must form gradually at intermediate times.
	
	Quantum entanglement is stronger than classical correlation, in that it can produce a composite system with less entropy than either of its parts. For example, it's possible to produce local entropy by measuring a quantum system, while keeping the global entropy at zero; a global operation can then erase all traces of this measurement and its entropy. \citet{maccone2009quantum} argues that whenever entropy decreases for some local observer (let's name her Alice), any trace or memory must be erased so that she cannot recall the decrease. This argument is imprecise: while it's true that Alice can only retain as much information as the final entropy of her system, a reduction in her entropy need not necessarily eliminate evidence that the entropy used to be higher. Indeed, rather than preserve a high-entropy measurement in its entirety, Alice can simply summarize its entropy, perhaps rounded to a few significant figures. A summary has less entropy than the full measurement, so it can survive a reduction in entropy. If she later measures a lower entropy, she would notice the decrease.
	
	Finally, \citet{heinrich2016entropy} produce simulations as evidence that natural selection favors agents with knowledge of the past, over those with knowledge of the future. However, their arguments do not consider non-living memories (e.g., geological records), nor alternative scenarios in which knowledge of the future may be more advantageous. They also present no plausible mechanism for natural selection between forward-thinking and backward-thinking organisms.
	
	We should emphasize that, despite calling it the psychological arrow, the subject of this article is fundamentally physical (or more abstractly, computational), applying equally well to non-living records. Our physics-first approach has the advantage of applying not only to certain species of animals, but to \emph{all possible brains}, wherever the laws of physics hold. That being said, we do not treat the animal psychology or neuroscience of time perception. Being ill-positioned to do these literatures justice, we only point to the review of \citet{grondin2010timing} on the perception of short time intervals, and to the study by \citet{gauthier2019building} on perception of the more distant past and future.
	
	\section{Technical Preliminaries}
	\label{sec:prelim}
	
	\subsection{Notational Conventions}
	
	Throughout this article, $\mathbb R^+$ denotes the non-negative real numbers, $\mathbb Z^+,\mathbb Z^-$ the non-negative and non-positive integers, $\mathbb N:=\mathbb Z^+\setminus\{0\}$ the natural numbers, and $\mathbb Z_m:=\{0,1,\ldots,m-1\}$ denotes the first $m$ elements of $\mathbb Z^+$. Let $\oplus$ denote the bitwise XOR on $\mathbb Z^+$, e.g., $5\oplus 3 = 6$.
	
	Uppercase Latin letters $A,B,C$ denote generic sets or random variables, while lowercase letters denote elements of the corresponding sets $a\in A,\,b\in B$, or realizations of random variables. $\powerset A$ denotes $A$'s power set, consisting of all subsets of $A$. Script letters are used for $\sigma$-algebras $\mathcal F,\mathcal G$, as well as for three important sets that are treated as fixed in most contexts: the state space $\mathcal S$ and the discrete geometry $(\mathcal X,\mathcal T)$. The Greek letters $\mu,\nu,\pi,\Gamma$ denote measures.
	
	$B^A$ denotes the set of functions from a domain $A$ to a codomain $B$. Functions are also called indexed collections or, when $A$ is $\mathbb N$ or $\mathbb Z$, sequences. The notation $f:A\rightarrow B$ is synonymous with $f\in B^A$, and means that each $a\in A$ maps to $f(a)=f_a\in B$. The image of a set $A'\subset A$ under $f$ is $f(A') := \{f(a):a\in A'\}$. $\mathrm{Bij}(A)\subset A^A$ denotes the set of bijections from $A$ onto itself; these functions are synonymously called bijective, invertible, or reversible. The notation $x\mapsto f(x)$ is synonymous with $f$ itself, and is useful for quickly referring to ``anonymous'' functions without naming them. Given either $f:A\rightarrow B$ or $f:A\setminus\{a\}\rightarrow B$, the substitution $f_{a\leftarrow b}:A\rightarrow B$ equals $f$ everywhere except at $a$, where it equals $b$ instead. That is,
	
	\begin{align*}
		f_{a\leftarrow b}(a') &:= \begin{cases}
			b &\text{if }a'=a,
			\\f(a') &\text{if }a'\ne a.
		\end{cases}
	\end{align*}
	
	Some intuitive shorthands are used: for instance, $f_{A'}$ denotes the restriction of $f$ on the set $A'\subset A$. If $A$ is ordered, $f_{\le a}$ is shorthand for $f_{\{a'\in A:\; a'\le a\}}$. If $f$ takes multiple arguments, we allow partial application from left to right, so that, e.g., $((f_a)_b)_c := (f_{a,b})_c := f_{a,b,c}$.
	
	We want to make correct use of measure theory without overburdering ourselves with notation and unimportant details. We can usually avoid mentioning $\sigma$-algebras by assuming some defaults: for a countable set $A$, take it to be $\powerset A$; for subsets of $\mathbb R^d$, take the Borel $\sigma$-algebra; for a product of sets that are each equipped with $\sigma$-algebras, take it to be generated by the corresponding cylinder sets. On the implied $\sigma$-algebra of $A$, $\pmeas(A)$ denotes the set of probability measures, and $\meas(A)$ denotes the set of measures that are finite on singletons and not everywhere zero. Thus, $\pmeas(A) = \{\mu\in\meas(A): \mu(A)=1\}$. In a slight abuse of notation, we identify singleton sets with their element, so that $\mu(a) := \mu(\{a\})$. In this manner, when $A$ is countable, it's convenient to identify measures with real-valued functions of $A$:
	\begin{align*}
		\pmeas(A)&:=\{\mu\in(\mathbb R^+)^A: \sum_{a\in A}\mu(a) = 1\}
		\\\subset\meas(A)&:=\{\mu\in(\mathbb R^+)^A: \sum_{a\in A}\mu(a) > 0\}.
	\end{align*}
	
	A few measures are special enough to merit their own symbols. $\lambda^d\in\meas(\mathbb R^d)$ denotes the Lebesgue (i.e., volume) measure. For $a\in A$, $\delta_a\in\pmeas(A)$ denotes the point mass that puts probability one on $a$. $\sharp:=\sum_{a\in A}\delta_a$ denotes the counting measure, i.e., for $A'\subset A$, $\sharp(A') := |A'|$.
	
	At times, we describe a collection of conditional distributions consistent with a causal graph, thus uniquely determining the variables' joint distribution. In such cases, it's implied that all the variables live in a common probability space $(\Omega, \mathcal F, \prob)$. That is, the random variables are assumed to be defined on some sample space $\Omega$, equipped with the $\sigma$-algebra $\mathcal F\subset 2^\Omega$ generated by the variables, and a probability measure $\prob\in\pmeas(\Omega)$. When referring to a random variable $C:\Omega\rightarrow A$, we abuse notation and write $C\in A$. Depending on context, the expression $F(C)$ may refer to either: (1) $F(\prob_C)$, a function of $C$'s distribution; or (2) $\omega\mapsto F(\omega)(C(\omega))$, a random variable obtained by applying a (possibly random) function to $C$'s realization.
	
	We use some customary shorthands, such as $\prob(C=c):=\prob(\{\omega\in\Omega:\; C(\omega)=c\})$. The marginal and conditional distributions of $C$ are written $\prob_C,\prob_{C\mid E}\in\pmeas(A)$, respectively, for any event $E\in\mathcal F$. For example, $\prob_C(c) := \prob(C=c)$, and $\prob_{C\mid E}(c) := \prob(C=c\mid E)$. Indexed collections of random variables $\mathbf C = (C_a)_{a\in A}$ may be bolded for emphasis. Via the minor abuse of swapping argument order, the whole collection $\mathbf C$ can also be thought of as a random variable: $\mathbf C(\omega)_a := C_a(\omega)$.
	
	Finally, the total variation norm of a signed measure $\mu$ is denoted by
	\[\lVert \mu \rVert_\mathrm{TV} := \sup_{A\in\mathcal F} |\mu(A)|.\]
	
	It's used to express the total variation distance between probability measures $\mu,\nu$ as $\lVert \mu-\nu \rVert_\mathrm{TV}$.
	
	\subsection{The Kullback-Leibler Divergence}
	\label{sec:klinfo}
	
	All of our information-theoretic quantities are derived from this definition.
	
	\begin{definition}
		Let $\mathcal S$ be a countable set, $\mu\in\pmeas(\mathcal S),$ and $\pi\in\meas(\mathcal S)$. The \textbf{Kullback-Leibler divergence} of $\mu$ relative to $\pi$ is
		\[\kl\mu\pi
		:= \sum_{s\in\mathcal S}\mu(s)\log\frac{\mu(s)}{\pi(s)}.\]
	\end{definition}
	
	Following standard conventions, terms with $\mu(s)=0$ are treated as $0$, and terms with $\mu(s)>\pi(s)=0$ are treated as $+\infty$ \footnote{When the sum's positive and negative terms both diverge to infinity, the result is ill-defined. For definiteness, we may define $\kl\mu\pi$ to be either of $\pm\infty$ in this case.}. When the logarithm's base is $2$, the KL divergence is measured in units of \textbf{bits}; when the base is $e$, it's measured in \textbf{nats}.
	
	When $\pi$ is also a probability measure, it's customary to think of $\kl\mu\pi$ as a kind of ``distance'' between $\mu$ and $\pi$ \citep[\S2.3]{thomas2006elements}. However, in this article, a different interpretation is more useful: $\kl\mu\pi$ should be thought of as the \emph{precision} with which samples from $\mu$ are known, meaning, loosely speaking, the smallness of the $\pi$-measure of the ``typical support'' of $\mu$. For a concrete example, let $\mu$ be \textbf{$\pi$-uniform} on a \textbf{$\pi$-finite} support $A\subset\mathcal S$; that is, suppose $\pi(A)<\infty$, and let
	\[\mu(s):=
	\begin{cases}
		\pi(s)/\pi(A) &\text{if }s\in A,
		\\0 &\text{if }s\in\mathcal S\setminus A.
	\end{cases}\]
	
	Then, it's easily verified that $\kl\mu\pi = \log\frac 1{\pi(A)} = \log\frac 1{\sum_{s\in A}\pi(s)}$. In general, the \emph{least} precise distributions are those that are spread $\pi$-uniformly everywhere, whereas the \emph{most} precise are focused on a single element with small $\pi$-measure:
	\begin{lemma}
		\label{lem:klbounds}
		For fixed $\pi\in\meas(\mathcal S)$,
		\begin{align*}
			\inf_{\mu\in\pmeas(\mathcal S)}\kl\mu\pi
			&= \log\frac{1}{\sum_{s\in\mathcal S}\pi(s)},
			\\
			\sup_{\mu\in\pmeas(\mathcal S)}\kl\mu\pi
			&= \log\frac{1}{\inf_{s\in\mathcal S}\pi(s)}.
		\end{align*}
	\end{lemma}
	
	\begin{proof}
		First, we show that for all $\mu\in\pmeas(\mathcal S)$,
		\[\log\frac{1}{\sum_{s\in\mathcal S}\pi(s)}
		\le \kl\mu\pi
		\le \log\frac{1}{\inf_{s\in\mathcal S}\pi(s)}.\]
		
		If $\pi(\mathcal S)=\infty$, the left-hand inequality follows from the convention that $\log\frac 1\infty = -\infty$. Otherwise,
		\begin{align*}
			\kl\mu\pi
			&= \sum_{s\in\mathcal S}\mu(s)\log\mu(s) - \sum_{s\in\mathcal S}\mu(s)\log\pi(s)
			\\&\ge \sum_{s\in\mathcal S}\mu(s)\log\frac{\pi(s)}{\sum_{s'\in\mathcal S}\pi(s')} - \sum_{s\in\mathcal S}\mu(s)\log\pi(s) &\text{by Gibbs' inequality}
			\\&= \log\frac{1}{\sum_{s'\in\mathcal S}\pi(s')}.
		\end{align*}
		
		As for the right-hand inequality, since $0\le \mu(s)\le 1$ and $\pi(s)\ge\inf_{s'\in\mathcal S}\pi(s')$,
		\[\kl\mu\pi
		\le \sum_{s\in\mathcal S}\mu(s)\log\frac{1}{\inf_{s'\in\mathcal S}\pi(s')}
		= \log\frac{1}{\inf_{s'\in\mathcal S}\pi(s')}.\]
		
		Now, to actually achieve the infimum, enumerate $\mathcal S=\{s_1,s_2,\ldots\}$ and consider the $\pi$-uniform distribution on $\{s_1,\ldots,s_i\}$:
		\[\mu_i(s_j) :=
		\begin{cases}
			\pi(s_j)/\sum_{k=1}^i \pi(s_k) &\text{if }j\le i
			\\0 &\text{if }j>i.
		\end{cases}\]
		
		Then, as $i\rightarrow\infty$,
		\[\kl{\mu_i}\pi=\log\frac{1}{\sum_{k=1}^i\pi(s_k)} \rightarrow\log\frac{1}{\sum_{s\in\mathcal S}\pi(s)}.\]
		
		Finally, to achieve the supremum, choose a sequence $(s_i)_{i\in\mathbb N}$ such that $\lim_{i\rightarrow\infty}\pi(s_i)=\inf_{s\in\mathcal S}\pi(s)$. Then, as $i\rightarrow\infty$,
		\[\kl{\delta_{s_i}}\pi=\log\frac{1}{\pi(s_i)}\rightarrow\log\frac{1}{\inf_{s\in\mathcal S}\pi(s)}.\]
	\end{proof}
	
	When the bounds in \Cref{lem:klbounds} are finite, it's useful to compare them to $\kl\mu\pi$.
	
	\begin{definition}
		\label{def:entropy}
		Relative to a fixed $\pi\in\meas(\mathcal S)$, the \textbf{entropy} and \textbf{negentropy}, respectively, of any $\mu\in\pmeas(\mathcal S)$ are
		\begin{align}
			\label{eq:entropy}
			H_\pi(\mu)
			:= -\kl\mu{\frac{\pi}{\inf_{s\in\mathcal S}\pi(s)}}
			&= \log\frac 1{\inf_{s\in\mathcal S}\pi(s)} - \kl\mu\pi
			\ge 0,
			\\\label{eq:negentropy}
			J_\pi(\mu)
			:= \kl\mu{\frac{\pi}{\sum_{s\in\mathcal S}\pi(s)}}
			&=\kl\mu\pi - \log\frac 1{\sum_{s\in\mathcal S} \pi(s)}
			\ge 0.
		\end{align}
		
		If $\inf_{s\in\mathcal S}\pi(s)=0$ or $\sum_{s\in\mathcal S}\pi(s)=\infty$, we instead define $H_\pi(\mu) := \infty$ or $J_\pi(\mu) := \infty$, respectively, regardless of $\mu$.
	\end{definition}
	
	Note that the \textbf{information capacity} of $\pi$, generalized from \citet{frank2018physical} by
	\[\textbf{I}(\pi):=H_\pi(\mu) + J_\pi(\mu) = \log\frac 1{\inf_{s\in\mathcal S}\pi(s)} - \log\frac 1{\sum_{s\in\mathcal S} \pi(s)},\]
	
	does not depend on $\mu$. Thus, each of $H_\pi(\mu)$ or $J_\pi(\mu)$ is maximal precisely when the other is zero. $J_\pi(\mu) = 0$ iff $\mu$ is the $\pi$-uniform distribution $\pi/\pi(\mathcal S)$, and $H_\pi(\mu) = 0$ iff $\mu$ is a point mass $\delta_s$ satisfying $s\in\arg\min_{s'\in\mathcal S}\pi(s')$. As a special case, we recover the Shannon entropy and negentropy:
	\begin{align*}
		H(\mu)
		:= H_\sharp(\mu)
		= -\kl\mu\sharp
		&= \sum_{s\in\mathcal S}\mu(s)\log\frac{1}{\mu(s)},
		\\J(\mu)
		:= J_\sharp(\mu)
		= \log|\mathcal S| + \kl\mu\sharp
		&= \log|\mathcal S| - H(\mu).
	\end{align*}
	
	It may aid the reader's intuition to specialize the results of this article to the case $\pi=\sharp$. The additional generality enables analysis of dynamics with non-uniform stationary distributions. The next definition is useful when we discuss memory, and want to quantify how much two systems know about one another. In what follows, let $\mathcal S_X,\mathcal S_Y$ be countable sets.
	
	\begin{definition}
		\label{def:mutinf}
		The \textbf{mutual information} between a pair of random variables $X\in\mathcal S_X,\,Y\in\mathcal S_Y$ is
		
		\[I(X;\,Y) := \kl{\prob_{(X,Y)}}{\prob_X\times\prob_Y}
		= \sum_{x,y}\prob(X=x,Y=y)\log\frac{\prob(X=x,Y=y)}{\prob(X=x)\prob(Y=y)}.\]
	\end{definition}
	
	It's often convenient to identify a random variable with its distribution. For example, we may speak of the KL divergence of a random variable, perhaps conditioned on some event, as a shorthand for the KL divergence of its marginal or conditional probability measure. That is, $\kl X\pi := \kl{\prob_X}\pi$ and $\kl{X\mid Y=y}\pi := \kl{\prob_{X\mid Y=y}}\pi$. As a further shorthand, let $\kl{X\mid Y}{\pi_X} := \expect_y \kl{X\mid Y=y}{\pi_X}$, where the expectation is over realizations of $Y$, i.e., $\expect_yf(y):=\sum_y\prob(Y=y)f(y)$.
	
	\begin{lemma}
		\label{lem:mutinf}
		Let $\pi_X\in\meas(\mathcal S_X),\,\pi_Y\in\meas(\mathcal S_Y)$. For any random variables $X\in\mathcal S_X,\,Y\in\mathcal S_Y$,
		\begin{align*}
			\kl{X,Y}{\pi_X\times\pi_Y}
			&= \kl Y{\pi_Y} + \kl{X\mid Y}{\pi_X},
			\\\kl{X\mid Y}{\pi_X}
			&= \kl X{\pi_X} + I(X;\, Y).
		\end{align*}
	\end{lemma}
	
	\begin{proof}
		From their respective definitions,
		\begin{align*}
			\kl X{\pi_X}
			&= \expect_{x,y} \log\frac{\prob(X=x)}{\pi_X(x)},
			\\\kl Y{\pi_Y}
			&= \expect_{x,y} \log\frac{\prob(Y=y)}{\pi_Y(y)},
			\\\kl{X\mid Y}{\pi_X}
			&= \expect_{x,y} \log\frac{\prob(X=x\mid Y=y)}{\pi_X(x)},
			\\\kl{X,Y}{\pi_X\times\pi_Y}
			&= \expect_{x,y} \log\frac{\prob(X=x,Y=y)}{\pi_X(x)\pi_Y(y)},
			\\I(X;\,Y)
			&= \expect_{x,y} \log\frac{\prob(X=x,Y=y)}{\prob(X=x)\prob(Y=y)}.
		\end{align*}
		Therefore,
		\begin{align*}
			\kl{X,Y}{\pi_X\times\pi_Y}
			&= \expect_{x,y} \log\frac{\prob(X=x\mid Y=y)\prob(Y=y)}{\pi_X(x)\pi_Y(y)}
			\\&= \kl Y{\pi_Y} + \kl{X\mid Y}{\pi_X},\text{ and}
			\\\kl{X\mid Y}{\pi_X}
			&= \expect_{x,y} \log\frac{\prob(X=x,Y=y)}{\pi_X(x)\prob(Y=y)}
			\\&= \kl X{\pi_X} + I(X;\, Y).
		\end{align*}
	\end{proof}
	
	Our first corollary is a generalization of \Cref{eq:mutinfsimple} that expresses the joint negentropy as a sum of three non-negative contributions.
	
	\begin{corollary}
		\label{cor:abstractdecomp}
		Let $\pi_X\in\meas(\mathcal S_X),\,\pi_Y\in\meas(\mathcal S_Y)$. For any random variables $X\in\mathcal S_X,\,Y\in\mathcal S_Y$,
		\begin{align*}
			J_{\pi_X\times\pi_Y}(X,Y) = J_{\pi_X}(X) + J_{\pi_Y}(Y) + I(X;\, Y).
		\end{align*}
	\end{corollary}
	
	\begin{proof}
		Substitute the second part of \Cref{lem:mutinf} into the first, to get
		\begin{equation}
			\label{eq:kldecomp}
			\kl{X,Y}{\pi_X\times\pi_Y} = \kl X{\pi_X} + \kl Y{\pi_Y} + I(X;\, Y).
		\end{equation}
		
		Substituting \Cref{eq:negentropy} into \Cref{eq:kldecomp} yields the desired result.
	\end{proof}
	
	We close this section with another pair of bounds.
	
	\begin{corollary}
		\label{cor:klprod}
		Let $\pi_X\in\meas(\mathcal S_X),\,\pi_Y\in\meas(\mathcal S_Y)$. For any random variables $X\in\mathcal S_X,\,Y\in\mathcal S_Y$,
		\begin{align*}
			\log\frac{1}{\sum_{s\in\mathcal S_X}\pi_X(s)}
			\le \kl{X,Y}{\pi_X\times\pi_Y} - \kl Y{\pi_Y}
			\le \log\frac{1}{\inf_{s\in\mathcal S_X}\pi_X(s)}.
		\end{align*}
	\end{corollary}
	
	\begin{proof}
		By \Cref{lem:mutinf},
		\[\kl{X,Y}{\pi_X\times\pi_Y} - \kl Y{\pi_Y} = \kl{X\mid Y}{\pi_X} = \expect_y \kl{X\mid Y=y}{\pi_X}.\]
		
		The result now follows from \Cref{lem:klbounds}, because
		\begin{align*}
			\log\frac{1}{\sum_{s\in\mathcal S_X}\pi_X(s)}
			&\le \inf_y \kl{X\mid Y=y}{\pi_X}
			\\&\le \expect_y \kl{X\mid Y=y}{\pi_X}
			\\&\le \sup_y \kl{X\mid Y=y}{\pi_X}
			\le \log\frac{1}{\inf_{s\in\mathcal S_X}\pi_X(s)}.
		\end{align*}
	\end{proof}
	
	
	\section{Markov Chains}
	\label{sec:markov}
	
	In practical applications, when one wishes to model a system that evolves with time, without concern for microscopic details, it's common to use a very simple kind of causal model called a \textbf{Markov chain}. Its random variables are arranged in a linear sequence, like so:
	
	\begin{align*}
		C_0\rightarrow C_1&\rightarrow C_2\rightarrow C_3\rightarrow\cdots
	\end{align*}
	
	Recall that causal graphs encode two things: independence constraints and intervention rules. For a closed system, interventions play no role, so the graph only encodes the independences given by Pearl's $d$-separation criterion. For a Markov chain, these are nicely summarized by the \textbf{Markov property}, which can be stated in a symmetric fashion:
	
	\begin{align*}
		\text{For all }t\in\mathbb N,\; C_{>t} \text{ is independent of } C_{<t}, \text{ conditioned on } C_t.
	\end{align*}
	
	The Markov property is unchanged for the time-reversed Markov chain:
	
	\begin{align*}
		C_0\leftarrow C_1&\leftarrow C_2\leftarrow C_3\leftarrow\cdots
	\end{align*}
	
	Since the two graphs yield identical constraints on the joint probability distribution of $\mathbf C := (C_t)_{t\in\mathbb Z^+}$, there is no observation on $\mathbf C$ that can distinguish between the graphs. Thus, in the general case, Markov chains lack a well-defined causal direction.
	
	There exist graphs in which a causal direction \emph{can} be inferred from passive observation alone, without interventions: see \cite[\S2]{pearl2009causality} for a systematic treatment, or \Cref{sec:spca} for the relevant case of SPCAs. On the other hand, non-graphical constraints may also contribute an arrow of time. We'll soon see that non-stationary \emph{time-homogeneous} Markov chains cease to be time-homogeneous in the reverse direction. This particular asymmetry is responsible for the second law of thermodynamics; see \Cref{thm:secondmarkov}.
	
	This section has two main purposes. The first is to forge a bridge between Markov chains and deterministic dynamical systems, which allows useful properties to transfer from one to the other. We formally define three ways to specify a Markov chain's dynamics, the last of which provides our bridge: it expresses the sequence $\mathbf C$ as the macroscopic component of a deterministic system with random initialization.
	
	This section's second purpose is to better understand how the arrow of time arises in Markov chains. We learn how time-homogeneity breaks time-reversal symmetry. Then, using the bridge to deterministic dynamics, we trace the time-homogeneity back to microscopic variations in the initial conditions. Finally, by crossing the bridge to a \emph{reversible} deterministic system, running that system backward beyond the initial condition, and then crossing back into the world of Markov chains, our Markov chains acquire a natural bidirectional extension:
	
	\begin{align}
		\label{fig:bidirectional}
		\cdots\leftarrow C_{-2}\leftarrow C_{-1}\leftarrow C_0\rightarrow C_1&\rightarrow C_2\rightarrow\cdots
	\end{align}
	
	Throughout this section, let $\mathcal S$ be a fixed countable set. Its elements should be thought of as the fine macrostates discussed in \Cref{sec:hcontinuous}. $\mathcal S$ is our Markov chain's state space, which we eventually see as a discretization of some dynamical system's state space.
	
	When speaking of a collection of random variables, such as $\mathbf C$, we really only care about their joint distribution $\prob_{\mathbf C}\in\pmeas(\mathcal S^{\mathbb Z^+})$. Thus, we formally define a Markov chain to be any element of $\pmeas(\mathcal S^{\mathbb Z^+})$ that satisfies the Markov property. We say a Markov chain is homogeneous in time, or \textbf{time-homogeneous}, if its conditional distribution $\prob(C_{t+1}=c_{t+1} \mid C_t=c_t)$ does not depend on $t$.
	
	Let $\mathrm{Markov}(\mathcal S)\subset \pmeas(\mathcal S^{\mathbb Z^+})$ denote the set of time-homogeneous Markov chains. An element of $\mathrm{Markov}(\mathcal S)$ is uniquely determined by (1) its initial distribution $\prob_{C_0}$, and (2) its dynamics, i.e., the time-independent conditional distributions $\prob_{C_{t+1}\mid C_t=c_t}$. We always specify the initial distribution directly, but consider three different means of specifying the dynamics. Let's start with the direct approach.
	
	\subsection{Transition Matrix Presentation}
	\label{sec:markovmat}
	
	Since the conditional distributions $\prob_{C_{t+1}\mid C_t=c_t}$ are assumed not to depend on $t$, we need only specify one distribution for each possible preceding state $c_t\in\mathcal S$. In other words, the dynamics are given by an element of $\pmeas(\mathcal S)^{\mathcal S}$ or, equivalently, a $|\mathcal S|\times |\mathcal S|$ stochastic matrix.
	
	\begin{definition}
		\label{def:stochastic}
		Suppose $\pi\in\meas(\mathcal S)$. The set of \textbf{stochastic matrices}, \textbf{doubly stochastic matrices}, and \textbf{$\pi$-stochastic matrices} on $\mathcal S$ are, respectively,
		\begin{align*}
			\mathrm{SM}(\mathcal S)
			&:= \{p\in(\mathbb R^+)^{\mathcal S\times\mathcal S}:\; \forall s\in\mathcal S,\; \sum_{s'\in\mathcal S} p(s,s')=1\},
			\\\mathrm{DM}(\mathcal S)
			&:= \{p\in(\mathbb R^+)^{\mathcal S\times\mathcal S}:\; \forall s\in\mathcal S,\; \sum_{s'\in\mathcal S} p(s,s')=\sum_{s'\in\mathcal S} p(s',s)=1\},
			\\\mathrm{SM}_\pi(\mathcal S)
			&:= \{p\in\mathrm{SM}(\mathcal S):\; \forall s\in\mathcal S,\; \sum_{s'\in\mathcal S} \pi(s')p(s',s)=\pi(s)\}.
		\end{align*}
		Note that $\mathrm{SM}_\sharp(\mathcal S) = \mathrm{DM}(\mathcal S)$. $\pi$ is said to be \textbf{stationary} for $p$ if $p\in\mathrm{SM}_\pi(\mathcal S)$. For $m\in\mathbb N$, a measure $\mu$ or matrix $p$ is said to be \textbf{$m$-sided} if all its entries are multiples of $\frac 1m$; that is, if $\mu\in(\frac 1m\mathbb Z^+)^{\mathcal S}$ or $p\in(\frac 1m\mathbb Z^+)^{\mathcal S\times\mathcal S}$, respectively. If we don't want to specify $m$, we can also say it is \textbf{finite-sided}. It is said to be \textbf{strictly positive} if none of its entries are zero.
	\end{definition}
	
	\begin{definition}
		\label{def:markovm}
		The \textbf{transition matrix presentation} of time-homogeneous Markov chains on $\mathcal S$ is the surjective map
		
		\[\Phi_\mathrm{matrix}:\pmeas(\mathcal S) \times \mathrm{SM}(\mathcal S)
		\rightarrow \mathrm{Markov}(\mathcal S),\]
		
		defined as follows. Consider any \textbf{initial condition} $\mu\in\pmeas(\mathcal S)$ and \textbf{transition matrix} $p\in\mathrm{SM}(\mathcal S)$. Let the random variable $\mathbf C = (C_t)_{t\in\mathbb Z^+}$ be distributed such that, for all $c:\mathbb Z^+\rightarrow\mathcal S$,
		\[\prob(C_{\le t}=c_{\le t}) = \mu(c_0)\prod_{u=0}^{t-1} p(c_{u},\,c_{u+1})
		\qquad\forall t\in\mathbb Z^+.\]
		
		Or equivalently,
		\begin{align}
			\label{eq:markovm1}
			\prob(C_0=c_0) &= \mu(c_0),
			\\\label{eq:markovm2}
			\prob(C_{t+1}=c_{t+1} \mid C_{\le t}=c_{\le t}) &= p(c_t,\,c_{t+1}) \qquad\forall t\in\mathbb Z^+.
		\end{align}
		
		Then, $\Phi_\mathrm{matrix}(\mu,\,p) := \prob_{\mathbf C}.$
	\end{definition}
	
	Now, just as a Markov chain's causal graph can be reversed without affecting $\prob_{\mathbf C}$, one might wonder: can \Cref{eq:markovm2} be reversed, or does it represents a true time-reversal asymmetry? This is really a two-part question: (1) can $\prob(C_{t-1}=c_{t-1} \mid C_{\ge t}=c_{\ge t})$ be expressed as a function of only $c_t$ and $c_{t-1}$, independent of $t$; and (2) can $\mathbf C$ be extended to negative times, in such a way that \Cref{eq:markovm2} is satisfied there as well?
	
	In general, the answer to both questions is no. Extensions to negative times are not guaranteed to exist, nor be unique if they do: the simplest counterexamples are deterministic transitions that are non-surjective or non-injective, respectively. Even at positive times, the backward probabilities may be inhomogeneous. Let's derive them using the Markov property and Bayes' rule:
	
	\begin{align*}
		\prob(C_{t-1}=c_{t-1} \mid C_{\ge t}=c_{\ge t})
		&= \prob(C_{t-1}=c_{t-1} \mid C_t=c_t)
		\\&= \frac{\prob(C_{t-1}=c_{t-1}) \prob(C_t=c_t\mid C_{t-1}=c_{t-1})} {\prob(C_t=c_t)}
		\\&= \frac{\prob(C_{t-1}=c_{t-1})}{\prob(C_t=c_t)} p(c_{t-1}, c_t). \numberthis\label{eq:bayes}
	\end{align*}
	
	The analysis is simpler when $p$ is \textbf{recurrent}, meaning that from every $s\in\mathcal S$, we'll almost surely eventually revisit $s$. In that case, $p$ has a strictly positive stationary measure $\pi$ which, on each irreducible component of $\mathcal S$, is uniquely determined up to constant multiples \citep[Thm 17.48 and Rmk 17.51(i)]{klenke2020probability}. In particular, if $p(s',s)>0$, then $s',s$ are in the same irreducible component, and so the ratio $\pi(s')/\pi(s)$ is uniquely determined. Therefore, the \textbf{dual transition matrix}
	
	\begin{align}
		\label{eq:dual}
		p_\mathrm{dual}(s, s') := \frac{\pi(s')}{\pi(s)} p(s', s)
	\end{align}
	
	is well-defined, and it's easy to check that $p_\mathrm{dual}\in\mathrm{SM}_\pi(\mathcal S)$. If $p$ is not recurrent but has one or more stationary measures, the same expression defines the dual transition matrix with respect to a specific stationary measure $\pi$, provided that we delete the states with $\pi(s)=0$.
	
	Now, suppose the initial distribution $\mu$ is stationary for $p$. Then, $\prob_{C_t} = \mu$ for all $t$ so, wherever $\mu(c_t)>0$, the backward dynamics
	\begin{equation}
		\label{eq:markovmback}
		\prob(C_{t-1}=c_{t-1} \mid C_{\ge t}=c_{\ge t})
		= \frac{\mu(c_{t-1})}{\mu(c_t)} p(c_{t-1}, c_t)
		= p_\mathrm{dual}(c_t, c_{t-1})
	\end{equation}
	
	are indeed time-homogeneous. This also grants $\mathbf C$ a unique extension to negative times, satisfying both \Cref{eq:markovm2,eq:markovmback} for all $t\in\mathbb Z$.
	
	On the other hand, when $\mu$ is non-stationary, the backward probabilities in \Cref{eq:bayes} may differ from $p_\mathrm{dual}(c_t, c_{t-1})$. Typically, the ratio $\prob(C_{t-1}=c_{t-1})/\prob(C_t=c_t)$ is time-dependent, rendering the backward probabilities inhomogeneous in time.
	
	It's still possible to extend $\mathbf C$ in such a way that \Cref{eq:markovmback}, but not \Cref{eq:markovm2}, is satisfied at negative times. \Cref{sec:markovdual} justifies such an extension in terms of reversible microscopic dynamics. Picturing two arrows of time diverging from $t=0$ to $\pm\infty$, we find homogeneous dynamics ($p$ at positive times, $p_\mathrm{dual}$ at negative times) \emph{along} the arrow, and inhomogenous probabilities (inferred using Bayes' rule) \emph{against} the arrow. Thus, just as ``down'' always points toward Earth's center, the ``past'' always points toward the initial condition.
	
	Suppose $p\in\mathrm{SM}_\pi(\mathcal S)$. Intuitively, as $t$ increases, the state distribution $\prob_{C_t}$ moves ``closer'' to $\pi$. To make this precise, we define the thermo-majorization relation $\succeq_\pi$.
	
	\begin{definition}
		\label{def:majorize}
		Suppose $\mu,\nu\in\pmeas(\mathcal S)$ and $\pi\in\meas(\mathcal S)$. We say $\mu$ \textbf{thermo-majorizes} $\nu$ with respect to $\pi$, and write $\mu\succeq_\pi\nu$, if there exists $p\in\mathrm{SM}_\pi(\mathcal S)$ transforming $\mu$ into $\nu$, i.e.,
		
		\begin{align*}
			\sum_{s'\in\mathcal S}\mu(s')p(s',s)=\nu(s),
			\quad \sum_{s'\in\mathcal S}\pi(s')p(s',s)=\pi(s)
			\qquad\forall s\in\mathcal S.
		\end{align*}
	\end{definition}
	
	From the fact that $\mathrm{SM}_\pi(\mathcal S)$ is closed under matrix multiplication, it follows that $\succeq_\pi$ is a preorder. Progress along $\succeq_\pi$ can be quantified by the KL divergence:
	
	\begin{theorem}[Second law of thermodynamics]
		\label{thm:secondmarkov}
		Suppose $\mu,\nu\in\pmeas(\mathcal S)$, $\pi\in\meas(\mathcal S)$, and $\mu\succeq_\pi\nu$. Then,
		\[\kl\mu\pi \ge \kl\nu\pi.\]
	\end{theorem}
	
	A proof is provided by \citet[\S4.4]{thomas2006elements}, and a broad generalization by \citet[Thm 17]{brandao2015second}. For a Markov chain $\mathbf C$ with stationary measure $\pi$, it's immediate from \Cref{def:stochastic,def:majorize} that $C_t\succeq_\pi C_{t+1}$ for all $t\in\mathbb Z^+$. Therefore, \Cref{thm:secondmarkov} implies $\kl{C_t}\pi \ge \kl{C_{t+1}}\pi$. Using \Cref{def:entropy}, we recover the familiar statement that entropy increases monotonically:
	\[H_\pi(C_t) \le H_\pi(C_{t+1})\qquad\forall t\in\mathbb Z^+.\]
	
	\subsection{Random Function Presentation}
	\label{sec:markovrand}
	
	In physics, dynamics are not given as transition matrices, but as reversible \emph{equations of motion}. Integrating equations of motion for a fixed time interval yields their discrete-time analogue: invertible functions on $\mathcal S$. These functions should be deterministic; however, as a first step, we present Markov chains in terms of \emph{random} functions on $\mathcal S$.
	
	\begin{definition}
		\label{def:markovr}
		The \textbf{random function presentation} of time-homogeneous Markov chains on $\mathcal S$ is the map
		
		\[\Phi_\mathrm{random}:\pmeas(\mathcal S) \times \pmeas(\mathcal S^\mathcal S)
		\rightarrow \mathrm{Markov}(\mathcal S),\]
		
		defined as follows. Consider any \textbf{initial condition} $\mu\in\pmeas(\mathcal S)$ and \textbf{random dynamics} $\Gamma\in\pmeas(\mathcal S^\mathcal S)$. Let the random variables $(\mathbf C,\mathbf F) = (C_t,F_t)_{t\in\mathbb Z^+}$ be distributed such that
		
		\begin{align}
			\label{eq:markovr1}
			\prob_{(C_0,\mathbf F)} &= \mu\times\Gamma^{\mathbb Z^+},
			\\\label{eq:markovr2}
			C_{t+1} &= F_t(C_t) \qquad\forall t\in\mathbb Z^+.
		\end{align}
		
		Then, $\Phi_\mathrm{random}(\mu,\,\Gamma) := \prob_{\mathbf C}.$
	\end{definition}
	
	In other words, the random functions $(F_t)_{t\in\mathbb Z^+}$ are i.i.d. samples from a fixed distribution $\Gamma$. At each time $t\in\mathbb Z^+$, $C_{\le t}$ is independent of $F_{\ge t}$; applying $F_t$ to $C_t$ produces the next state $C_{t+1}$.
	
	For example, a random walk on $\mathcal S := \mathbb Z$ can be described as follows. Let $\Gamma$ be the uniform distribution on the multiset $\{\iota_{-1},\iota_0,\iota_0,\iota_1\}$, where $\iota_i:\mathbb Z\rightarrow\mathbb Z$ is defined by $\iota_i(s) := s+i$. Let the initial distribution $\mu\in\pmeas(\mathbb Z)$ be arbitrary. Then, $\Phi_\mathrm{random}(\mu,\,\Gamma)$ is a random walk that steps left or right on the number line, with $1/4$ probability each. The equivalent transition matrix is
	\[p(s,\,s+i) =
	\begin{cases}
		1/2 &\text{if }i=0,
		\\1/4 &\text{if }i=\pm 1,
		\\0 &\text{otherwise}.
	\end{cases}\]
	
	We extend the definition of $m$-sidedness to measures whose range is constrained to multiples of $1/m$. Thus, $\Gamma\in\pmeas(\mathcal S^\mathcal S)$ is $m$-sided iff it is uniform over some multiset of $m$ functions. In our example, both $\Gamma$ and $p$ are $4$-sided. The functions $\iota_i$ are bijective, whereas the matrix $p$ is doubly stochastic. These parallels are no coincidence: we now prove a general correspondence between the random function and transition matrix presentations.
	
	\begin{theorem}
		\label{thm:equivalence1}
		For every $\Gamma\in\pmeas(\mathcal S^\mathcal S)$, there is a \emph{unique} $p\in\mathrm{SM}(\mathcal S)$ such that, for all $\mu\in\pmeas(\mathcal S)$, $\Phi_\mathrm{matrix}(\mu, p) = \Phi_\mathrm{random}(\mu, \Gamma)$.
		If $\Gamma(\mathrm{Bij}(\mathcal S))=1$, then $p\in\mathrm{DM}(\mathcal S)$.
		If $\Gamma$ is $m$-sided, then so is $p$.
		
		Conversely, for every $p\in\mathrm{SM}(\mathcal S)$, there exists $\Gamma\in\pmeas(\mathcal S^\mathcal S)$ such that, for all $\mu\in\pmeas(\mathcal S)$, $\Phi_\mathrm{matrix}(\mu, p) = \Phi_\mathrm{random}(\mu, \Gamma)$.
		If $p\in\mathrm{DM}(\mathcal S)$ and/or $p$ is $m$-sided, then $\Gamma$ can be chosen to be supported on $\mathrm{Bij}(\mathcal S)$ and/or $m$-sided, respectively.
	\end{theorem}
	
	\begin{proof}
		Uniqueness of $p(s,s')$ follows from the fact that it can be read off a Markov chain, by initializing with $C_0=s$ and checking the probability that $C_1=s'$:
		\[p(s,\,s')=\Phi_\mathrm{matrix}(\delta_s,\,p)(C_1=s').\]
		
		To prove the first implication, fix $\Gamma\in\pmeas(\mathcal S^\mathcal S)$. For $s,s'\in\mathcal S$, define the sets
		\begin{align*}
			A(s,s') &:= \{f\in\mathcal S^\mathcal S:\; f(s)=s'\},
			\\\text{and let }p(s,s') &:= \Gamma(A(s,s')).
			\\\text{Then, for all }s\in\mathcal S,\;
			\sum_{s'\in\mathcal S}p(s,s')
			&= \Gamma\left(\bigcup_{s'\in\mathcal S} A(s,s')\right)
			= \Gamma(\mathcal S) = 1,
		\end{align*}
		
		because the sets $A(s,s')$, with fixed $s$, form a partition of $\mathcal S^\mathcal S$. Hence, $p\in\mathrm{SM}(\mathcal S)$. The roles of $s$ and $s'$ become symmetric when restricting to $\mathrm{Bij}(\mathcal S)$, so that the sets $A(s,s')\cap\mathrm{Bij}(\mathcal S)$, with either $s$ or $s'$ held fixed, form a partition of $\mathrm{Bij}(\mathcal S)$. Therefore, by the same argument, if $\Gamma(\mathrm{Bij}(\mathcal S))=1$, then $p\in\mathrm{DM}(\mathcal S)$. In addition, if $\Gamma$ is $m$-sided, then so is $p$ by its definition.
		
		Now, fix $\mu\in\pmeas(\mathcal S)$. Let $(\mathbf C,\mathbf F)$ be as in \Cref{def:markovr}, so that $\prob_{\mathbf C}=\Phi_\mathrm{random}(\mu, \Gamma)$. By definition, $\prob_{C_0}=\mu$, so \Cref{eq:markovm1} holds. Furthermore, since $F_t$ is independent of $C_{\le t}$,
		\begin{equation*}
			\prob(C_{t+1}=c_{t+1} \mid C_{\le t}=c_{\le t})
			= \prob(F_t(c_t) = c_{t+1})
			= \Gamma(A(c_t,c_{t+1}))
			= p(c_t,c_{t+1}),
		\end{equation*}
		
		so \Cref{eq:markovm2} holds as well. Therefore, $\Phi_\mathrm{matrix}(\mu, p) = \Phi_\mathrm{random}(\mu, \Gamma)$.
		
		To prove the converse, fix $p\in\mathrm{SM}(\mathcal S)$. Enumerate $\mathcal S=\{s_1,s_2,\ldots\}$. For each $r\in [0,1)$ and $s\in\mathcal S$, let $f_r(s) := s_i$, where $i\in\mathbb N$ is chosen such that $\sum_{j=1}^{i-1} p(s,s_j) \le r < \sum_{j=1}^i p(s,s_j)$. If we take $r$ to be drawn uniformly from $[0,1)$, then $f_r$ is drawn from the pushforward measure $\Gamma(B):=\lambda(\{r\in[0,1): f_r\in B\})$. For all $s,s_i\in\mathcal S$, it satisfies
		\[\Gamma(A(s,s_i))
		= \lambda(\{r\in [0,1):\, f_r(s)=s_i\})
		= \lambda\left(\left[\sum_{j=1}^{i-1} p(s,s_j),\; \sum_{j=1}^i p(s,s_j)\right)\right)
		= p(s,s_i).\]
		
		Just as in the previous case, this equation implies that $\Phi_\mathrm{matrix}(\cdot, p) = \Phi_\mathrm{random}(\cdot, \Gamma)$. In the case where $p$ is $m$-sided, by definition, $f_r$ = $f_{\lfloor rm\rfloor/m}$ for all $r\in [0,1)$. $\Gamma$ is then the uniform distribution over $\{f_0,f_{1/m},\ldots,f_{(m-1)/m}\}$, making it $m$-sided as well.
		
		If $p\in\mathrm{DM}(\mathcal S)$, special care is needed to construct $\Gamma$ using only bijections, satisfying $\Gamma(A(s,s')\cap\mathrm{Bij}(\mathcal S)) = p(s,s')$. The existence of such a $\Gamma$ is a generalization of the Birkhoff-von Neumann theorem; see the proof by \citet{revesz1962probabilistic}. An easier construction is possible under the additional hypothesis that $p$ is $m$-sided: consider the bipartite graph with vertex partition $(\mathcal S, \mathcal S)$ and, for each $s,s'\in\mathcal S$, $m\cdot p(s,s')$ edges from $s$ in the left partition to $s'$ in the right. Since $p\in\mathrm{DM}(\mathcal S)$, the graph is $m$-regular, meaning that each vertex is incident to exactly $m$ edges. Repeated application of Hall's marriage theorem partitions the edges into $m$ perfect matchings, i.e., bijections on $\mathcal S$. We take $\Gamma$ to be the uniform distribution over these $m$ bijections.
	\end{proof}
	
	\subsection{Deterministic Function Presentation}
	\label{sec:markovdet}
	
	To complete our bridge into the theory of dynamical systems, we must derandomize the transition function. Consider again our $4$-sided random walk on $\mathbb Z$, whose dynamics sample uniformly from the multiset of transition functions $\{\iota_{-1},\iota_0,\iota_0,\iota_1\}\subset\mathrm{Bij}(\mathbb Z)$. We can simulate one step of the walk by rolling a fair $4$-sided die. Denoting the roll at time $t$ by $R_{t,0}\in\mathbb Z_4$, $C_{t+1}$ becomes a deterministic function of the pair $(C_t,R_{t,0})$. Since a deterministic function cannot add new rolls to the state, we should attach an infinite supply of rolls to the initial state $C_0$.
	
	The idea, then, is to start with an extended state $(C_0,R_0)$, which includes the bi-infinite sequence $R_0=(R_{0,i})_{i\in\mathbb Z}$ of pre-rolled dice. Define a ``macroscopic'' transition function $T:\mathbb Z\times\mathbb Z_4\rightarrow\mathbb Z\times\mathbb Z_4$ by, for all $s\in\mathbb Z$,
	\begin{align*}
		T(s,0) &:= (s-1,\,0),
		\\T(s,1) &:= (s,\,1),
		\\T(s,2) &:= (s,\,2),
		\\T(s,3) &:= (s+1,\,3).
	\end{align*}
	
	Note that $T$ is bijective. The full dynamics consist of two interleaved operations on the expanded state space $\mathbb Z\times(\mathbb Z_4)^\mathbb Z$. First, $T$ is applied using only the central (i.e., $0$'th) die:
	\[(C_{t+1},R'_{t,0}) := T(C_t,R_{t,0}).\]
	
	Note that $R'_{t,0}=R_{t,0}$ in our example. Then, the $4$-symbol shift map $\sigma$ is applied to the sequence of dice, centering $R_{t,1}$ in preparation for the next application of $T$:
	\begin{align*}
		R_t
		= (\ldots,R_{t,-1},&R_{t,0},R_{t,1},\ldots),
		\\R_{t+1}
		= \sigma((R_t)_{0\leftarrow R'_{t,0}})
		= \sigma(\ldots,R_{t,-1},R'_{t,0},R_{t,1},\ldots)
		= (\ldots,R'_{t,0},&R_{t,1},R_{t,2},\ldots).
	\end{align*}
	
	In a more compact notation, at time $t+1$ the sequence becomes
	\[R_{t+1,i} =
	\begin{cases}
		R'_{t,0} &\text{if }i=-1,
		\\R_{t,i+1} &\text{if }i\ne -1.
	\end{cases}\]
	
	By induction, for all $t,u,i,j\in\mathbb Z^+$ satisfying $t+i=u+j$, we have $R_{t,i} = R_{u,j}$ . In particular, $R_{t,\mathbb Z^+}=R_{0,t+\mathbb Z^+}$; therefore, if the initial dice $R_{0,\mathbb Z}$ are i.i.d., then the unused dice $R_{t,\mathbb Z^+}$ will continue to be i.i.d. at all later times $t\in\mathbb Z^+$. As a result, $C_t$'s evolution is time-homogeneous in the direction of increasing $t$. In contrast, if we were to follow $C_t$'s evolution in the direction of \emph{decreasing} $t$, toward zero, we would encounter used dice. These dice correlated with $C_t$, in such a way as to return the process to its initial configuration $C_0$.
	
	Let's get some geometric intuition, using the correspondence $\Psi$ from \Cref{eq:bakershift}. It identifies the sequence $R_t$ with a point in the unit square; thus, $(C_t,R_t)\in\mathbb Z\times(\mathbb Z_4)^\mathbb Z$ is identified with a point in $\mathbb Z\times[0,1)\times[0,1)\simeq\mathbb R\times[0,1)$. The resulting system is essentially the multibaker chain, first introduced and visualized by \citet{gaspard1992diffusion}. Thanks to $\Psi$, we can think of the sequence
	\[C_t,\;R_{t,0},\;R_{t,\pm 1},\;R_{t,\pm 2},\cdots\]
	
	as the base-$4$ digits, in order of decreasing significance, of a continuous state's coordinates. Although the maps $T,\sigma$ are deterministic and reversible, their result on $C_t$ appears stochastic if we ignore $R_t$. $\sigma$ is chaotic, gradually amplifying formerly insignificant figures in the real-valued coordinates of $R_t$. $T$ uses these as a source of randomness to evolve $C_t$. Since real numbers have infinitely many digits, they can provide an endless supply of entropy by repeating $\sigma$.
	
	More generally, we can view any time-homogeneous Markov chain $\mathbf C$ as the macroscopic part of some chaotic dynamical system $(\mathbf C,\mathbf R)$. The initial microscopic variables $(R_{0,i})_{i\in\mathbb Z}$ are i.i.d. samples from a probability space $(\mathcal R,\mathcal G,\Gamma)$. Adapting \Cref{eq:shift} to symbols in $\mathcal R$, the shift map $\sigma:\mathcal R^\mathbb Z\rightarrow\mathcal R^\mathbb Z$ is again defined by $\sigma(r)_t := r_{t+1}$. Given a macroscopic transition function $T:\mathcal S\times\mathcal R\rightarrow\mathcal S\times\mathcal R$, the full dynamics $\overline T:\mathcal S\times\mathcal R^\mathbb Z\rightarrow\mathcal S\times\mathcal R^\mathbb Z$ are given by
	\begin{align}
		\label{eq:tshift}
		\overline T(s,r):=(s',\sigma(r_{0\leftarrow r'})),\text{ where }(s',r'):=T(s,r_0).
	\end{align}
	
	In other words, $\overline T$ first applies $T$ to $(s,r_0)$, and then applies $\sigma$ to the entire sequence $r$. In the following definition, let $\mathrm{Dyn}(\mathcal S)$ denote the class of all tuples of the form $(\mathcal R,\mathcal G,\Gamma,T)$, where $(\mathcal R,\mathcal G,\Gamma)$ is a probability space, and $T:\mathcal S\times\mathcal R\rightarrow\mathcal S\times\mathcal R$ is $(\powerset{\mathcal S}\times\mathcal G)$-measurable.
	
	\begin{definition}
		\label{def:markovd}
		The \textbf{deterministic function presentation} of time-homogeneous Markov chains on $\mathcal S$ is the map
		
		\[\Phi_\mathrm{determ}:\pmeas(\mathcal S) \times \mathrm{Dyn}(\mathcal S)
		\rightarrow \mathrm{Markov}(\mathcal S),\]
		
		defined as follows. Consider any \textbf{initial condition} $\mu\in\pmeas(\mathcal S)$ and \textbf{deterministic dynamics} $(\mathcal R,\mathcal G,\Gamma,T)\in\mathrm{Dyn}(\mathcal S)$. Let the random variables $(\mathbf C,\mathbf R) = (C_t,R_t)_{t\in\mathbb Z^+}$ be distributed such that
		
		\begin{align}
			\label{eq:markovd1}
			\prob_{(C_0,R_0)} &= \mu\times\Gamma^\mathbb Z,
			\\\label{eq:markovd2}
			(C_{t+1},R_{t+1}) &= \overline T(C_t,R_t) \qquad\forall t\in\mathbb Z^+,
		\end{align}
		
		with $\overline T$ defined by \Cref{eq:tshift}. Then, $\Phi_\mathrm{determ}(\mu,\mathcal R,\mathcal G,\Gamma,T) := \prob_{\mathbf C}.$
	\end{definition}
	
	We proceed to this section's main result, establishing a general correspondence between the deterministic function and transition matrix presentations. In order to continue thinking of the microscopic variables $\mathbf R$ as base-$m$ digits of real-valued phase coordinates, we prefer the probability space $(\mathcal R,\mathcal G,\Gamma)$ to be the \textbf{fair $m$-sided die} $(\mathbb Z_m, \powerset{\mathbb Z_m}, \frac 1m\sharp)$. A necessary and sufficient condition for the existence of such a presentation, is that the transition matrix $p$ be $m$-sided.
	
	As an analogue of Liouville's theorem, we also prefer $\overline T$ to be reversible and $(\pi\times\Gamma^\mathbb Z)$-measure-preserving, for some $\pi\in\meas(\mathcal S)$. By \Cref{eq:tshift}, these properties hold iff $T$ is reversible and $(\pi\times\Gamma)$-measure-preserving.
	
	\begin{theorem}
		\label{thm:equivalence2}
		For every $(\mathcal R,\mathcal G,\Gamma,T)\in\mathrm{Dyn}(\mathcal S)$, there is a \emph{unique} $p\in\mathrm{SM}(\mathcal S)$ such that, for all $\mu\in\pmeas(\mathcal S)$, $\Phi_\mathrm{matrix}(\mu, p) = \Phi_\mathrm{determ}(\mu,\mathcal R,\mathcal G,\Gamma,T)$.
		If $T$ is $(\pi\times\Gamma)$-measure-preserving, then $p\in\mathrm{SM}_\pi(\mathcal S)$. If $\Gamma$ is $m$-sided, then so is $p$.
		
		Conversely, for every $p\in\mathrm{SM}(\mathcal S)$, there exists $(\mathcal R,\mathcal G,\Gamma,T)\in\mathrm{Dyn}(\mathcal S)$ such that, for all $\mu\in\pmeas(\mathcal S)$, $\Phi_\mathrm{matrix}(\mu, p) = \Phi_\mathrm{determ}(\mu,\mathcal R,\mathcal G,\Gamma,T)$.
		If $p\in\mathrm{SM}_\pi(\mathcal S)$, with $\pi$ strictly positive and finite-sided, then $T$ can be chosen to be bijective and $(\pi\times\Gamma)$-measure-preserving. If $p$ is $m$-sided, then $(\mathcal R,\mathcal G,\Gamma)$ can be chosen to be the fair $m$-sided die, with $T$ simultaneously being bijective if $p\in\mathrm{DM}(\mathcal S)$.
	\end{theorem}
	
	The proof is based on that of \Cref{thm:equivalence1}, and is detailed in \Cref{sec:appendixproof}.
	
	\Cref{thm:equivalence2} expresses every time-homogeneous Markov chain as the macroscopic component of a deterministic system, whose dynamics are given by \Cref{eq:tshift}. Conversely, in every dynamical system of this form, the macroscopic component is a time-homogeneous Markov chain. While the dynamics of real physical systems may take other forms, many of them are likewise believed to closely approximate time-homogeneous Markov chains \citep{werndl2009deterministic}.
	
	For ease of analysis, we base our modeling on \Cref{eq:tshift}. In order to build further intuition on time and reversibility, the following subsections demonstrate some basic applications. More sophisticated applications, which require the theory of \Cref{sec:spca}, are postponed until \Cref{sec:discussion}.
	
	\subsection{The Past Hypothesis}
	\label{sec:pasthypothesis}
	
	Consider \citeauthor{albert2001time}'s \citep{albert2001time} proposal that the arrow of time requires a \textbf{Past Hypothesis}, i.e., a special initial condition. Within our model, we can investigate exactly which initial conditions are suitable. \Cref{thm:secondmarkov} says that the entropy of a Markov chain, with respect to a fixed stationary distribution $\pi$, cannot decrease. If we want entropy to actually increase rather than stay fixed, it's necessary for $\mu$ to have less than maximal entropy. However, sub-maximal entropy alone is insufficient: if the underlying dynamics have a CPT-like symmetry, any configuration that exhibits sub-maximal and increasing entropy, can be converted into one in which entropy is sub-maximal but \emph{decreasing}.
	
	A closer look at the initial distribution $\bar\mu=\mu\times\Gamma^\mathbb Z$ from \Cref{eq:markovd1} reveals that its microscopic component is assumed to have the special form $\Gamma^\mathbb Z$. In the case of the fair $m$-sided die, we saw that $\Psi$ maps this to the uniform distribution on $\mathcal S_\mathrm{baker}$. This agrees with \citeauthor{albert2001time}'s Mentaculus, in which the starting distribution is assumed to be uniform over some macrostate.
	
	In fact, we can relax the form of the initial distribution $\bar\mu$: it's sufficient that it be \textbf{absolutely continuous} with respect to $\sharp\times\Gamma^\mathbb Z$, meaning that it can be written $d\bar\mu = f\,d(\sharp\times\Gamma^\mathbb Z)$. Here, $f:\mathcal S\times\mathcal R^\mathbb Z\rightarrow\mathbb R^+$ is called the Radon-Nikodym derivative, or \textbf{probability density function}, of $\bar\mu$ with respect to $\sharp\times\Gamma^\mathbb Z$. Explicitly, this means that for all $A\in\powerset{\mathcal S}\times\mathcal G^\mathbb Z$,
	\begin{equation}
		\label{eq:generalinit}
		\bar\mu(A)
		:= \int_A f\,d(\sharp\times\Gamma^\mathbb Z)
		= \sum_{s\in\mathcal S}\int_{\{r\in\mathcal R^\mathbb Z:\,(s,r)\in A\}} f(s,r)\,d\Gamma^\mathbb Z(r).
	\end{equation}
	
	When $\bar\mu$ has this form, to an increasingly good approximation as $t$ gets large, $\mathbf C$ behaves like a time-homogeneous Markov chain. To prove it, we first need a technical lemma. In measure theory, a \textbf{semiring of sets} is a collection of sets $\mathcal A$, such that $\emptyset\in\mathcal A$ and, for all $A,B\in\mathcal A$, $A\cap B$ is in $\mathcal A$ and $A\setminus B$ is a finite disjoint union of elements of $\mathcal A$. We show that probability density functions can be approximated by finite linear combinations of characteristic functions on a semiring.
	
	\begin{lemma}
		\label{lem:semiring}
		Let $(\Omega,\mathcal F,\mu)$ be a $\sigma$-finite measure space, let $\mathcal A$ be a semiring of sets generating $\mathcal F$, and let $f:\Omega\rightarrow\mathbb R^+$ be a probability density with respect to $\mu$, i.e., $\int_\Omega f \,d\mu = 1$. Then, for all $\varepsilon>0$, there exist $k\in\mathbb Z^+$ and $(\alpha_i,A_i)_{i=1}^k$, with $\alpha_i>0$ and the sets $A_i\in\mathcal A$ mutually disjoint, such that $g:=\sum_{i=1}^k\alpha_i\mathbf 1_{A_i}$ is a probability density with respect to $\mu$, satisfying
		\begin{align}
			\label{eq:semiring}
			\int_\Omega\left|f - g\right| \,d\mu &< \varepsilon.
		\end{align}
	\end{lemma}
	
	\begin{proof}
		By definition of the Lebesgue integral, $f$ can be approximated by a finite linear combination of characteristic functions of sets in $\mathcal F$. By the measure approximation theorem \citep[Thm 1.65(ii)]{klenke2020probability}, these characteristic functions can themselves be approximated by finite linear combinations of characteristic functions of sets in $\mathcal A$; putting them together yields \Cref{eq:semiring}. Since semirings are closed under set difference and intersection, the sets $A_i$ can be made disjoint. Furthermore, \Cref{eq:semiring} implies that
		\[
		\left|1 - \int_\Omega g\,d\mu\right|
		= \left|\int_\Omega (f-g)\,d\mu\right|
		\le \int_\Omega |f-g|\,d\mu
		< \varepsilon,
		\]
		
		so $g$ is almost a probability density. Normalizing $g$ increases the error in \Cref{eq:semiring} to at most $2\varepsilon$; since $\varepsilon$ was arbitrary, the proof is finished.
	\end{proof}
	
	\begin{theorem}
		\label{thm:markovapprox}
		Suppose $(\mathcal R,\mathcal G,\Gamma,T)\in\mathrm{Dyn}(\mathcal S)$, and consider the \textbf{generalized initial condition} $\bar\mu\in\pmeas(\mathcal S\times\mathcal R^\mathbb Z)$, defined by $d\bar\mu = f\,d(\sharp\times\Gamma^\mathbb Z)$ for some probability density function $f:\mathcal S\times\mathcal R^\mathbb Z\rightarrow\mathbb R^+$. Let the random variables $(\mathbf C,\mathbf R) = (C_t,R_t)_{t\in\mathbb Z^+}$ be distributed such that
		
		\begin{align}
			\label{eq:markovapprox1}
			\prob_{(C_0,R_0)} &= \bar\mu,
			\\\label{eq:markovapprox2}
			(C_{t+1},R_{t+1}) &= \overline T(C_t,R_t) \qquad\forall t\in\mathbb Z^+,
		\end{align}
		
		with $\overline T$ defined by \Cref{eq:tshift}. Then,
		\[\lim_{t\rightarrow\infty}\lVert \prob_{C_{\ge t}} - \Phi_\mathrm{determ}(\prob_{C_t},\mathcal R,\mathcal G,\Gamma,T) \rVert_\mathrm{TV} = 0,\]
		
		and the convergence rate is uniform in the choice of transition function $T$.
	\end{theorem}
	
	\begin{proof}
		Let $\mathcal A$ consist of all subsets of $\mathcal S\times\mathcal R^\mathbb Z$ that are of the form
		\begin{equation}
			\label{eq:cylinder}
			B' \times \prod_{i\in\mathbb Z} B_i,
		\end{equation}
		
		where $B'\subset\mathcal S$, $B_i\in\mathcal G$, and, for all but finitely many $i\in\mathbb Z$, $B_i=\mathcal R$. $\mathcal A$ is the semiring of cylinder sets generating the $\sigma$-algebra $\powerset{\mathcal S}\times \mathcal G^\mathbb Z$.
		
		Fix the approximation error $\varepsilon>0$. By \Cref{lem:semiring}, there exist positive coefficients $\alpha_1,\ldots,\alpha_k$ and disjoint sets $A_1,\ldots,A_k\in\mathcal A$, such that the probability measure $d\bar\nu = g\,d(\sharp\times\Gamma^\mathbb Z)$, with density function $g := \sum_{j=1}^k \alpha_j\mathbf 1_{A_j}$, satisfies
		\[\lVert \bar\mu - \bar\nu \rVert_\mathrm{TV}
		= \frac 12\int_{\mathcal S\times\mathcal R^\mathbb Z}\left|f - g\right| \,d(\sharp\times\Gamma^\mathbb Z) < \frac\varepsilon 2.\]
		
		Now, expressing each of the $A_j\in\mathcal A$ as in \Cref{eq:cylinder}, let $t$ be the maximum of $1+|i|$, over all $i\in\mathbb Z$ for which at least one of the $A_j$ has $B_i\ne\mathcal R$ (or $0$ if there is no such $i$). Then, for each $j=1,\ldots,k$, we can write $A_j = B''_j\times \mathcal R^{\{i\in\mathbb Z:\, |i|\ge t\}}$, where $B''_j\in\powerset{\mathcal S}\times\mathcal G^{\{i\in\mathbb Z:\, |i|< t\}}$. As a result, under the distribution $\mathbb P^g$ defined by
		\begin{align*}
			\prob^g_{(C_0,R_0)} &= \bar\nu,
			\\(C_{t+1},R_{t+1}) &= \overline T(C_t,R_t) \qquad\forall t\in\mathbb Z^+,
		\end{align*}
		
		the ``dice'' $(R_{0,i})_{|i|\ge t}$ are i.i.d. with marginals $\Gamma$, independent of the other initial variables. Therefore, starting at time $t$, $\mathbb P^g$ becomes a time-homogeneous Markov chain, satisfying
		\[\prob^g_{C_{\ge t}} = \Phi_\mathrm{determ}(\prob^g_{C_t},\mathcal R,\mathcal G,\Gamma,T).\]
		
		Since the initial variables of $\Phi_\mathrm{determ}(\mu,\mathcal R,\mathcal G,\Gamma,T)$ are distributed according to $\mu$, and its subsequent variables evolve independently of $\mu$, it follows that for all $\mu,\nu\in\pmeas(\mathcal S)$,
		\[\lVert \Phi_\mathrm{determ}(\mu,\mathcal R,\mathcal G,\Gamma,T) - \Phi_\mathrm{determ}(\nu,\mathcal R,\mathcal G,\Gamma,T) \rVert_\mathrm{TV}
		= \lVert \mu - \nu \rVert_\mathrm{TV}.\]
		
		Putting these equations together yields
		\[\lVert \prob^g_{C_{\ge t}} - \Phi_\mathrm{determ}(\prob_{C_t},\mathcal R,\mathcal G,\Gamma,T) \rVert_\mathrm{TV}
		= \lVert \prob^g_{C_t} - \prob_{C_t} \rVert_\mathrm{TV},\]
		
		which we plug into the triangle inequality to get
		\begin{align*}
			\lVert \prob_{C_{\ge t}} - \Phi_\mathrm{determ}(\prob_{C_t},\mathcal R,\mathcal G,\Gamma,T) \rVert_\mathrm{TV}
			&\le \lVert \prob_{C_{\ge t}} - \prob^g_{C_{\ge t}} \rVert_\mathrm{TV}
			+ \lVert \prob^g_{C_{\ge t}} - \Phi_\mathrm{determ}(\prob_{C_t},\mathcal R,\mathcal G,\Gamma,T) \rVert_\mathrm{TV}
			\\&= \lVert \prob_{C_{\ge t}} - \prob^g_{C_{\ge t}} \rVert_\mathrm{TV}
			+ \lVert \prob_{C_t} - \prob^g_{C_t} \rVert_\mathrm{TV}
			\\&\le 2\lVert \prob_{(C_0,R_0)} - \prob^g_{(C_0,R_0)} \rVert_\mathrm{TV}
			\\&= 2\lVert \bar\mu-\bar\nu \rVert_\mathrm{TV}
			\\&< \varepsilon,
		\end{align*}
		
		where the middle inequality holds because \Cref{eq:markovapprox2} makes $\mathbf C$ a deterministic function of $(C_0,R_0)$. Since $\varepsilon$ did not depend on $T$, taking $\varepsilon\rightarrow 0$ proves the uniform convergence.
	\end{proof}
	
	For some intuition regarding the proof of \Cref{thm:markovapprox}, consider the case where $(\mathcal R,\mathcal G,\Gamma)$ is the fair $m$-sided die. Using $\Psi$, identify $\mathcal R^\mathbb Z$ with the unit square, so that $\mathcal G^\mathbb Z,\Gamma^\mathbb Z$ become its Borel $\sigma$-algebra and Lebesgue (i.e., area) measure, respectively. The cylinder sets $A_j\in\mathcal A$ become finite unions of squares whose vertices have $m$-adic rational coordinates. The absolutely continuous distribution $\bar\mu$ is ``spread out'' in such a way that it appears uniform at scales smaller than the $A_j$'s. After enough iterations of stretching by the baker's map, the $x$ coordinate becomes uniform on large scales as well. $T$ only looks at a coarse description of the $x$ coordinate; thus, for large $t$, it yields essentially Markovian behavior.
	
	Of course, this result would be much less interesting if processes only became Markovian after attaining heat death. Fortunately, the convergence is uniform in $T$. Therefore, for fixed $\bar\mu$, we can always choose a $T$ that mixes sufficiently slowly, so that the process becomes Markovian long before its entropy reaches a maximum.
	
	In summary, \Cref{eq:generalinit} provides our general Past Hypothesis. While it allows for messy behavior in some neighborhood of $t=0$, a stable arrow of time is sure to emerge eventually for $t\gg 0$. Of course, if $T$ is invertible, then the argument applies equally well in reverse: a stable arrow of time must also emerge in the opposite direction for $t\ll 0$. By \Cref{thm:markovapprox}, the forward dynamics are ultimately determined by $(\mathcal R,\mathcal G,\Gamma,T)$, while the backward dynamics are ultimately determined by $(\mathcal R,\mathcal G,\Gamma,T^{-1})$.  By \Cref{thm:equivalence2}, each of these corresponds to some transition matrix; to close out this section, we investigate how the forward and backward matrices are related.
	
	\subsection{Dual Arrows}
	\label{sec:markovdual}
	
	Fix $\pi\in\meas(\mathcal S)$ and $(\mathcal R,\mathcal G,\Gamma,T)\in\mathrm{Dyn}(\mathcal S)$, with $T$ bijective and $(\pi\times\Gamma)$-measure-preserving. Then, $\overline T$ is a deterministic, reversible, chaotic map that acts on the microstate $(C_t,R_t)$. The macroscopic transition probabilities $\prob(C_{t+1}=s'\mid C_t=s)$ depend on how $R_t$ is distributed. When $R_t$ is ``uniform'' (i.e., its distribution is $\Gamma^\mathbb Z$) and independent of $C_t$, these probabilities are given by the unique transition matrix $p\in\mathrm{SM}_\pi(\mathcal S)$ that \Cref{thm:equivalence2} assigns to $(\mathcal R,\mathcal G,\Gamma,T)$. By \Cref{thm:secondmarkov}, it follows that $H_\pi(C_t) \le H_\pi(C_{t+1})$.
	
	In \Cref{def:markovd}, $R_0$ meets this condition, but $R_t$ at later times tends to become correlated with $C_t$. Fortunately, since $T$ thereafter only encounters the terms $R_{t,\mathbb Z^+}$, which \emph{are} uniform and independent, the forward dynamics continue to abide by $p$. Thus, $R_t$ has the remarkable property that it appears well-mixed for the purposes of forward evolution, but not for the purposes of backward evolution. Indeed, using \Cref{thm:equivalence2} to pass to the transition matrix presentation, \Cref{eq:bayes} shows that the backward probabilities are inhomogeneous.
	
	As an intuition-building (albeit unrealistic) thought experiment, we can imagine perturbing an arbitrary microstate $(C_t,R_t)$, perhaps corresponding to the present state of the Universe. That is, without changing $C_t$, we replace $R_t$ with a random sample from $\Gamma^\mathbb Z$ \footnote{By \Cref{thm:markovapprox}, more general perturbations would yield similar results.}. Since $\overline T$ is chaotic, the evolution of $C_t$ is highly sensitive to this microscopic perturbation; nonetheless, the perturbed forward evolution would still appear \emph{plausible}. To be mathematically precise, the forward transition probabilities would be unchanged; the new trajectory simply resamples from the same matrix $p$ as the unperturbed trajectory.
	
	On the other hand, the perturbed backward evolution would appear truly incredible. The perturbation destroys any correlations in $R_t$ that would otherwise conspire to restore the past's low entropy. Reversible phenomena, such as celestial orbits, would rewind to their former configurations relatively unharmed. However, thermodynamically irreversible events would fail to become undone: a shattered vase would no longer enjoy a past in which it was whole, nor would a fried egg return to its raw origins (perhaps it would overcook on the pan instead). We do not presume to guess how bizarrely life on Earth would be afflicted by such a reversal in the arrow of time.
	
	However, the new trajectory can be understood in basic terms. Making $R_t$ uniform restores a symmetry between $t$'s past and future: applying \Cref{thm:equivalence2} to $(\mathcal R,\mathcal G,\Gamma,T)$ and $(\mathcal R,\mathcal G,\Gamma,T^{-1})$, respectively, we find that the forward and backward dynamics are both homogeneous in time, corresponding to two different transition matrices.
	
	We now show that these matrices are dual to each other. Recall that, for fixed $\pi\in\meas(\mathcal S)$, the dual of a stochastic matrix $p\in\mathrm{SM}_\pi(\mathcal S)$ is defined by \Cref{eq:dual}. If $p$ is recurrent, then it uniquely determines $\pi$, hence also $p_\mathrm{dual}$; otherwise, the choice of $\pi$ must be clarified by context.
	
	By introducing one additional piece of notation, we can also reverse the random function presentation. For $\Gamma\in\pmeas(\mathrm{Bij}(\mathcal S))$, define $\Gamma_\mathrm{dual}\in\pmeas(\mathrm{Bij}(\mathcal S))$ by $\Gamma_\mathrm{dual}(A) := \Gamma(\{f\in\mathrm{Bij}(\mathcal S):\, f^{-1}\in A\})$. In other words, $\Gamma_\mathrm{dual}$ samples a bijection from $\Gamma$ and inverts it.
	
	\begin{theorem}
		\label{thm:dual}
		Fix $\mu\in\pmeas(\mathcal S)$. Suppose $p\in\mathrm{DM}(\mathcal S)$ and $\Gamma\in\pmeas(\mathrm{Bij}(\mathcal S))$.
		If $\Phi_\mathrm{matrix}(\mu, p) = \Phi_\mathrm{random}(\mu, \Gamma)$,
		then $\Phi_\mathrm{matrix}(\mu, p_\mathrm{dual}) = \Phi_\mathrm{random}(\mu, \Gamma_\mathrm{dual})$.
		
		Alternatively, suppose $\pi\in\meas(\mathcal S)$ is strictly positive, $p\in\mathrm{SM}_\pi(\mathcal S)$, and $(\mathcal R,\mathcal G,\Gamma,T)\in\mathrm{Dyn}(\mathcal S)$, with $T$ bijective and $(\pi\times\Gamma)$-measure-preserving.
		If $\Phi_\mathrm{matrix}(\mu, p) = \Phi_\mathrm{determ}(\mu, \mathcal R,\mathcal G,\Gamma,T)$,
		then $\Phi_\mathrm{matrix}(\mu, p_\mathrm{dual}) = \Phi_\mathrm{determ}(\mu, \mathcal R,\mathcal G,\Gamma,T^{-1})$.
	\end{theorem}
	
	\begin{proof}
		Since $\mathrm{DM}(\mathcal S)=\mathrm{SM}_\sharp(\mathcal S)$, the implied choice of stationary measure in the first part is $\sharp$; hence, the dual matrix in \Cref{eq:dual} is simply the transpose of $p$. Let $(\mathbf C,\mathbf F)$ be as in \Cref{def:markovr}, so that $\prob_{\mathbf C}=\Phi_\mathrm{random}(\mu, \Gamma_\mathrm{dual})$. Then, since $F_t$ is independent of $C_{\le t}$, 
		\begin{align*}
			\prob(C_{t+1}=c_{t+1} \mid C_{\le t}=c_{\le t})
			&= \prob(F_t(c_t) = c_{t+1})
			\\&= \Gamma_\mathrm{dual}(\{f\in\mathrm{Bij}(\mathcal S):\, f(c_t)=c_{t+1}\})
			\\&= \Gamma(\{f\in\mathrm{Bij}(\mathcal S):\, f(c_{t+1})=c_t\})
			\\&= p(c_{t+1},c_t)
			\\&= p_\mathrm{dual}(c_t,c_{t+1}).
		\end{align*}
		
		Similarly, under the second set of hypotheses, let $(\mathbf C,\mathbf R)$ be as in \Cref{def:markovd}, so that $\prob_{\mathbf C}=\Phi_\mathrm{determ}(\mu,\mathcal R,\mathcal G,\Gamma,T^{-1})$. Since $R_{t,0} = R_{0,t}$ is independent of $C_{\le t}$,
		\begin{align*}
			\pi(c_t)\prob(C_{t+1}=c_{t+1} \mid C_{\le t}=c_{\le t})
			&= \pi(c_t)\prob(\exists a'\in\mathcal R,\, T^{-1}(c_t,R_{0,t}) = (c_{t+1},a'))
			\\&= \pi(c_t)\prob(\exists a'\in\mathcal R,\, T(c_{t+1},a') = (c_t,R_{0,t}))
			\\&= (\pi\times\Gamma)(\{c_t\}\times\{a\in\mathcal R:\, \exists a'\in\mathcal R, T(c_{t+1},a') = (c_t,a)\})
			\\&= (\pi\times\Gamma)((\{c_t\}\times\mathcal R)\cap T(\{c_{t+1}\}\times\mathcal R))
			\\&= (\pi\times\Gamma)((\{c_{t+1}\}\times\mathcal R)\cap T^{-1}(\{c_t\}\times\mathcal R))
			\\&= (\pi\times\Gamma)(\{c_{t+1}\}\times\{a\in\mathcal R:\, \exists a'\in\mathcal R,\, T(c_{t+1},a) = (c_t,a')\})
			\\&= \pi(c_{t+1}) p(c_{t+1},c_t).
		\end{align*}
		
		Therefore,
		\[\prob(C_{t+1}=c_{t+1} \mid C_{\le t}=c_{\le t})
		= \frac{\pi(c_{t+1})}{\pi(c_t)}p(c_{t+1},c_t) = p_\mathrm{dual}(c_t,c_{t+1}).\]
	\end{proof}
	
	In light of \Cref{thm:dual}, it's natural to extend \Cref{eq:markovm1,eq:markovm2} with the dual \Cref{eq:markovmback}. That is, given $\mu\in\pmeas(\mathcal S)$ and a recurrent matrix $p\in\mathrm{SM}(\mathcal S)$, the \textbf{bidirectional Markov chain} $\mathbf C = (C_t)_{t\in\mathbb Z}$ has its joint distribution uniquely determined by:
	
	\begin{align*}
		\prob(C_0=c_0) &= \mu(c_0),
		\\\prob(C_{t+1}=c_{t+1} \mid C_{\le t}=c_{\le t}) &= p(c_t,c_{t+1}) \qquad\quad\,\,\forall t\in\mathbb Z^+,
		\\\prob(C_{t-1}=c_{t-1} \mid C_{\ge t}=c_{\ge t}) &= p_\mathrm{dual}(c_t,c_{t-1}) \qquad\forall t\in\mathbb Z^-.
	\end{align*}
	
	$\mathbf C$ is effectively two Markov chains glued together at $t=0$, and its causal graph is \Cref{fig:bidirectional}. When the functions in question are bijective, the random function and deterministic function presentations have obvious extensions to negative times. \Cref{thm:dual} then extends the correspondences of \Cref{thm:equivalence1,thm:equivalence2} to bidirectional Markov chains. These might prove helpful for modeling cosmologies in which the Big Bang is replaced by a temporally bidirectional Big Bounce.
	
	\section{Stochastic Partitioned Cellular Automata}
	\label{sec:spca}
	
	As a final extension, we supplement the time coordinate with a space coordinate. The initial condition and dynamics now jointly determine the state (or its distribution) at every point in \textbf{spacetime}. We say the dynamics are \textbf{local} if they consist of independent transitions at all points in spacetime; we say they are homogeneous in space and time, or \textbf{spacetime-homogeneous}, if the transition matrix is identical at all points. We seek a model whose forward transitions are both local and spacetime-homogeneous.
	
	In the physical Universe, local interactions propagate no faster than the speed of light, ensuring a form of causal separation between spacelike separated events. When viewed in reverse, however, macroscopic events seem to conspire in a non-local manner: for example, the fragments of a broken vase rise up simultaneously to reassemble. There is no macroscopic indication that the fragments are about to rise because, once at rest, the fragments have lost the \emph{memory} of their common time of origin. In other words, the fragments decorrelate from each other when they come to rest.
	
	In this section, we identify a new time-reversal asymmetry, the \textbf{Memory Law}, which permits spontaneous decorrelation but forbids spontaneous correlation. It complements the \textbf{Resource Law}, which is a very general and local version of the second law of thermodynamics.
	
	\subsection{Core Definitions}
	
	We begin by defining our discrete spatial structure. It consists of a set of positions $\mathcal X$ linked by a neighbor relation; in other words, a graph. We want $\mathcal X$ to be homogeneous, meaning that all its elements should have neighborhoods that ``look the same''. To that end, we color the edges, and require that each $x\in\mathcal X$ have one outgoing and one incoming edge of each color $\tau\in\mathcal T$. Thus, each color corresponds to a bijection on $\mathcal X$. By coupling the set of positions $\mathcal X$ with the set of times $\frac 12\mathbb Z^+ = \{0,\frac 12,1,\frac 32,\ldots\}$, we obtain the spacetime $\frac 12\mathbb Z^+\times\mathcal X$. The reason for including half-integer times will become

	clear in \Cref{def:spca}. Let's formalize what we have so far.
	
	\begin{definition}
		\label{def:geometry}
		A homogeneous discrete \textbf{spatial geometry} $(\mathcal X, \mathcal T)$ is a set of spatial positions or \textbf{cells} $\mathcal X$, along with a finite set of local translations or \textbf{tracks} $\mathcal T\subset\mathrm{Bij}(\mathcal X)$. We assume an identity track: $\mathrm{id}\in\mathcal T$. Associated with $(\mathcal X, \mathcal T)$ is the quasimetric
		\[d(x,y) := \min\{n\in\mathbb Z^+:\; \exists \tau_1,\ldots\tau_n\in\mathcal T,\; \tau_n(\ldots(\tau_1(x))\ldots)=y\},\]
		
		with $d(x,y):=\infty$ if no such path exists. $d$ is a metric iff $\mathcal T$ is closed under inverses. It determines a partial order on $\frac 12\mathbb Z^+\times\mathcal X$: for times $t,u\in\frac 12\mathbb Z^+$ and positions $x,y\in\mathcal X$, we write $(t,x) < (u,y)$ iff $t< u$ and $d(x,y)\le \lceil u \rceil - \lceil t \rceil$. In this case, we say $(u,y)$ is \textbf{reachable} from, or in the \textbf{future} of, $(t,x)$, and that $(t,x)$ is in the \textbf{past} of $(u,y)$.
		
		Also associated with $(\mathcal X, \mathcal T)$ is the causal graph $\mathrm G(\mathcal X,\mathcal T)$ on the vertex set $\{O\}\cup\left(\frac 12\mathbb Z^+\times\mathcal X\right)$. It has a directed edge from $(t,x)$ to $(u,y)$ iff $(t,x) < (u,y)$ and $u=t+\frac 12$. To allow for dependences in the initial state, there is also an edge from $O$ to $(0,x)$ for all $x\in\mathcal X$.
		
		If either $(t,x) < (u,y)$, $(t,x)=(u,y)$, or $(t,x) > (u,y)$, we say the pairs are \textbf{timelike} related; otherwise, they are \textbf{spacelike} related. In analogy with light cones in relativity, for any subset of spacetime $A\subset\frac 12\mathbb Z^+\times\mathcal X$, let
		\begin{align*}
			\mathrm{Past}(A) &:= \bigcup_{(t,x)\in A} \left\{(u,y)\in\frac 12\mathbb Z^+\times\mathcal X:\;
			(u,y) < (t,x)\right\},
			\\\mathrm{Future}(A) &:= \bigcup_{(t,x)\in A} \left\{(u,y)\in\frac 12\mathbb Z^+\times\mathcal X:\;
			(u,y) > (t,x)\right\},
			\\\mathrm{NonFut}(A) &:= \left(\frac 12\mathbb Z^+\times\mathcal X\right)\setminus\mathrm{Future}(A).
		\end{align*}
	\end{definition}
	
	It's easy to check that $\mathrm G(\mathcal X,\mathcal T)$ has a path from $(t,x)$ to $(u,y)$ iff $(t,x)<(u,y)$. One example of a spatial geometry is the infinite grid $\mathcal X=\mathbb Z^d$, equipped with $2d+1$ tracks: the identity, and the unit displacements. The finite grid $\mathcal X=(\mathbb Z_m)^d$ can be made into a spatial geometry using a periodic variant of the same tracks, connecting at opposite boundaries.
	
	After choosing a spatial geometry, the next step is to define a dynamical model over it. We have a few options here, with different tradeoffs. For example, if we wish to study how nature breaks the symmetry between space and time, then the dynamics should not distinguish between the space and time axes. For completeness, we include such a model in \Cref{sec:appendixspacetime}. There, we find that a spacelike initial condition suffices to break the symmetry, leading to a macroscopic description in which causality is strictly timelike. To keep this section simple, however, we take the special status of the time axis for granted, and seek only to break its symmetry between past and future.
	
	Perhaps the most straightforward model is the ordinary \textbf{cellular automaton (CA)}, whose dynamics compute the state at coordinate $(t,x)\in\mathbb Z^+\times\mathcal X$, as a function of the previous state of all neighboring cells $\{(t-1,\,\tau^{-1}(x)):\; \tau\in\mathcal T\}$. For our purposes, the CA has two major drawbacks: (1) it combines neighborhood aggregation and dynamical evolution into a single step, making analysis more difficult; and (2) the question of whether a given CA is reversible is undecidable and, even when the answer is affirmative, a CA's forward and backward dynamics may require drastically different neighborhood sizes \citep{kari1994reversibility}.
	
	To solve these issues, \citet{morita1989computation} invented the \textbf{partitioned cellular automaton (PCA)}. Both CAs and PCAs are divided into cells $\mathcal X$, but the cells in a PCA are further subdivided along tracks $\mathcal T$. Thus, each pair $(x,\tau)\in\mathcal X\times\mathcal T$ is the coordinate of a \textbf{subcell}. The dynamics of a PCA interleave two distinct steps. In a \textbf{trackwise translation step}, the contents of each subcell $(x,\tau)$ move to $(\tau(x), \tau)$. Then, in a \textbf{cellwise evolution step},  each cell $x$ evolves independently of all the others. Thanks to this separation, a PCA is reversible iff its cellwise evolution rule is. See \citet{kari2018reversible} for a modern survey of reversible cellular automata, including PCAs.
	
	Just as with the Markov chains in the previous section, we will find macroscopic views of certain PCAs to have homogeneous transition probabilities; we call such a view a \textbf{stochastic partitioned cellular automaton (SPCA)}. For the remainder of this section, fix a spatial geometry $(\mathcal X,\mathcal T)$. Additionally, for each track $\tau\in\mathcal T$, fix a countable set $\mathcal S_\tau$ of fine macrostates. Subcells on track $\tau$ take states in $\mathcal S_\tau$; therefore, a whole cell's state belongs to $\mathcal S := \prod_{\tau\in\mathcal T}\mathcal S_\tau$.
	
	An SPCA's \textbf{configuration history} is a collection of random variables $\mathbf C = (C_{t,x,\tau})_{(t,x,\tau)\in\frac 12\mathbb Z^+\times\mathcal X\times\mathcal T}$, such that $C_{t,x,\tau}\in\mathcal S_\tau$ for all $t,x,\tau$. To interpret $\mathrm G(\mathcal X,\mathcal T)$ as a graphical model on $\mathbf C$, replace each vertex $(t,x)$ with the corresponding cell $C_{t,x}$. Thus, $\mathbf C$ is Markov relative to $\mathrm G(\mathcal X,\mathcal T)$ iff the conditional distribution of each vertex $C_{t,x}$, given all its nondescendants $C_{\mathrm{NonFut}(t,x)\setminus\{(t,x)\}}$, depends only on its immediate parents. In this case, the joint distribution is uniquely determined as the product of the conditional distributions. We say that $\mathbf C$ is spacetime-homogeneous if the conditional distributions do not depend on $t$ or $x$.
	
	Let $\mathrm{SPCA}(\mathcal X,\mathcal T,\mathcal S)\subset \pmeas(\mathcal S^{\frac 12\mathbb Z^+\times\mathcal X\times\mathcal T})$ denote the set of spacetime-homogeneous distributions that are Markov relative to $\mathrm G(\mathcal X,\mathcal T)$ and satisfy the trackwise translation
	\begin{align}
		\label{eq:cellulart}
		C_{t+\frac 12,\tau(x),\tau} = C_{t, x, \tau} \qquad\forall (t,x,\tau)\in\mathbb Z^+\times\mathcal X\times\mathcal T.
	\end{align}
	
	What data are needed to uniquely determine an element of $\mathrm{SPCA}(\mathcal X,\mathcal T,\mathcal S)$? Since \Cref{eq:cellulart} determines the behavior at half-integer times, it only remains to specify the behavior at whole-integer times. By homogeneity, and the fact that $(t,x)\in\mathbb N\times\mathcal X$ in $\mathrm G(\mathcal X,\mathcal T)$ has $(t-\frac 12,x)$ as its sole parent, this amounts to choosing a pairing of initial condition and cellwise evolution. In analogy with \Cref{def:markovm,def:markovr,def:markovd}, the cellwise evolution comes in three presentations.
	
	\begin{definition}
		\label{def:spca}
		The transition matrix, random function, and deterministic function presentations of SPCAs on $(\mathcal X,\mathcal T,\mathcal S)$ are, respectively, the surjective maps
		
		\begin{align*}
			\Phi_\mathrm{cell\text-matrix}:\pmeas(\mathcal S^\mathcal X) \times \mathrm{SM}(\mathcal S)
			\rightarrow \mathrm{SPCA}(\mathcal X,\mathcal T,\mathcal S),
			\\\Phi_\mathrm{cell\text-random}:\pmeas(\mathcal S^\mathcal X) \times \pmeas(\mathcal S^\mathcal S)
			\rightarrow \mathrm{SPCA}(\mathcal X,\mathcal T,\mathcal S),
			\\\Phi_\mathrm{cell\text-determ}:\pmeas(\mathcal S^\mathcal X) \times \mathrm{Dyn}(\mathcal S)
			\rightarrow \mathrm{SPCA}(\mathcal X,\mathcal T,\mathcal S),
		\end{align*}
		
		defined as follows. Fix an initial condition $\mu\in\pmeas(\mathcal S^\mathcal X)$.
		
		Given $p\in\mathrm{SM}(\mathcal S)$, let $\mathbf C$ be a configuration history satisfying \Cref{eq:cellulart} and, for all $c:\frac 12\mathbb Z^+\times\mathcal X\rightarrow\mathcal S$,
		\begin{align}
			\label{eq:cellularm1}
			\prob_{C_0} &= \mu,
			\\\label{eq:cellularm2}
			\prob(C_{t,x} = c_{t,x} \mid C_{\mathrm{NonFut}(t-\frac 12,x)} = c_{\mathrm{NonFut}(t-\frac 12,x)})
			&= p(c_{t-\frac 12,x},\, c_{t,x})
			\qquad\forall (t,x)\in\mathbb N\times\mathcal X.
		\end{align}
		
		Then, $\Phi_\mathrm{cell\text-matrix}(\mu, p) := \prob_{\mathbf C}$.
		
		Similarly, given $\Gamma\in\pmeas(\mathcal S^\mathcal S)$, let the random variables $(\mathbf C,\mathbf F) = (C_{t,x},C_{t+\frac 12,x},F_{t,x})_{(t,x)\in\mathbb Z^+\times\mathcal X}$ satisfy \Cref{eq:cellulart}, $F_{t,x}\in\mathcal S^\mathcal S$ for all $(t,x)\in\mathbb Z^+\times\mathcal X$, and
		\begin{align}
			\label{eq:cellularr1}
			\prob_{(C_0,\mathbf F)} &= \mu\times\Gamma^{\mathbb Z^+\times\mathcal X},
			\\\label{eq:cellularr2}
			C_{t,x} &= F_{t-1,x}(C_{t-\frac 12, x})
			\qquad\forall (t,x)\in\mathbb N\times\mathcal X.
		\end{align}
		
		Then, $\Phi_\mathrm{cell\text-random}(\mu, \Gamma) := \prob_{\mathbf C}$.
		
		Finally, given $(\mathcal R,\mathcal G,\Gamma,T)\in\mathrm{Dyn}(\mathcal S)$, let the random variables $(\mathbf C,\mathbf R) = (C_{t,x},R_{t,x})_{(t,x)\in\frac 12\mathbb Z^+\times\mathcal X}$ satisfy \Cref{eq:cellulart}, $R_{t,x}\in\mathcal R^\mathbb Z$ for all $(t,x)\in\frac 12\mathbb Z^+\times\mathcal X$, and
		\begin{align}
			\label{eq:cellulard1}
			\prob_{(C_0,R_0)} &= \mu\times\Gamma^{\mathcal X\times\mathbb Z},
			\\\label{eq:cellulardt}
			R_{t+\frac 12,x} &= \sigma(R_{t,x})
			&\forall (t,x)\in\mathbb Z^+\times\mathcal X,
			\\\label{eq:cellulard2}
			(C_{t,x},R_{t,x}) &= \overline T(C_{t-\frac 12, x},\,R_{t-\frac 12, x})
			&\forall (t,x)\in\mathbb N\times\mathcal X,
		\end{align}
		where $T(s,r_0)=(s',r')\implies\overline T(s,r) := (s',r_{0\leftarrow r'})$. Then, $\Phi_\mathrm{cell\text-determ}(\mu,\mathcal R,\mathcal G,\Gamma,T) := \prob_{\mathbf C}$.
	\end{definition}
	
	Markov chains and SPCAs mutually generalize each other. On one hand, every Markov chain can be thought of as an SPCA with the trivial geometry: $|\mathcal X|=1$ and $\mathcal T=\{\mathrm{id}\}$. On the other hand, every SPCA with $|\mathcal X|<\infty$ can be thought of as a Markov chain, by collapsing all its cells and tracks into a \textbf{global state} $C_t$, taking values in $\mathcal S_\mathrm{global}:=\mathcal S^\mathcal X$. 
	
	The correspondences of \Cref{thm:equivalence1,thm:equivalence2} generalize to the SPCA setting verbatim, replacing each $\Phi$ in their statements with the corresponding $\Phi_\mathrm{cell\text-}$. To see this, note that the initial condition and trackwise translation steps are shared among all three presentations. While the cellwise evolution steps have three separate definitions, each one applies independently at all points $(t,x)$ in spacetime. Thus, consider the evolution at each point in isolation: these are identical to the Markov chain transitions defined in \Cref{eq:markovm2,eq:markovr2,eq:markovd2}, respectively.
	
	Therefore, every SPCA $\mathbf C$ can be understood as the macroscopic component of some deterministic PCA $(\mathbf C,\mathbf R)$, whose identity track has the state space $\mathcal S_\mathrm{id}\times\mathcal R^\mathbb Z$. Having established this equivalence, we need no longer worry about the microscopic details. From here onward, it's convenient to focus on the transition matrix presentation.
	
	There is just one complication: the dynamical law $p\in\mathrm{SM}_\pi(\mathcal S)$, with $\pi\in\meas(\mathcal S)$, is insufficient to guarantee that the global distribution $\pi^\mathcal X\in\meas(\mathcal S^\mathcal X)$ is stationary for an SPCA. While it's true that $p$ preserves the distribution $\pi$ in each cell, so that cellwise evolution steps preserve $\pi^\mathcal X$, the trackwise translation steps need not share this property. For example, if $\pi$ has perfectly correlated subcells, then a translation would mix adjacent cells, likely breaking the correlation. To remedy this, we restrict attention to measures that factorize into their subcellular marginals:
	\begin{align}
		\label{eq:cellulars}
		\meas(\mathcal T,\mathcal S) := \left\{ \prod_{\tau\in\mathcal T}\pi_\tau:\; \pi_\tau\in\meas(\mathcal S_\tau)\text{ for each }\tau\in\mathcal T \right\}.
	\end{align}
	
	If we add the condition that $\pi\in\meas(\mathcal T,\mathcal S)$, then $\pi^\mathcal X$ distributes all subcells independently and trackwise identically, so it's clear that trackwise translations also preserve $\pi^\mathcal X$. $\pi^\mathcal X$ is truly stationary for the SPCA, making it a suitable reference measure for the entropy-like quantities of \Cref{sec:klinfo}; in particular, applying \Cref{thm:secondmarkov} to an SPCA's global state yields a global second law of thermodynamics. It has some unfortunate limitations: it only applies to finite-sized SPCAs, and says little about the entropy within specific regions of interest. In order to develop a local version, we need some notation dealing with spatial regions.
	
	A \textbf{region} is a finite set of cells $A\subset\mathcal X$; its \textbf{volume} is the set cardinality $|A|$. When $\pi\in\meas(\mathcal T,\mathcal S)$, the expressions $\kl{C_{t,A}}{\pi^A}$, $H_{\pi^A}(C_{t,A})$, and $J_{\pi^A}(C_{t,A})$ are invariant to permutations of the subcells within any track $A\times\{\tau\}$. In these expressions, as well as in the relation $\succeq_{\pi^A}$, we may write $\pi$ instead of $\pi^A$; this notational abuse presents no ambiguity.
	
	A \textbf{system} (or \textbf{time-varying region}) is a collection of regions $(A_t)_{t\in\frac 12\mathbb Z^+}$ held fixed during cellwise evolutions, meaning that $A_t=A_{t-\frac 12}$ for all $t\in\mathbb N$. The pair $(t,A_t)$ may be abbreviated by $(t,A)$ so that, for example,
	\[C_{t,A} := C_{t,A_t} = (C_{t,x,\tau})_{(x,\tau)\in A_t\times\mathcal T}.\]
	
	Let $\mathcal T_+(A)$ denote the set of cells reachable from $A\subset\mathcal X$ in one local translation, and $T_-(A)$ denote the set of cells that can be reached \emph{only} from $A$ in one local translation: formally,
	
	\begin{align*}
		\mathcal T_+(A) &:= \{\tau(x): x\in A,\tau\in\mathcal T\},
		\\\mathcal T_-(A) &:= \mathcal X\setminus\mathcal T_+(\mathcal X\setminus A).
	\end{align*}
	
	Since $\mathrm{id}\in\mathcal T$, it follows that $\mathcal T_-(A) \subset A\subset\mathcal T_+(A)$. Additionally, iterating $T_+$ yields the reachable set at arbitrary future times, so that
	\begin{align*}
		\mathrm{Future}(\{t\}\times A)
		&= \bigcup_{u>t} \left(\{u\}\times\mathcal T_+^{\lceil u \rceil - \lceil t \rceil}(A)\right).
	\end{align*}
	
	\subsection{The Laws of Information Dynamics}
	\label{sec:laws}
	
	Conservation laws, for physical quantities such as energy and charge, have local statements in the form of \textbf{continuity equations}. These equate the decrease (or increase) in a system's quantity, with the net ``flux'' of the quantity exiting (or entering) the system. Entropy is known to increase spontaneously, so it's not a conserved quantity; thus, while a system's entropy must increase by \emph{at least} the net incoming flux, it may well increase more. As a result, entropy should satisfy a \textbf{continuity inequality} instead of an equation.
	
	More generally, we derive continuity inequalities for two types of information-theoretic quantity: (1) generic information, measured by the KL divergence, and (2) mutual information about specific random variables. Since these quantities are non-additive, their flux is defined relative to information already contained in the system.
	
	\begin{definition}
		Fix a measure $\pi\in\meas(\mathcal T,\mathcal S)$, an SPCA $\mathbf C$, and a random variable $Y$. The \textbf{net flux} from system $A$ at time $t\in\mathbb Z^+$, relative to $\pi$ or $Y$, respectively, is
		\begin{align*}
			F_\pi(t,A) &:= \kl{C_{t,A}}\pi - \kl{C_{t+\frac 12,A}}\pi,
			\\F_Y(t,A) &:= I(C_{t,A};\,Y) - I(C_{t+\frac 12,A};\,Y).
		\end{align*}
	\end{definition}
	
	By definition, flux is a change in either $D_\mathrm{KL}$ or $I$, over the time interval $[t,\,t+\frac 12]$. Since $t$ is a whole integer, this time interval is a trackwise translation step (see \Cref{eq:cellulart}): it moves, but does not modify, the SPCA's information content. Thus, the fluxes should be thought of as quantifying the net movement of information through a system boundary. $F_\pi$ is generic flux, whereas $F_Y$ is the flux of information pertaining to $Y$.
	
	When stating the inequalities, we tacitly assume that they contain no conflicting infinities (e.g., expressions of the form $\infty-\infty$), whose result would be ill-defined. Sufficient (but not necessary) conditions are $\mathbf I(\pi)<\infty$ and $H(Y)<\infty$, since then the $D_\mathrm{KL}$ and $I$ terms are finite.
	
	Our first continuity inequality is a local version of the second law of thermodynamics (cf. \Cref{thm:secondmarkov}). It states that the KL divergence is non-increasing, net of incoming and outgoing fluxes. Hence, the KL divergence can be thought of as a kind of \emph{resource}: the only way to obtain more of it is to take some from elsewhere.
	
	\begin{theorem}[Resource Law]
		\label{thm:resgeneral}
		Fix $\pi\in\meas(\mathcal T,\mathcal S)$, and an SPCA $\mathbf C$ with dynamics $p\in\mathrm{SM}_\pi(\mathcal S)$. For $t,u\in\frac 12\mathbb Z^+$ with $t\le u$, and any system $A$,
		\[\kl{C_{t,A}}\pi \ge \kl{C_{u,A}}\pi + \sum_{v=\lceil t\rceil}^{\lceil u\rceil-1} F_\pi(v,A).\]
	\end{theorem}
	
	\begin{proof}
		For $v\in\mathbb Z^+$, rearranging the definition of net flux relative to $\pi$ yields
		\[\kl{C_{v,A}}\pi = \kl{C_{v+\frac 12,A}}\pi + F_\pi(v,A).\]
		
		For $v\in\mathbb Z^++\frac 12$, every system satisfies $A_v = A_{v+\frac 12}$ by definition. Cellwise application of the dynamics $p$ corresponds to a single application of the Kronecker product
		\[p^{A_v}:=\bigotimes_{x\in A_v}p\in\mathrm{SM}_{\pi^{A_v}}(\mathcal S^{A_v}).\]
		
		Thus, $C_{v,A}\succeq_\pi C_{v,A+\frac 12}$. By \Cref{thm:secondmarkov}, it follows that
		\[\kl{C_{v,A}}\pi \ge \kl{C_{v+\frac 12,A}}\pi.\]
		
		Summing these equations and inequalities over all $v\in[t,u)\cap\frac 12\mathbb Z^+$ yields the desired result.
	\end{proof}
	
	Compared to the famed second law of thermodynamics, the next result seems more mysterious. It concerns information pertaining to an event that is either (1) in the past, or (2) spacelike separated. In the absence of communication through the system boundary, such information cannot increase, though it may spontaneously decrease.
	
	\begin{theorem}[Memory Law]
		\label{thm:memgeneral}
		Fix an SPCA $\mathbf C$. For any $t,u\in\frac 12\mathbb Z^+$ with $t\le u$, and any system $A$, let $Y$ be any function of $C_{\mathrm{NonFut}(A')}$, where $A':=\bigcup_{v\in [t,u)}\left(\{v\}\times A_v\right)$. Then,
		\[I(C_{t,A};\;Y) \ge I(C_{u,A};\;Y) + \sum_{v=\lceil t\rceil}^{\lceil u\rceil-1} F_Y(v,A).\]
	\end{theorem}
	
	\begin{proof}
		For $v\in\mathbb Z^+$, rearranging the definition of net flux relative to $Y$ yields
		\[I(C_{v,A};\;Y) = I(C_{v+\frac 12,A};\;Y) + F_Y(v,A).\]
		
		For $v\in\mathbb Z^++\frac 12$, we have $A_v = A_{v+\frac 12}$. The cellwise evolution step from \Cref{eq:cellularm2} is independent of $\mathrm{NonFut}(\{v\}\times A_v)$, which contains $\mathrm{NonFut}(A')$; thus, any stochasticity in this evolution is independent of $Y$. By the data processing inequality, it follows that
		\[I(C_{v,A};\;Y) \ge I(C_{v+\frac 12,A};\;Y).\]
		
		Summing these equations and inequalities over all $v\in[t,u)\cap\frac 12\mathbb Z^+$ yields the desired result.
	\end{proof}
	
	In order to apply \Cref{thm:resgeneral,thm:memgeneral}, we would like to get a handle on the flux terms. In the worst-case, their magnitude is comparable to $\mathbf I(\pi)$, multiplied by the system's effective ``surface area'', i.e., the number of cells worth of content that may cross its boundary. \Cref{cor:klprod} bounds the contributions to $F_\pi$ arising from addition and removal of arbitrary cells. For a system that expands or contracts at the ``speed of light'', only additions or removals, respectively, are possible. The one-way nature of such a system's fluxes enables their elimination from the corresponding inequalities.
	
	\begin{corollary}
		\label{cor:opensystems}
		Fix $\pi\in\meas(\mathcal T,\mathcal S)$, and an SPCA $\mathbf C$ with dynamics $p\in\mathrm{SM}_\pi(\mathcal S)$. For $t,u\in\frac 12\mathbb Z^+$ with $t\le u$, and any region $A$, let $A_+ := \mathcal T_+^{\lceil u \rceil - \lceil t \rceil}(A)$ and $A_- := \mathcal T_-^{\lceil u \rceil - \lceil t \rceil}(A)$. Let $Y$ be any function of $C_{\mathrm{NonFut}(\{t\}\times A)}$. Then,
		
		\begin{align*}
			H_\pi(C_{t,A}) &\le H_\pi(C_{u,A_+}),
			\\J_\pi(C_{t,A}) &\ge J_\pi(C_{u,A_-}),
			\\I(C_{t,A};\, Y) &\ge I(C_{u,A_-};\, Y).
		\end{align*}
	\end{corollary}
	
	\begin{proof}
		We start with the first inequality. It holds trivially if $\inf_{s\in\mathcal S}\pi(s)=0$, since then $H\equiv\infty$ by definition. Otherwise, by multiplying $\pi$ by a constant, we can assume without loss of generality that $\inf_{s\in\mathcal S}\pi(s)=1$, so that $\kl\cdot\pi \equiv -H_\pi(\cdot)$. Extend $A$ into a system by setting $A_t:=A$ and $A_{v+1}:=A_{v+\frac 12} := \mathcal T_+(A_v)$ for $v\in\mathbb [t,u)\cap\mathbb Z^+$. It follows that $A_u = A_+$ and, by \Cref{thm:resgeneral},
		\[-H_\pi(C_{t,A}) \ge -H_\pi(C_{u,A_+}) + \sum_{v=\lceil t\rceil}^{\lceil u\rceil-1} F_\pi(v,A).\]
		
		The first inequality now follows if the fluxes can be shown to be non-negative. Indeed, since $C_{v+\frac 12,A}$ contains all of the subcells in $C_{v,A}$, \Cref{lem:mutinf} implies
		\[F_\pi(v,A) = H_\pi(C_{v+\frac 12,A}) - H_\pi(C_{v,A}) = H_\pi(C_{v+\frac 12,A}\mid C_{v,A}) \ge 0.\]
		
		The second inequality is proved in a similar manner, assuming without loss of generality that $\sum_{s\in\mathcal S}\pi(s)=1$, so that $\kl\cdot\pi \equiv J_\pi(\cdot)$. Extend $A$ into a system by setting $A_t:=A$ and $A_{v+1}:=A_{v+\frac 12} := \mathcal T_-(A_v)$ for $v\in\mathbb [t,u)\cap\mathbb Z^+$. It follows that $A_u = A_-$ and, by \Cref{thm:resgeneral},
		\[J_\pi(C_{t,A}) \ge J_\pi(C_{u,A_-}) + \sum_{v=\lceil t\rceil}^{\lceil u\rceil-1} F_\pi(v,A).\]
		
		The second inequality now follows if the fluxes can be shown to be non-negative. Indeed, since $C_{v,A}$ contains all of the subcells in $C_{v+\frac 12,A}$, \Cref{lem:mutinf} implies
		\[F_\pi(v,A) = J_\pi(C_{v,A}) - J_\pi(C_{v+\frac 12,A}) = J_\pi(C_{v,A}\mid C_{v+\frac 12,A}) \ge 0.\]
		
		For the third inequality, apply \Cref{thm:memgeneral} to the same system $A$ to get
		\[I(C_{t,A};\,Y) \ge I(C_{u,A_-};\,Y) + \sum_{v=\lceil t\rceil}^{\lceil u\rceil-1} F_Y(v,A).\]
		
		We are done, since $C_{v,A}$ containing all of the subcells in $C_{v+\frac 12,A}$ also implies that
		\[F_Y(v,A) = I(C_{v,A};\,Y) - I(C_{v+\frac 12,A};\,Y) \ge 0.\]
	\end{proof}
	
	\subsection{Closed Systems}
	\label{sec:closed}
	
	\Cref{cor:opensystems} successfully eliminates the flux terms, at the price of needing to grow or shrink the region $A$ at the ``speed of light''; this is done to block information from crossing $A$'s boundary in the undesired direction. In practice, we often encounter physical systems of fixed size that are nonetheless well-behaved, experiencing little to no flux. The ideal case is a \textbf{closed system} which has no interaction with its external environment, hence zero flux.
	
	A system and environment can block interactions by working together to maintain an inert medium between them. We do not concern ourselves with the mechanisms for maintaining such a medium, but simply assume every cell on the boundary to be fixed to a quiescent state $q\in\mathcal S$. If such a system's initial and final volumes are equal, it emits and absorbs equal numbers of copies of $q$, resulting in a total of zero net flux.
	
	\begin{definition}
		\label{def:closed}
		Fix an SPCA $\mathbf C$ and a system $A$. For $t\in\mathbb Z^+,\,\tau\in\mathcal T$, define the \textbf{boundary operator} $\partial$ by
		\[\partial_{t,\tau}A:=\left(A_t\cap\tau^{-1}(\mathcal X\setminus A_{t+1})\right)\cup\left((\mathcal X\setminus A_t)\cap\tau^{-1}(A_{t+1})\right).\]
		
		Given $q\in\mathcal S$ and $t,u\in\frac 12\mathbb Z^+$, we say that $A$ is \textbf{$q$-closed} over the time interval $[t,u]$ if, with probability one, for all $v\in [t,u)\cap\mathbb Z^+$, $\tau\in\mathcal T$, $x\in\partial_{v,\tau}A$, we have $C_{v,x,\tau}=q_\tau$.
	\end{definition}
	
	Note that, since $\partial$ is only defined at integer times, $q$-closedness holds vacuously over cellwise evolution steps. This makes sense, since no information travels across cell boundaries during these steps. We are finally ready to state the laws for closed systems.
	
	\begin{theorem}[Resource Law for Closed Systems]
		\label{thm:resclosed}
		Fix $\pi\in\meas(\mathcal T,\mathcal S)$, and an SPCA $\mathbf C$ with dynamics $p\in\mathrm{SM}_\pi(\mathcal S)$. Given $q\in\mathcal S$ and $t,u\in\frac 12\mathbb Z^+$ with $t\le u$, let the system $A$ be $q$-closed over $[t,u]$. If $|A_t|=|A_u|$ or $\pi(q)=1$, then
		\[\kl{C_{t,A}}{\pi} \ge \kl{C_{u,A}}{\pi}.\]
		
		Furthermore, if $|A_t|=|A_u|$, then $C_{t,A}\succeq_\pi C_{u,A}$.
	\end{theorem}
	
	\begin{proof}
		By definition, $\partial_{v,\tau}A$ consists precisely of the positions on track $\tau\in\mathcal T$, whose contents at time $v\in\mathbb Z^+$ are either about to enter $A$ from outside, or exit it from inside. Since A is $q$-closed, the subcells at these positions are all in the state $q_\tau$. The number of such subcells leaving $A$ exceeds the number entering $A$ by precisely $|A_v| - |A_{v+1}|$. This number does not depend on $\tau$; hence, up to permutations on each track, the translation step simply emits (or absorbs, if negative) this many whole copies of $q$ from the system. Therefore,
		\[F_\pi(v,A) = (|A_v| - |A_{v+1}|)\kl{\delta_q}{\pi}
		= (|A_v| - |A_{v+1}|)\log\frac 1{\pi(q)}.\]
		
		Substituting into \Cref{thm:resgeneral} yields
		\[\kl{C_{t,A}}{\pi}
		\ge \kl{C_{u,A}}{\pi} + (|A_t|-|A_u|)\log\frac 1{\pi(q)},\]
		
		which reduces to the desired inequality when $|A_t|=|A_u|$ or $\pi(q)=1$.
		
		In order to show that $|A_t|=|A_u|\implies C_{t,A}\succeq_\pi C_{u,A}$, we take advantage of the fact that, from $A$'s point of view, the external Universe might as well be an endless sea of $q$'s. To be precise, let $I:=[t,u]\cap\frac 12\mathbb Z^+$, $B := \bigcup_{v\in[t,u]}A_v$, and define the random variable $\mathbf{C'}:I\times B\rightarrow\mathcal S$ by
		\[C'_{v,x} :=
		\begin{cases}
			C_{v,x} & \text{if }x\in A_v,
			\\q & \text{if }x\in B\setminus A_v.
		\end{cases}\]
		
		Let's show that $C'_v\succeq_\pi C'_{v+\frac 12}$ for all $v\in [t,u)\cap\frac 12\mathbb Z^+$.
		
		The case $v\in\mathbb Z^+$ corresponds to a trackwise translation step, with the caveat that each $\tau\in\mathcal T$ may map some elements of $B$ to its complement, and vice-versa. Luckily, all such subcells are in the state $q_\tau$, so there exists some $\tau'\in\mathrm{Bij}(B)$ that produces the same effect as $\tau$ on $C'_v$. Since $\pi\in\meas(\mathcal T,\mathcal S)$, $\pi^B$ is stationary with respect to permutations of each track, and therefore $C'_v\succeq_\pi C'_{v+\frac 12}$.
		
		The case $v\in\mathbb Z^++\frac 12$ corresponds to a cellwise dynamics step. Since $C'_{v,x}=C'_{v+\frac 12,x}=q$ for all $x\in B\setminus A_v$, the Kronecker product $p^{A_v}\otimes\mathrm{id}^{B\setminus A_v}$ transforms the distribution of $C'_v$ into that of $C'_{v+\frac 12}$. Since $p,\mathrm{id}\in\mathrm{SM}_\pi(\mathcal S)$, we have $p^{A_v}\otimes\mathrm{id}^{B\setminus A_v}\in\mathrm{SM}_{\pi^B}(\mathcal S^B)$, and therefore $C'_v\succeq_\pi C'_{v+\frac 12}$.
		
		By transitivity of the relation $\succeq_\pi$, it follows that $C'_{t}\succeq_\pi C'_{u}$. By the definition of $\mathbf{C'}$,
		\[C_{t,A}\times q^{B\setminus A_t}
		\succeq_\pi C_{u,A}\times q^{B\setminus A_u}.\]
		
		When $|A_t|=|A_u|$, the deterministic $q$ factors drop out and we conclude that $C_{t,A}\succeq_\pi C_{u,A}$.
	\end{proof}
	
	\begin{theorem}[Memory Law for Closed Systems]
		\label{thm:memclosed}
		Fix an SPCA $\mathbf C$, $q\in\mathcal S$, and $t,u\in\frac 12\mathbb Z^+$ with $t\le u$. Let the system $A$ be $q$-closed over $[t,u]$, and let $Y$ be any function of $C_{\mathrm{NonFut}(A')}$, where $A':=\bigcup_{v\in [t,u)}\{v\}\times A_v$. Then,
		
		\[I(C_{t,A};\, Y) \ge I(C_{u,A};\, Y).\]
		
		If $B$ is another $q$-closed system over $[t,u]$, satisfying $A_v\cap B_v=\emptyset$ for all $v\in[t,u)\cap\mathbb Z^+$, then
		
		\[I(C_{t,A};\, C_{t,B}) \ge I(C_{u,A};\, C_{u,B}).\]
	\end{theorem}
	
	\begin{proof}
		For $v\in [t,u)\cap\mathbb Z^+$, $C_{v,A}$ and $C_{v+\frac 12,A}$ hold identical data aside from the deterministic addition or removal of $q$'s; hence, $F_Y(v,A)=0$. The first inequality now follows directly from \Cref{thm:memgeneral}.
		
		The second inequality will follow by induction if we can show that, for all $v\in [t,u)\cap\frac 12\mathbb Z^+$,
		\[I(C_{v,A};\, C_{v,B}) \ge I(C_{v+\frac 12,A};\, C_{v+\frac 12,B}).\]
		
		The case $v\in\mathbb Z^+$ follows as in the first inequality, because the trackwise translation step does not change the information content of either system.
		
		The case $v\in\mathbb Z^++\frac 12$ follows from two applications of the first inequality: first with $C_{v,B}$ in place of $Y$, and then a second time with $C_{v+\frac 12,A}$ in place of $Y$. In the second application, $A_v\cap B_v=\emptyset$ ensures that $\{v+\frac 12\}\times A_v\subset\mathrm{NonFut}(\{v\}\times B_v)$. Chaining the two applications together yields
		\[I(C_{v,A};\, C_{v,B})
		\ge I(C_{v+\frac 12,A};\, C_{v,B})
		\ge I(C_{v+\frac 12,A};\, C_{v+\frac 12,B}).\]
	\end{proof}
	
	To review what we've learned, the ``arrow of time'' is an umbrella term for a bundle of important asymmetries with respect to time reversal. For Markov chains, two such asymmetries are time-homogeneity and the second law of thermodynamics. The former may be viewed as more fundamental, since it implies the latter. Indeed, even if the dynamics are not quite homogeneous, provided that they stay within $\mathrm{SM}_\pi(\mathcal S)$, \Cref{thm:secondmarkov} implies monotonicity of $\kl\cdot\pi$.
	
	For SPCAs, the analogous asymmetries are spacetime-homogeneity and the Resource Law. The fundamental asymmetry, from which all others are derived, is the combination of spacetime-homogeneity and locality that is encoded in \Cref{eq:cellularm2}. In the deterministic presentation, the role of fundamental asymmetry is instead taken by the initial condition in \Cref{eq:cellulard1}. By \Cref{thm:equivalence2}, the two presentations are equivalent and therefore equally valid.
	
	When the spatial geometry $(\mathcal X,\mathcal T)$ is nontrivial, there are two additional derived asymmetries: the Memory Law, and Pearl's $d$-separation criterion as applied to $\mathrm G(\mathcal X,\mathcal T)$. We've seen that a Markov chain's edges can be reversed, and its $d$-separation criterion expressed as a symmetric Markov property. According to \citet[Thm 1.2.8]{pearl2009causality}, this is not the case for the more interesting causal graphs we encounter here. Consequently, an SPCA's backward ``dynamics'' are not merely inhomogeneous; in failing to be Markov relative to the reversed graph, they are also \emph{nonlocal}.
	
	Our definition of a closed system is somewhat restrictive; in future work, it may prove useful to develop more general estimates of the flux terms. On the other hand, \Cref{thm:resclosed,thm:memclosed} are valid asymmetries, and they are powerful enough to illustrate some key applications of our model. Before turning to those, let's derive one easy extension. Perhaps the main weakness of \Cref{def:closed} is that it must occur with probability one, given only the initial condition, i.e., the Big Bang. In practice, we would like to condition on an event $E$, conveying the preparation of a closed system before some experiment start time $t\in\mathbb Z^+$. Fortunately, if $E$ is a function of $C_{\le t}$, it follows from the structure of $\mathrm G(\mathcal X,\mathcal T)$ that the conditional distribution $\prob_{C_{\ge t}\mid E}$ is itself an SPCA with the same dynamics, only its initial condition is changed to $\prob_{C_t\mid E}$. Therefore, our results extend to the setting of $\prob_{C_{\ge t}\mid E}$.
	
	
	\section{Applications}
	\label{sec:discussion}
	
	
	The arrow of time is so deeply ingrained in our intuitions, and even in the language of science and mathematics, that it's hard to reason about from first principles. In the absence of a precise model, it's extremely easy for assumptions to go unchecked, as evidenced by historical debates such as the one over Maxwell's demon \citep{bennett1982thermodynamics}. As a partial solution, \citet{pearl2009causality} brought mathematical precision to causal concepts; however, he did not specify their connection to fundamentally reversible physical laws.
	
	In this article, the SPCA brings a similar degree of mathematical precision to emergent irreversibility, from reversible first principles. It's quite general, being parametrized by $(\mathcal X,\mathcal T,\mathcal S,\mu,p)$, and more tractable than realistic field theories from physics. Formerly ambiguous concepts, such as an experimenter's ``knowledge'' (usually of the past) and ``free will'' (to influence the future), can be endogenized within such a model, enabling clearer analyses. For maximum generality, rather than constructing specific SPCAs, we utilize the theorems of \Cref{sec:markov,sec:spca}.
	
	This section provides a brief overview of some contentious topics that illustrate the clarifying power of our theory. For the sake of brevity, we sometimes blur the distinction between SPCAs and physical reality. Implicit in this approach is the hypothesis that SPCAs make accurate predictions regarding information-theoretic properties of the Universe, at least in some classical limit that evades the peculiarities of quantum information. Some of these predictions may merit future experiments; however, for the purposes of our theoretical exposition, we find assurance in our predictions' consistency with established science and everyday experience.
	
	\subsection{Negentropy as a Resource}
	\label{sec:appresource}
	
	Colloquially, energy is often thought of as a resource, to be irreversibly consumed to perform various tasks. Physics describes the situation differently, with conservation laws that forbid the total energy from changing. Instead, energy is irreversibly converted from more useful forms (referred to as \textbf{free energy}) into less useful forms (e.g., ambient heat). While a detailed analysis of conserved quantities is beyond the scope of this article, let's take a moment to summarize their relationship with negentropy, the fundamental resource of statistical mechanics.
	
	A collection of $d\in\mathbb N$ scalar quantities is represented by a vector function $E:\bigcup_\tau \mathcal S_\tau \rightarrow \mathbb R^d$. We say $E$ is \textbf{conserved} for $p\in\mathrm{SM}_\pi(\mathcal S)$ if, for all $s,s'\in\mathcal S$ with $p(s,s') > 0$, we have $\sum_\tau E(s_\tau) = \sum_\tau E(s'_\tau)$. $E$ extends to arbitrary systems by $E(C_{t,A}):=\sum_{(x,\tau)\in A_t\times\mathcal T}E(C_{t,x,\tau})$. For conserved $E$ and a closed system $A$, the distribution of $E(C_{t,A})$ is unchanging with $t$ \footnote{Given a $q$-closed system $A$, without loss of generality let $E(q)=0$. Then, cellwise evolution steps conserve $E(C_{t,A})$, and so do trackwise translation steps which only add or remove copies of $q$.}.
	
	Now, for any non-intersecting pair of systems $A,B$, a direct substitution into \Cref{cor:abstractdecomp} decomposes their aggregate negentropy into three non-negative parts:
	\begin{align}
		\label{eq:spcadecomp}
		J_\pi(C_{t,A\cup B})
		= J_\pi(C_{t,A}) + J_\pi(C_{t,B}) + I(C_{t,A};\, C_{t,B}).
	\end{align}
	
	By \Cref{def:entropy}, $J_\pi$ is the amount by which a system's KL divergence (relative to $\pi$) exceeds the infimum over all probability distributions. In the presence of conserved quantities, this unconstrained infimum may be unattainable. Consequently, it's more useful to let $J^E_\pi(X)$ be the amount by which $\kl{X}{\pi}$ exceeds its \emph{constrained} infimum, subject to keeping $\prob_{E(X)}$ unchanged. Similarly, let $I^E(X;Y) := I(X;Y) - I(E(X); E(Y))$ be the excess mutual information between a pair of systems, beyond that required by their conserved quantities.
	
	When both $A$ and $B$ are closed, the joint distribution $\prob_{E(C_{t,A}),E(C_{t,B})}$ is fixed. By simultaneously minimizing every term on the right-hand side of \Cref{eq:spcadecomp}, subject to this constraint, we find that the KL divergence of $C_{t,A\cup B}$ exceeds its constrained infimum by $J^E_\pi(C_{t,A}) + J^E_\pi(C_{t,B}) + I^E(C_{t,A}; C_{t,B})$. On the other hand, if the systems $A$ and $B$ are allowed to come into contact, then only the sum's distribution $\prob_{E(C_{t,A\cup B})}=\prob_{E(C_{t,A})+E(C_{t,B})}$ is fixed. Subject to this weaker constraint, the excess KL divergence increases to $J^E_\pi(C_{t,A\cup B})$. Letting $F^E(C_{t,A}; C_{t,B})$ denote the difference between the two values, we have
	\begin{align}
		\label{eq:spcadecompF}
		J^E_\pi(C_{t,A\cup B})
		= J^E_\pi(C_{t,A}) + J^E_\pi(C_{t,B}) + I^E(C_{t,A};\, C_{t,B}) + F^E(C_{t,A};\, C_{t,B}).
	\end{align}
	
	\Cref{eq:spcadecompF} generalizes the decomposition of \Cref{eq:spcadecomp}, and reduces to it in the case where $E$ is a constant function. In intuitive terms, \Cref{eq:spcadecompF} decomposes the \emph{aggregate} resource of $A\cup B$ into four non-negative components: the \emph{internal} resource ($J_\pi^E$) in each of $A$ and $B$, as well as \emph{mutual} resources due to their correlation ($I^E$) and their potential to trade conserved quantities ($F^E$).
	
	Suppose $A\cup B$ is a closed system to which \Cref{thm:resclosed} applies, so that the left-hand side is non-increasing. In general, when $A$ and $B$ are in contact, the negentropy may redistribute arbitrarily between the right-hand terms, subject to their total not increasing. However, when $A$ and $B$ are also closed off from each other, we can constrain the individual terms. Indeed, all four terms are non-increasing: the first and second by \Cref{thm:resclosed}, and the third by \Cref{thm:memclosed}. The fourth term, being a function of the unchanging distribution $\prob_{E(C_{t,A}),E(C_{t,B})}$, is constant.
	
	Thus, if the two systems $A,B$ are allowed to relax separately, then the first three resource terms may drop to zero, but the fourth retains its starting value. When
	\[J^E_\pi(C_{t,A}) = J^E_\pi(C_{t,B}) = I^E(C_{t,A};\, C_{t,B}) = 0,\]
	both systems are said to be in \textbf{thermodynamic equilibrium}, and \Cref{eq:spcadecompF} simplifies to
	\[J^E_\pi(C_{t,A\cup B})=F^E(C_{t,A};\, C_{t,B}).\]
	
	For example, consider two tanks of inert gas: a hot tank $A$ and a cold tank $B$. Individually, not much can be done with either tank. However, when they are brought together, a heat engine can extract the jointly held resource $F^E(C_{t,A}; C_{t,B})$; this is essentially the free energy, expressed in units of negentropy. For further study of the equilibrium setting, in which only the $F^E$ term is positive, we defer to the standard literature in thermodynamics. This article focuses on the remaining terms, so let's dismiss $E$ and return to \Cref{eq:spcadecomp}.
	
	The third term $I(C_{t,A}; C_{t,B})$ merits further examination. It's typically much smaller than the other terms: it becomes zero when either of $A,B$ is relaxed to equilibrium, and can only stay high by protecting densely packed data from environmental noise. To get a sense of scale, copying a 125 TB magnetic tape drive filled with random data (e.g., a one-time pad) would yield only $8\cdot 125\cdot 10^{12}=10^{15}$ bits of mutual information, whereas the physical entropy of any household object exceeds Avogadro's number ($>10^{23}$ bits). At room temperature, each bit is worth $3\times 10^{-21}$ Joules of free energy \citep{bennett1982thermodynamics}, so the data on our hypothetical drive is only worth $3$ microjoules. This renders the mutual information irrelevant to most thermodynamical calculations; on the other hand, we will see its relevance in some of the topics to follow.
	
	\subsection{Energy-Efficient Computing}
	\label{sec:applandauer}
	
	In the digital age, the design of computing devices is constrained by limits on power consumption and heat generation. Thus, a question of great practical interest is: how much computation can be performed on a given free energy budget? Von Neumann initially proposed that every elementary computation should cost at least one bit of negentropy, thereby dissipating $3\times 10^{-21}$ Joules at room temperature. \citet{landauer1961irreversibility} refined this hypothesis, arguing that only irreversible computations, which erase information, carry this cost. \citet{frank2018physical} gave further refinements; the present subsection is mainly an effort to streamline his arguments using our \Cref{thm:equivalence1,thm:resclosed,thm:memclosed}.
	
	Landauer's principle must be interpreted with care. Consider the following operations:
	\begin{itemize}
		\item FLIP: overwrite a known bit with another known bit
		\item ROLL: overwrite a random bit with an independently random bit
		\item FORGET: overwrite a known bit with a random bit
		\item SET: overwrite a random bit with a known bit
	\end{itemize}
	
	Naively, one might expect at least some of these operations to be thermodynamically costly. However, suppose we carry reserves of 0s, 1s, and randomized bits, in specified regions of memory. Then, each of the operations is logically reversible, with the net effect of draining one reserve while replenishing another. In fact, the reserve of 1s is redundant, since 0s and 1s can be reversibly negated.
	
	For the purposes of analysis, it's helpful to simplify the SPCA model. To start with, we imagine an SPCA with doubly stochastic dynamics, whose subcellular state spaces $\mathcal S_\tau$ encode fixed-width binary numbers. By the last part of \Cref{thm:resclosed}, any closed system of constant volume also inherits doubly stochastic dynamics. Consequently, \Cref{thm:equivalence1} implies that no matter how we engineer a system \footnote{One might wonder how it's possible to engineer a system's behavior at all, when the dynamics $p\in\mathrm{DM}(\mathcal S)$ are forever fixed. A straightforward way to embed a variety of dynamics into a single $p$ is to break the state into \emph{control} and \emph{data} parts: $\mathcal S=\mathcal S_\mathrm{control}\times\mathcal S_\mathrm{data}$. If the control is fixed, i.e., $p((c,d),(c',d')) = 0$ whenever $c\ne c'$, then it's easy to show that for every fixed $c\in\mathcal S_\mathrm{control}$, we get an effective dynamics $p_\mathrm{data}(c)\in\mathrm{DM}(\mathcal S_\mathrm{data})$. Cellular automata can be very versatile, as evidenced by the universality results surveyed by \citet{kari2018reversible}.}, its effective transition must always be a random mixture of bijections. To avoid the trouble of constructing full-fledged SPCAs, we therefore model our hypothetical computer as a grid of binary numbers, on which we are free to apply any random mixture of bijections.
	
	We make use of two memory reserves: one containing only zeros, the other containing unbiased i.i.d. bits. The two reserves are jointly modeled by a single system with state space
	\[\mathcal E = \left\{(n,m)\in\mathbb Z^+\times\{0,1\}^{\mathbb Z^+}:\;\forall i<n,\,m_i=0\right\}.\]
	
	The \textbf{charge level} $n$ indicates the size of the reserve of 0s, or \textbf{battery}, taking up a prefix of $m$. The reserve of random bits, or \textbf{heat bath}, occupies the trailing end of $m$. When the system is left alone, its dynamics are as follows: the charge level and battery do not change, but the heat bath mixes rapidly.
	
	To express these dynamics as a mixture of bijections, we treat $n$ as a fixed control. At each position $i<n$, we preserve the battery by applying the identity to $m_i$. At each position $i\ge n$, we randomize the heat bath by choosing randomly between the identity and negation. From any given state with charge level $n$, this system immediately collapses into a uniform equilibrium over all states compatible with the charge level. With this system in hand, let's see how each of the four listed operations are implemented.
	
	FLIP is clearly a bijection: either the identity, if the initial and target values match, or a negation otherwise. We should clarify the meaning of the requirement that the initial bit be \emph{known}. Our choice of bijection must not depend on a direct observation of the initial bit, as such a composite operation would not be bijective. We can try to guess the initial value a priori, in which case an incorrect guess would lead to the wrong result.
	
	Alternatively, we might have a copy of the initial bit elsewhere in memory. One way to copy $x$ is to take a $0$ from our reserve, and apply the reversible XOR-assignment operation $(x,y)\mapsto(x,x\oplus y)$. It maps the pair $(x,0)$ to $(x,x)$, copying $x$ at some expense to our reserve. Thus, knowledge comes from negentropy, a finding that we expand upon in \Cref{sec:appinduction}. A copy of $x$, or more generally any data that implies its value, can serve as a control, reversibly allowing for a choice of bijection that depends on $x$. The battery's charge level $n$ serves the same kind of role, informing its user of the $n$ zeros.
	
	Moving on to the ROLL operation, we'd be tempted to call it irreversible since it erases the initial value. However, it has no effect on the bit's \emph{distribution} \footnote{Indeed, since the old value is permanently lost, there is no way to detect whether ROLL has been applied. Of course, we are assuming the initial bit is independent of all external data; the correlated case is discussed toward the end of this subsection.}. ROLL can be implemented by choosing randomly between the identity and negation; alternatively, it can be implemented by swapping with a bit from the heat bath, leaving the charge level unchanged.
	
	FORGET increases the entropy (hence decreases the negentropy) by one bit. To make it reversible, we should collect the initial bit's negentropy for future use. First, we apply FLIP if necessary to set its value to zero. Then, we swap it into position $n$ of the reserve system and increment its energy level to $n+1$; this has the effect of withdrawing a random bit from the heat bath, while simultaneously depositing a zero into the battery.
	
	SET decreases the entropy from one bit to zero. The second law of thermodynamics forbids this operation in isolation, requiring that we store the entropy elsewhere. The solution is to apply FORGET in reverse, withdrawing a zero from the battery while depositing our random bit into the heat bath. Since FORGET and SET are inverses of one another, neither operation consumes negentropy; they merely transfer negentropy from bit to battery and vice-versa, respectively.
	
	The real issue is that we often want to treat nonrandom bits \emph{as if} they were random, clearing them with SET instead of FLIP. To see why, consider any long computation that produces a lot of intermediate data. Once this data is no longer needed, we'd like to clear or overwrite it to make room for new data. Since it was deterministically computed, it should in principle be erasable with FLIP, at no thermodynamic cost. On the other hand, if there is no substantially faster way to compute it than to run the source program (i.e., the data has high \textbf{logical depth} as defined by \citet{bennett1988logical}), then a successful implementation of FLIP requires repeating a similar amount of computation \footnote{Proof sketch: suppose there were a fast reversible algorithm to clear the data. Then, reversing this algorithm computes the data equally quickly, implying that its logical depth was actually low.}. Such approaches are studied in the field of \textbf{reversible computing}; for references, see \citet{Vitnyi2005TimeSA} and \citet{morita2017theory}.
	
	To avoid the extra computations, we can instead pretend the junk data is random and apply SET to it. Once placed into the heat bath, the nonrandom bit is randomized, producing new entropy. Thus, what's costly is not the SET operation \emph{per se}, but rather, the decision to give up on the bit's deterministic origins and throw it into a noisy environment. In hindsight, the conclusion that randomizing a nonrandom bit increases entropy seems too obvious; however, it was Landauer who first realized that the entropy of digital data is essentially the same as that of thermodynamics, enabling the bit's ejection into a heat bath.
	
	Our reversible implementation of SET, with its controlled exchange of exactly one bit, is quite idealized. \citet{neyman1966negentropy} considers a more realistic implementation that imposes an energy differential (say, using an electric field) between a transistor's 0 and 1 states. The transistor is then relaxed to equilibrium with an external heat bath. Conservation of energy implies that, at the low-energy state, the heat bath has more energy, hence also a larger volume of phase space configurations. As a result, the equilibrium distribution favors the low-energy state.
	
	In order to bound the probability of error per such operation by $2^{-50}$, the high-energy state must occupy a fraction no bigger than $2^{-50}$ of phase space. Thus, relaxing the transistor from that state expands our phase space volume by a factor of at least $2^{50}$. Even assuming no other imperfections, entropy then increases by not just 1, but 50 bits! In conclusion, if we ever want to dissipate less than 50 bits per operation, we must choose a different sort of implementation; and if we want to dissipate less than 1 bit per operation, the \emph{only} possibility consistent with thermodynamics is reversible computing.
	
	Finally, we present an extension of Landauer's principle. Consider two closed systems $A$ and $B$. By \Cref{thm:resclosed,thm:memclosed}, every term in the decomposition of \Cref{eq:spcadecomp} is non-increasing. Therefore, whenever any term on the right-hand side decreases, the joint negentropy $J_\pi(C_{t,A\cup B})$ must also decrease by at least the same amount. In particular, any decreases in the mutual information $I_\pi(C_{t,A};C_{t,B})$ are irreversible and thermodynamically costly.
	
	For example, suppose Alice and Bob know an inside joke, represented as copies of the same random data in their respective memories. While they are apart, if Bob forgets the joke, their mutual information would decrease. As a result, negentropy would decrease, not necessarily in Alice or Bob individually, but certainly in their union. This limitation is lifted when Alice and Bob are together: indeed, they may jointly apply XOR-assignments to freely copy and uncopy data, mapping $(x,0)$ to $(x,x)$ and vice-versa.
	
	As a consequence of this discussion, we predict that no future technology can eliminate the thermodynamic cost of sensing. While we might hope that reversible computers will someday be able to dispose of junk data arising from internal computations, the external world is much harder to control. To reversibly ``unsense'' something, we would have to revisit precisely the same scene to perform an uncopy operation. Absent such an opportunity, clearing our senses requires dissipation.
	
	\subsection{Induction and Randomness}
	\label{sec:appinduction}
	
	It's important to address the use of a probabilistic Past Hypothesis. No matter which of the three presentations we choose, \Cref{def:spca} posits a probability measure on an entire ensemble of trajectories. This presents a bit of a paradox: the inhabitants of an SPCA can only gather observations from one trajectory in the ensemble. Empirically speaking, it's meaningless to ask whether or not the remaining trajectories are ``real'', as we can never access them. Instead, the ensemble expresses the subjective uncertainty within our mind.
	
	Now, subjective does not mean arbitrary. If we take it for granted that minds are entirely physical, then any inferences contained therein must be a product of the specific trajectory in which the mind finds itself. In other words, the general ensemble somehow arises as a function of one particular member. To determine which ensemble should be inferred, let's put ourselves in the shoes of an agent inhabiting an SPCA, and ask: what can we infer about our own Universe?
	
	\subsubsection{The Oblivious God}
	
	To demonstrate the problem's difficulty, let's generously overstate the agent's prior knowledge and physical capabilities. Recalling that an SPCA is parametrized by $(\mathcal X,\mathcal T,\mathcal S,\mu,p)$, suppose that $\mathcal X,\mathcal T,\mathcal S,p$, with $|\mathcal X|<\infty$, are all known to the agent a priori. This leaves only the task of inferring the actual state $C_t\in\mathcal S^\mathcal X$, whose initial distribution $\mu$ is now unspecified. To be even more generous, we allow the dynamics to be arbitrarily inhomogeneous and nonlocal in the agent's favor. That is, imagine a God-like agent \footnote{While realistic agent models are beyond the scope of this article, we should intuitively convince ourselves that the God-agent can emulate all such agents. To sketch the idea, model an agent by some system $A$ living inside the SPCA. Let's imagine the agent maintains a periodic internal homeostasis, in which for some pair of times $t<u$, with probability one, $C_{t,A}=C_{u,A}$ (so in particular, $|A_t|=|A_u|$). The agent's presence acts as a control on $\mathcal X\setminus A$, manipulating its dynamics in a manner that depends on the specific contents of $C_{t,A}$. Regardless, if $\prob(C_u\mid C_t)$ is doubly stochastic, then so is $\prob(C_{u,\mathcal X\setminus A}\mid C_{t,\mathcal X\setminus A})$. As a result, the agent can be emulated by a disembodied God-agent living in $\mathcal X\setminus A$, with the tracks $\mathcal T$ mended to skip through the ``hole'' made by removing $A$.} who can choose which transition matrix $p(t)\in\mathrm{SM}_{\pi^\mathcal X}(\mathcal S^\mathcal X)$ the Universe applies at each time $t\in\mathbb N$, subject to a fixed stationary measure $\pi\in\meas(\mathcal T,\mathcal S)$ with $\mathbf I(\pi)<\infty$. The theorems of \Cref{sec:spca} still apply because, aside from fixing $\pi$, their proofs make no specific use of spacetime-homogeneity.
	
	To be clear, when we say that the state is unknown to the agent a priori, we mean that the transition matrices $p(t)$ are chosen independently of the configuration $\mathbf C$. This restriction prevents violations of \Cref{sec:spca}'s laws \`a la Maxwell's demon. It also means that an agent hoping to exercise state-dependent control over the environment cannot do so by manipulating $p(t)$ directly; instead, state knowledge must be stored in a physical memory. As a result of interactions with such a memory, the environment can be made to evolve in a manner that depends on the memory's contents.
	
	Let's designate some region $\mathrm{Mem}\subset\mathcal X$, in which the agent should store state information. One might imagine that such a powerful agent would have no trouble filling $\mathrm{Mem}$ with basic information about the state of its surroundings; at a minimum, it should be able to look around and verify its own physical form to be more than a disembodied brain! Formally, can the agent ensure that, at some future time $t$,
	\[I(C_{t,\mathrm{Mem}};\, C_{t,\mathcal X\setminus\mathrm{Mem}})>0?\]
	
	To gain intuition, consider the case $\pi=\sharp$. By \Cref{thm:equivalence1}, the allowed operations are random mixtures of bijections. Unfortunately, a measurement that overwrites the previous contents of $\mathrm{Mem}$, with a copy of external data, would be non-bijective. Therefore, it cannot be implemented without additional assumptions on the contents of the memory and/or environment.
	
	In the general case, we might model our ignorance about the current state $C_t$ with a maximum entropy (hence, zero negentropy) Bayesian prior over all possible states. By \Cref{thm:resgeneral}, the Universe's negentropy will thereafter always remain zero. By \Cref{eq:spcadecomp}, our knowledge about other systems, as measured by the mutual information, cannot exceed the aggregate negentropy; hence, it too is stuck at zero. In other words, if we start with a zero negentropy prior, then no experimental evidence could ever persuade us to improve our posterior belief. This prior is effectively an epistemic \textbf{heat death}. Even if we assume a priori that $\mathrm{Mem}$ is zeroed out, this would only grant as much negentropy as the size of the memory. We could then use XOR-assignments to copy small amounts of data into our memory; however, we would never acquire the considerable levels of free energy that life requires.
	
	This thought experiment is a variant of the famous \textbf{Boltzmann brain} paradox. A Boltzmann brain is a short-lived formation of a person's entire mental state, complete with senses and memories, arising from sheer chance in an unstructured environment such as a gas cloud. The paradox is that, under a heat death prior, it would be a much bigger coincidence to find a large section of structured Universe, than to find a smaller section (e.g., $\mathrm{Mem}$) encoding only your current state of consciousness. Since the latter is deemed more likely, you should not believe that your senses or memories are at all indicative of any reality outside your mind. As we've seen, not even the physics-bending God-agent can escape epistemic heat death.
	
	There is a static variant of this problem, in which we drop the dynamical structure of $p\in\mathrm{SM}_\pi(\mathcal S)$ and simply ask: is it possible to extrapolate knowledge from part of an object (e.g., $\mathrm{Mem}$), to the whole object (e.g., $\mathcal X$)? For suitable choices of ``object'', this is addressed in philosophy as the \textbf{problem of induction}, and in statistical learning theory by \textbf{No Free Lunch} theorems. Unfortunately, without some prior bias, the answer appears to be negative \citep{wolpert1997no}.
	
	\subsubsection{Universal Induction}
	
	If the preceding conclusions seem counterintuitive, it's because everyday inductive reasoning is not based on zero negentropy priors. Instead, our priors contain biases, learned from experiences collected throughout our lives. Unfortunately, this merely regresses the problem to another instance of itself: the lifelong learning of good priors is itself an inductive procedure, subject to the same No Free Lunch theorem, requiring that we carry a biased prior from birth. At the next level of regress, we turn to Darwinian evolution, which is yet another inductive procedure, again subject to the No Free Lunch theorem! Which prior does evolution use?
	
	To illustrate, consider a Darwinian contest among abstract creatures tasked with predicting an unknown function $f:\mathbb N\rightarrow\mathbb Z_m$. On day $n$, the creatures are asked to produce guesses of $f(1),f(2),\ldots,f(n)$. The creatures who correctly guess all $n$ values survive, and repopulate for the next day. As a conceptual aid, we can identify each creature with the guesses $\hat f:\mathbb N\rightarrow\mathbb Z_m$ that it would produce. Note that any given creature's guesses are deterministic.
	
	Suppose the creatures randomly mutate in such a way that, without selective pressure, the values $(\hat f(n))_{n\in\mathbb N}$ would be i.i.d. and uniform on $\mathbb Z_m$. Now, the filter of natural selection amounts to conditioning the population distribution on the constraints $\hat f(i)=f(i)$, for all $i\le n$. The surviving agents would satisfy all $n$ constraints, but be uniformly distributed on all arguments larger than $n$. Even if $f$ has some convenient structure, such as periodicity, it appears that natural selection fails to promote generalization.
	
	Why is this failure not reflected in the real world? According to \citet{johnston2022symmetry}, mutations actually produce highly non-uniformly distributed \textbf{phenotypes} (e.g., functions $\hat f$). This is because mutations do not act directly on the phenotype $\hat f$, but on encodings known as \textbf{genotypes}, which are algorithmically mapped into phenotypes. If $\hat f$ is \emph{simple}, meaning it corresponds to a short genotype, then it's much easier for random mutations to chance upon it.
	
	Thus, evolution is biased toward simplicity \footnote{Since the genotype-phenotype maps considered by \citet{johnston2022symmetry} are not Turing-complete, their notion of simplicity is highly dependent on the map in question. While it seems plausible that the behavior $\hat f$ can be genetically programmed in a Turing-complete language, the encoding of such \emph{innate} behaviors remains a subject of active research \citep{anholt2020evolution}.}. How can we know if this is a good bias? Obviously, we get a Free Lunch if the truth happens to agree with our bias; however, mismatched biases would just as easily blind us to truths. The situation appears no better than guessing!
	
	Fortunately, the simplicity bias, also known as \textbf{Occam's razor}, has an exceptionally strong theoretical justification. If we accept the Church-Turing thesis, then any practical inference procedure must be \emph{computable}. An environment that defeats all computable inference procedures (e.g., a typical sample from a heat death distribution), is one in which inference is impossible. In all other environments, one can only hope to do as well as the computable procedure that's best suited for that environment. The most celebrated discovery of \citet{solomonoff1964formal} is that all computable probability distributions are dominated, up to constant factors, by a \textbf{universal prior}.
	
	Without going into formal definitions, we give a generative description of Solomonoff's universal prior. First, a random computer program (or genotype) is generated by uniformly sampling an infinite sequence of $0$s and $1$s. The program is mapped through a universal computer (or genotype-phenotype map), to get a non-uniformly distributed output (or phenotype). Despite its infinite length, the program can contain instructions to stop reading after a finite prefix; thus, a finite program $p$ is read with probability $2^{-|p|}$, giving more weight to shorter programs. Just as in our discussion of the algorithmic entropy in \Cref{sec:hdiscrete}, the constant factor hidden in the choice of universal computer can be bounded by restricting to computers with small physical implementations.
	
	Solomonoff induction can be applied to any partially observed dataset, essentially treating it as the output of a uniformly random computer program. The dataset should include any prior information that we want to take into account. Performing Bayesian updates on the observed part yields posterior inferences on the unobserved part. Despite its Bayesian character, \citeauthor{rathmanner2011philosophical} \citep{hutter2004universal,rathmanner2011philosophical} survey some frequentist guarantees of Solomonoff induction. Moreover, they explain how the universal prior can serve as an \emph{interpretation} of probability.
	
	The universal prior's major drawback is that it can only be computed in the limit of infinite computation time; as such, it is not a \emph{practical} inference procedure. Instead, it provides a mathematically well-defined ideal against which all inference procedures may be compared. Its role for induction is analogous to that of a brute force search algorithm that plays perfect chess.
	
	After replacing the uniform prior on $C_t$ with the universal prior, Boltzmann brains cease to pose a problem. This is because simpler extrapolations, in which the external Universe $C_{t,\mathcal X\setminus\mathrm{Mem}}$ is consistent with our senses $C_{t,\mathrm{Mem}}$, are weighted higher. Despite the concerns expressed by \citet{carroll2020boltzmann}, it's even possible to handle theories of physics that predict many more Boltzmann brains than real brains, because the universal prior does not weigh the brains equally. We leave a detailed estimation of the relevant posteriors to future work.
	
	Note that Solomonoff's posterior amounts to a weighted sum over all theories consistent with observation. Individual probabilistic theories, such as our SPCAs, \citeauthor{albert2001time}'s Mentaculus, and even quantum theories, can be understood in terms of the two-part codes from \Cref{sec:hdiscrete}. The code's first part is a \textbf{theory}, detailing a ``typical set'' of possible outcomes, often (though not always) arising from a well-behaved probability measure. The second part is an \textbf{index} into this set, specifying one of its elements \footnote{For probabilistic theories, a convenient way to form the index is with Shannon-Fano or arithmetic coding.}. Practical considerations demand that the theory's description be much shorter than the index.
	
	In the case of quantum theories, the index would use Born's rule to optimally encode measurements. An SPCA's snapshot $C_t\in\mathcal S^\mathcal X$ also has a canonical two-part code. The theory part consists of $(\mathcal X,\mathcal T,\mathcal S,\mu,p,t)$, which together determine the probability distribution $\prob_{C_t}$. The index part encodes the snapshot $C_t$ in a manner optimized for this distribution. If the resulting code is approximately incompressible, its weight in Solomonoff's posterior can be shown to be substantial; hence, the theory is considered a good explanation of the observed data. In practice, compressibility is checked using statistical tests \citep[\S2.4]{li2019introduction}. For more on two-part codes, see \citet{wallace2005statistical}.
	
	\subsection{Compressibility as a Resource}
	\label{sec:appcompression}
	
	The universal prior has fractal-like non-uniformities, precluding the existence of a typical set in the sense of \Cref{sec:hdiscrete}. Its Shannon entropy is infinite \citep[Exer 4.3.5]{li2019introduction} and, even for posteriors where it's finite, the Shannon entropy need not be physically meaningful. For example, suppose we have an electric vehicle battery that is either full or empty, and the evidence puts equal odds on the two possibilities. This is modeled by a composite statistical ensemble: an equal mixture of a low-entropy and a high-entropy ensemble. Its Shannon entropy is roughly the average of the two, equivalent to a half-full battery. However, it's common sense that the mystery battery is \emph{not} operationally equivalent to a half-full battery.
	
	The algorithmic entropy, on the other hand, assigns the battery \emph{either} a low entropy \emph{or} a high one, depending on its particular state. Moreover, by definition, the universal prior concentrates on states with low algorithmic entropy. The definitions also imply the following link to data compression: if a string $s$ is losslessly compressed to a small size by a small algorithm, then its algorithmic entropy $K(s)$ is small. The connection between Shannon entropy and algorithmic entropy, from \Cref{sec:hdiscrete}, suggests a dual connection between Shannon negentropy and the \textbf{algorithmic compressibility} $|s|-K(s)$. By analogy, one might expect compressible strings to be exploitable as a resource.
	
	We demonstrate that this is indeed the case. As a result, compressibility is an objective thermodynamic resource that does not depend on a choice of probability distribution (as in Shannon's entropy) or coarse partition (as in Boltzmann's entropy). Compressibility depends on a choice of universal Turing machine but, as explained in \Cref{sec:hdiscrete}, the laws of physics narrow the options to within a margin that's negligible on macroscopic scales.
	
	We return to the SPCA-inspired digital Universe of \Cref{sec:applandauer}, consisting of a grid of binary data on which we are allowed to apply arbitrary bijections. Let $\{0,1\}^*$ denote the set of finite-length binary strings. Lossless compression algorithms can be used to fill a bank of zeros, which serve as an ``information battery''. To see how, start with a computable injective function $f:\{0,1\}^*\rightarrow\{0,1\}^*$, with finite worst-case blowup
	\[c := \sup_{x\in\{0,1\}^*}\left(|f(x)|-|x|\right).\]
	
	Fix the battery capacity $C\gg c$. Assuming $|x| \le C$, a fixed-length binary encoding of $|x|$ takes $1+\lfloor\log C\rfloor$ bits. Therefore, for all such $x$, the concatenation $e(x):=(|x|, x, 0^{C-|x|})$ takes exactly $1+\lfloor\log C\rfloor+C$ bits. By the injectivity of $f$, there exists a bijection $g:\{0,1\}^{1+\lfloor\log C\rfloor+C}\rightarrow\{0,1\}^{1+\lfloor\log C\rfloor+C}$ such that, for all $x$ satisfying $|x|\le C-c$,
	\[g(e(x)) := e(f(x)).\]
	
	In words, $g$ applies the compression algorithm $f$ to the battery, changing its state from $e(x)$ to $e(f(x))$. Since $C$ is fixed, the prefix $|x|$ indicates the battery's \textbf{charge level} $C-|x|$, i.e., the number of known zeros in the suffix of $e(x)$. Even if there are additional zeros at the end of $x$, these cannot be used, because the battery's usage is controlled on its charge indicator. Moreover, since $g$ has undefined behavior when $|x| > C-c$, the battery should be considered dead if the charge ever falls below $c$. Let's see how an agent can recharge and consume this battery.
	
	Suppose the current battery state is $e(y)$. To recharge, the agent ``hunts'' its environment for a source of candidate strings to compress. After finding a fresh string $x$ with $|x|\le |y|$, the agent ``eats'' it by reversibly swapping $x$ with a substring of $y$ in the battery. Then, the agent ``digests'' the string by applying $g$. In the case where $|x|=|y|$, the resulting state transformation is summarized as follows:
	\[e(y)\xrightarrow{\mathrm{eat}}e(x)\xrightarrow{\mathrm{digest}}e(f(x))\]
	
	The (hopefully compressed) output $f(x)$ is treated as waste, to be swapped out in the next application of this procedure. If $f$ successfully compresses $x$, the battery's charge increases; in the worst case, it decreases by up to $c$. Hence, the agent's performance depends on how well $f$ compresses the strings that it finds. The algorithmic entropy corresponds to ideal compression, in the sense that $K(x) < |f(x)| + O(1)$ \footnote{Strictly speaking, since the domain of $f$ is not restricted to a prefix-free language, $K$ here must refer to the plain algorithmic entropy, not the prefix-free version \cite[\S2]{li2019introduction}. ``$O(1)$'' represents an additive constant that depends on $f$ but not on $x$. While one may be tempted to cheat this bound by making the constant very large, that would require a very complex $f$, which must itself be physically encoded somewhere. In that case, it's more instructive to view the compression as acting jointly on the pair: $(f,x)\mapsto (f,f(x))$.}.
	
	The battery's charge can be consumed to perform various manipulations. For example, to clear a variable in some memory, we simply swap it with a prefix of the battery's bank of zeros, simultaneously incrementing the indicator $|x|$ by the length of the swapped strings. To reproduce, the agent first clears a section of its environment in this way. Then, it reversibly XOR-assigns a copy of itself onto this section. More generally, the battery can be used to emulate non-injective (i.e., many-to-one) computations as bijections. This is necessary for error correction, e.g., to repair physical damage to the agent, as multiple possibly-``bad'' states must be mapped to a smaller set of ``good'' states. As we saw in \Cref{sec:applandauer}, an information battery is also helpful for sensing and irreversible computing.
	
	We close with a general comment on the thermodynamics of life. Since the ability to sense, repair, and reproduce is so closely tied to negentropy (i.e., compressibility), life has a natural drive to collect as much negentropy as possible. Ideally, negentropy would be reversibly collected from the environment. However, from a living agent's perspective, the environment's negentropy is an economic \textbf{externality} for which the agent does not suffer fully. When presented with an opportunity to take negentropy, even at the cost of destroying negentropy elsewhere, the agent is driven to take advantage of it; and the higher the potential for gain, the higher the tolerance for waste. Thus, it's quite natural for life to produce a lot of entropy, \emph{not} out of an explicit drive to do so, but as a result of inefficient efforts to acquire more negentropy for itself.
	
	Advanced life can cooperate to internalize, or otherwise control, some of these externalities. To cite a famous example, multi-cellular organisms have evolved mechanisms to resist the evolution of cancer cells, whose short-term success would ultimately destroy the organism \citep{casas2011cancer}. As a larger-scale example, human societies negotiate systems of law and ethics.
	
	\subsection{The Paradoxes of Time}
	\label{sec:appparadox}
	
	Arguably, the air of mystery surrounding the arrow of time is owed to the abundance of emergent phenomena that appear to conflict with the underlying physics. While these \textbf{paradoxes} present no formal contradiction, they nonetheless demonstrate cases where our intuitive conclusions go astray, implying gaps in our understanding. Without making any claims to rigor, we now make an effort to mend these gaps by examining five such paradoxes.
	\begin{enumerate}
		\item \textbf{Boltzmann brains:} in the vast majority of Universe configurations consistent with our current thoughts and senses, we are nothing more than ephemeral thermal fluctuations.
		\item \textbf{Loschmidt's paradox:} the laws of physics have CPT symmetry, and yet we experience asymmetries with respect to time reversal, most notably the second law of thermodynamics.
		\item \textbf{Information acquisition:} the negentropy decreases with time, and yet we seem to learn more as time passes; typically, we only know about an event or discovery \emph{after} it occurs.
		\item \textbf{Darwinian evolution:} life seems to prefer order, collecting negentropy for itself, and yet proceeds in the temporal direction that favors \emph{decreasing} negentropy.
		\item \textbf{Dualism / free will:} we feel free to make the choices we desire, despite being subject to laws of physics that are either deterministic or probabilistically determined.
	\end{enumerate}
	
	Earlier in this article, we resolved the first two paradoxes. The Boltzmann brain paradox's resolution in \Cref{sec:appinduction} justifies the use of probabilities in a Past Hypothesis. A suitable Past Hypothesis then breaks the symmetry between past and future, resolving Loschmidt's paradox. By building on this idea, SPCAs model the emergent thermodynamic and causal asymmetries associated with time, as summarized at the end of \Cref{sec:spca}.
	
	Now, we turn to the information acquisition paradox. To be concrete, ``information'' in this context refers to either of two distinct concepts: (1) knowledge about a distant event in the past or present, as quantified by the \emph{mutual information}; and (2) abstract knowledge obtained by a difficult computation, as quantified by Bennett's \emph{logical depth}. It's no coincidence that these concepts correspond precisely to two of the three mechanisms for ``hiding'' entropy (and negentropy), listed at the end of \Cref{sec:hdiscrete}. They represent a form of negentropy that's almost lost, in a sense.
	
	For example, consider the acquisition of mutual information. Sensational news spreads like wildfire, its content being copied onto numerous listeners' memories, perhaps by XOR-assigning onto blank sections. The negentropy of a blank section provides the raw resource needed to collect memories. Although the copying is theoretically reversible, to uncopy requires careful coordination with the memory's original source. In practice, such opportunities are scarce (cf. the sensing example at the end of \Cref{sec:applandauer}), leaving dissipation as the primary means of clearing a memory for future use. 
	
	The same is true for logical depth: while reversible computing methods can in principle ``uncompute'' a deep string, to do so requires careful coordination as well as a repeated computation. In most settings, the only practical option is to dissipate.
	
	Note that all we've done here is explain why information acquisition requires (and typically consumes) negentropy, \emph{assuming} it operates along the arrow of time. The latter assumption was implicit in the need to use the macroscopically reversible XOR-assignment operation (e.g., using \Cref{thm:equivalence1}). The fact that information acquisition does indeed follow the arrow of time was previously addressed in \Cref{sec:intromemory} for the mutual information, and in \Cref{sec:relatedtime} for the logical depth.
	
	The latter argument about logical depth also helps explain the paradox of Darwinian evolution. Modern genomes are logically deep, as a result of billions of years optimizing for survival. Assuming the Universe has shallow initial conditions and Markovian dynamics, Bennett's slow growth law \citep{bennett1988logical} therefore implies that life must have had a long history.
	
	In addition, evolution's two primary mechanisms seem to depend on the dynamics being Markovian and spacetime-homogeneous: (1) genetic mutations must be sampled according to a fairly consistent distribution, and (2) natural selection must progress toward fairly consistent objectives. As a thought experiment, consider the trajectory of the world's first microorganism, during the time interval $[0,100]$. Let $A_t$ be the event that the organism is alive at time $t$. Then, the probability of survival over the whole interval is
	\begin{equation}
		\label{eq:abiogenesis}
		\prob(A_0,A_1,\ldots,A_{100}) = \prob(A_0)\prod_{t=0}^{99}\prob(A_{t+1}\mid A_t).
	\end{equation}
	
	Technically, it's equally valid to decompose this expression in terms of backward probabilities $\prob(A_{t-1}\mid A_t)$; however, the latter are inhomogeneous and may conspire in strange ways. The forward probabilities $\prob(A_{t+1}\mid A_t)$, being derived from homogeneous dynamics, are much easier to estimate. Now, suppose the organism is built to survive ``forward in time''. To compensate for damage incurred from noisy interactions with the environment, the organism may consume negentropy found in the environment to perform error correction. Thus, the probabilities $\prob(A_{t+1}\mid A_t)$ are large, and overall the trajectory is not too much less likely than the abiogenesis event $A_0$.
	
	On the other hand, suppose the organism is built to survive ``backward in time''. One might imagine performing the same error connection procedure to sustain its existence, backward from $t=100$ to $t=0$. To estimate such a trajectory's probability, we still use \Cref{eq:abiogenesis}. Seen forward in time, the error correction becomes error \emph{production}. That is, the agent becomes self-destructive, while the noisy environment must conspire to undo the destruction. Naturally, $\prob(A_{t+1}\mid A_t)$ becomes very low. The situation is even worse if the organism reproduces, requiring the initial event $A_0$ to account for the simultaneous appearance of multiple copies.
	
	In hindsight, paradoxes 3 and 4 arise from attempts to explain all asymmetries in terms of the second law of thermodynamics. Instead, the SPCA model suggests taking the local homogeneous Markovian transitions as fundamentally responsible, for both the thermodynamic and psychological arrows of time. Framed this way, it's no longer surprising to find that the arrows align.
	
	Finally, let's address the nature of free will. There is an ancient idea in philosophical tradition called \textbf{mind-body dualism}; some versions propose that minds are capable of causally influencing the future, along mechanisms independent from the rest of physics. Although this view has largely fallen out of favor, a pragmatic form of dualism persists in some of the most useful models of agent (e.g., human or robot) behavior. Why do such abstractions have explanatory power in the physical Universe?
	
	For example, \textbf{causal decision theory (CDT)} \citep{gibbard1978counterfactuals} models the relevant set of variables on a causal graph. Vertices corresponding to unintelligent physical processes are modeled in the usual probabilistic manner, conditional on their parents. However, CDT gives special treatment to \textbf{decision nodes}, which are vertices whose realization is \emph{chosen} by intelligent agents. While other vertices have their dynamics endogenized and made immutable, each decision node is treated as exogenous and associated with some agent. To derive the node's value, one must consider the counterfactual result of each possible value, and select one that optimizes the agent's objective.
	
	Decision nodes appear to serve as a useful shortcut to whatever physical processes really underlie an agent's behavior. In order for them to work, it should be the case that Darwinian evolution creates agents that behave according to CDT. This seems plausible: among all possible agents, the one whose behavior follows CDT makes the optimal counterfactual come to fruition, and therefore is naturally selected for. Sufficiently intelligent life can compute CDT, mentally simulating each counterfactual to choose the preferred one. Such an agent would likely appear to us to have free will.
	
	Two remarks are in order. Firstly, structural counterfactuals can also model settings without intelligent agents: in \Cref{sec:introcounterfactual}, a local model used counterfactuals to account for an exogenous influence. Secondly, our argument for CDT here is not sufficiently rigorous to verify in full generality. Future work should investigate the precise conditions under which counterfactuals and free will, as modeled by CDT, arise as emergent features of an SPCA-like Universe.
	
	Indeed, there may exist important counterexamples. \citet{yudkowsky2017functional} consider a class of \emph{Newcomb-like} problems, in which causal decision theorists tend to underperform. In these multi-agent scenarios, each agent's decision depends on its prediction of other agents' decisions. The authors provide some plausible human examples, but their concerns are sharpened in the context of advanced artificial intelligence.
	
	For example, suppose we design a highly intelligent robot. Before letting it roam the world freely, we decide to test its trustworthiness inside a simulation. The robot, having deduced this, must at all times consider the possibility that it might currently be in a simulation. In order to survive, its behavior should appear favorable to the simulation's human operator. Arguably, the simplest way for a causal model to account for such an influence is with an edge directed backward in time: even after the robot is free, its decisions can affect whether the robot attains freedom!
	
	By developing similar thought experiments in more detail, \citeauthor{yudkowsky2017functional} find CDT to be inconsistent in the following sense: given the ability to do so, a causal decision theorist would prefer to rewire its own brain so as to cease to follow CDT. The authors' proposed alternative, \textbf{functional decision theory (FDT)}, is still based on Pearlean causal graphs, but its edges need not always align with the physical arrow of time.
	
	If this seems absurd, one should bear in mind that neither CDT nor FDT are written into the known laws of physics. CDT already depends on counterfactuals which, if taken literally, constitute violations of physics. In our view, decision theories should primarily be judged on their ability to parsimoniously predict the behavior of Darwinian evolution's winners. On the road toward a formal justification of any decision theory, we propose that SPCAs, as well as the thought experiments of \citeauthor{yudkowsky2017functional}, can serve as helpful conceptual tools.
	
	
	\section{Conclusions}
	\label{sec:conclusions}
	
	In order to model the emergence of irreversibility, we introduce the stochastic partitioned cellular automaton (SPCA). Its microscopic description shares many properties with classical dynamical systems: it is deterministic, chaotic, reversible, and measure-preserving. Its trajectory is determined by a choice of initial condition (i.e., Past Hypothesis); this is where the symmetry is broken.

	As a result, we find a macroscopic view that is local, homogeneous in space and time, and Markovian, but only when its dynamics are described along an arrow of time that is directed away from the initial condition. The conditions of being local and Markovian are precisely summarized by Pearl's $d$-separation criterion, relative to a causal graph whose edges are timelike and future-directed. Thus, an SPCA is a Pearlean causal model with a certain homogeneous structure.
	
	The microscopic and macroscopic views inherit features from one another. For example, if the microscopic dynamics preserve the measure $\pi\times\Gamma$, then $\pi$ is stationary for the macroscopic dynamics. In addition, we can take advantage of the microscopic description's reversibility to extrapolate an SPCA to negative times; a duality relation holds between the macroscopic transition probabilities at positive and negative times.
	
	Within the setting of SPCAs, we study the dynamical evolution of entropy, negentropy, and mutual information. Our findings are summarized by the Resource and Memory Laws, which generalize the second law of thermodynamics. In particular, the Memory Law plays a role in explaining records of past events, as well as in extending Landauer's principle.
	
	As a mathematical convenience, we rely heavily on the Shannon theory of information. On the other hand, we find that entropy is more generally and concretely understood in algorithmic terms. Just as the opposite of entropy is negentropy, the opposite of algorithmic entropy is algorithmic compressibility. Thus, underlying all physical notions of free energy, we propose that the fundamental information-theoretic resource is compressibility. To support this hypothesis, we describe a rudimentary agent powered by algorithmic string compression.
	
	The empirical side of SPCAs remains unexplored. In future work, one might design SPCAs with convenient properties, such as computational universality or a sparse causal structure. By simulating these, one might demonstrate aspects of irreversible phenomena, such as memory acquisition or Darwinian evolution. Conjectures derived from simulations may later be tested in the real world or proved in general.
	
	We hope that such investigations will yield further insight on a variety of topics, potentially ranging from transistor design to the philosophy of mind. To give one specific possibility, consider a binary counter programmed to reversibly increment at fixed time intervals. If initialized to zero, it provides an accurate measurement of elapsed time, consuming no more negentropy than is needed to store a binary representation of the number of ticks. If such a counter can be implemented in the real world, it would represent a drastic efficiency gain over the thermodynamic clocks of \citet{pearson2021measuring}. On the other hand, if such a counter turns out to be infeasible, understanding the limitation can lead to better models.
	
	There are many interesting differences between SPCAs and modern field theories, and it remains an open question to what extent they affect the conclusions of this article. Differences worth examining include conservation laws, continuous spacetimes, quantum information, and more general forms of microscopic dynamics. Perhaps analogous laws will be discovered in each of these settings, but with additional nuances.
	
	Finally, our conceptual tools have the potential to impact other fields. Part of the SPCA's modeling power comes from its ability to endogenize information-processing agents, such as human experimenters. We've only briefly outlined the implications of this approach. In future work, endogenous agents can be modeled in full detail; this may lead to additional insights on all the topics we touched, from reversible computing to decision theory.
	
	\begin{acknowledgements}
		
		I am grateful to everyone with whom I've had the chance to talk. Deep inquiries from Jason Li helped to refine some key arguments. Xianda Sun's review of the manuscript pointed out many passages in need of clarification. Thanks also to Danica Sutherland, David Wakeham, Paul Liu, and William Evans for related conversations and advice.
		
	\end{acknowledgements}
	
	\appendix
	
	\section{Proof of Theorem 3}
	\label{sec:appendixproof}
	
	Before proving \Cref{thm:equivalence2}, let's develop a strategy to address its most challenging claim:
	
	\emph{If $p\in\mathrm{SM}_\pi(\mathcal S)$, with $\pi$ strictly positive and finite-sided, then $T$ can be chosen to be bijective and $(\pi\times\Gamma)$-measure-preserving}.
	
	The argument is more straightforward when $\pi=\sharp$. By \Cref{thm:equivalence1}, the dynamics $p\in\mathrm{SM}_\sharp(\mathcal S)=\mathrm{DM}(\mathcal S)$ can be represented by a sequence of i.i.d. random bijections $(F_0,F_1,\ldots)$, with $F_t$ being applied at time step $t$. By absorbing this sequence into the initial state $C_0$, the dynamics can be made deterministic. To be specific, the dynamics interleave two steps: (1) the bijection $(c,f)\mapsto (f(c),f)$, applied to the pair consisting of the macrostate and first bijection; and (2) the shift map, applied to the sequence of bijections to move its next element into place.
	
	More generally, we are interested in $p\in\mathrm{SM}_\pi(\mathcal S)$ with $\pi\in(\frac 1m\mathbb N)^\mathcal S$. In order to obtain random bijections using \Cref{thm:equivalence1}, we need a related Markov chain that is doubly stochastic. Its construction is fairly intuitive: we think of the positive integer $m\cdot\pi(s)$ as the ``size'' of each state $s\in\mathcal S$. We split $s$ into unit-sized pieces, such that any time spent in $s$ would instead be uniformly distributed among its pieces. The result is our desired dynamics $p_\mathrm{dup}\in\mathrm{DM}(\mathcal S_\mathrm{dup})$, on a new state space $\mathcal S_\mathrm{dup}$, consisting of $m\cdot\pi(s)$ duplicates of each $s\in\mathcal S$.
	
	\begin{definition}
		\label{def:dup}
		Fix $m\in\mathbb N$ and $\pi\in(\frac 1m\mathbb N)^\mathcal S$. The $m\pi$\textbf{-duplication} of the Markov chain parameters $(\mathcal S,\mu,p)$, with $p\in\mathrm{SM}_\pi(\mathcal S)$, is the triple  $(\mathcal S_\mathrm{dup},\mu_\mathrm{dup},p_\mathrm{dup})$, defined by
		\begin{align*}
			\mathcal S_\mathrm{dup} &:= \{(s,i):\; s\in\mathcal S,\;i\in\mathbb Z_{m\pi(s)}\},
			\\\mu_\mathrm{dup}((s,i)) &:= \frac{\mu(s)}{m\cdot\pi(s)}\quad\text{for }(s,i)\in\mathcal S_\mathrm{dup},
			\\p_\mathrm{dup}((s,i),(s',i')) &:= \frac{p(s,s')}{m\cdot\pi(s')}\quad\text{for }(s,i),(s',i')\in\mathcal S_\mathrm{dup}.
		\end{align*}
	\end{definition}
	
	It's straightforward to verify that $\Phi_\mathrm{matrix}(\mu, p)$ and $\Phi_\mathrm{matrix}(\mu_\mathrm{dup}, p_\mathrm{dup})$ (with the corresponding state spaces $\mathcal S$ and $\mathcal S_\mathrm{dup}$ implied) are identically distributed, if elements of $\mathcal S_\mathrm{dup}$ are identified with corresponding elements of $\mathcal S$. Furthermore, $p_\mathrm{dup}\in\mathrm{DM}(\mathcal S_\mathrm{dup})$, because
	\begin{align*}
		\sum_{(s,i)\in\mathcal S_\mathrm{dup}} p_\mathrm{dup}((s,i),(s',i'))
		&= \sum_{(s,i)\in\mathcal S_\mathrm{dup}} \frac{p(s,s')}{m\cdot\pi(s')}
		\\&= \sum_{s\in\mathcal S} m\cdot\pi(s)\cdot\frac{p(s,s')}{m\cdot\pi(s')}
		\\&=  \frac{\sum_{s\in\mathcal S}\pi(s)\cdot p(s,s')}{\pi(s')}
		\\&= \frac{\pi(s')}{\pi(s')}
		\\&= 1,\text{ and}
		\\\sum_{(s',i')\in\mathcal S_\mathrm{dup}} p_\mathrm{dup}((s,i),(s',i'))
		&= \sum_{(s',i')\in\mathcal S_\mathrm{dup}} \frac{p(s,s')}{m\cdot\pi(s')}
		= \sum_{s'\in\mathcal S} p(s,s')
		=1.
	\end{align*}
	
	With this construction in hand, we are ready to prove \Cref{thm:equivalence2}.
	
	\begin{proof}
		
		Uniqueness of $p$ follows as in the proof of \Cref{thm:equivalence1}. To prove the first implication, fix $(\mathcal R,\mathcal G,\Gamma,T)\in\mathrm{Dyn}(\mathcal S)$. For $s,s'\in\mathcal S$, define the sets
		\begin{align*}
			A(s,s') &:= \{a\in\mathcal R:\; \exists a'\in\mathcal R,\; T(s,a)=(s',a')\},
			\\\text{and let }p(s,s') &:= \Gamma(A(s,s')).
			\\\text{Then, }\sum_{s'\in\mathcal S}p(s,s') &= \Gamma\left(\bigcup_{s'\in\mathcal S} A(s,s')\right) = \Gamma(\mathcal R) = 1,
		\end{align*}
		
		because the sets $A(s,s')$, with fixed $s$, form a partition of $\mathcal R$. Hence, $p\in\mathrm{SM}(\mathcal S)$. Furthermore, if $T$ is $(\pi\times\Gamma)$-measure-preserving, then
		\begin{align*}
			\sum_{s\in\mathcal S}\pi(s) p(s,s')
			&= \sum_{s\in\mathcal S} (\pi\times\Gamma)(\{s\}\times A(s,s'))
			\\&= (\pi\times\Gamma) \left(\bigcup_{s\in\mathcal S}(\{s\}\times A(s,s'))\right)
			\\&= (\pi\times\Gamma) \left(\{(s,a)\in\mathcal S\times\mathcal R:\; \exists a'\in\mathcal R,\; T(s,a)=(s',a')\}\right)
			\\&= (\pi\times\Gamma) (T^{-1}({\{s'\}\times\mathcal R}))
			\\&= (\pi\times\Gamma) (\{s'\}\times\mathcal R)
			\\&= \pi(s')\Gamma(\mathcal R)
			\\&= \pi(s'),
		\end{align*}
		
		so $p\in\mathrm{SM}_\pi(\mathcal S)$. In addition, if $\Gamma$ is $m$-sided, then so is $p$ by its definition.
		
		Now, fix $\mu\in\pmeas(\mathcal S)$. Let $(\mathbf C,\mathbf R)$ be as in \Cref{def:markovd}, so that $\prob_{\mathbf C}=\Phi_\mathrm{determ}(\mu,\mathcal R,\mathcal G,\Gamma,T)$. By definition, $\prob_{C_0}=\mu$, so \Cref{eq:markovm1} holds. Furthermore, since $R_{t,0} = R_{0,t}$ is independent of $C_{\le t}$,
		\begin{equation*}
			\prob(C_{t+1}=c_{t+1} \mid C_{\le t}=c_{\le t})
			= \prob(R_{0,t} \in A(c_t,c_{t+1}))
			= \Gamma(A(c_t,c_{t+1}))
			= p(c_t,c_{t+1}),
		\end{equation*}
		
		so \Cref{eq:markovm2} holds as well. Therefore, $\Phi_\mathrm{matrix}(\mu, p) = \Phi_\mathrm{determ}(\mu,\mathcal R,\mathcal G,\Gamma,T)$.
		
		To prove the converse, fix $p\in\mathrm{SM}(\mathcal S)$. Define $T:\mathcal S\times\mathcal S^\mathcal S\rightarrow\mathcal S\times\mathcal S^\mathcal S$ by
		\[T(s,f):=(f(s),f).\]
		
		We must verify that $T$ is $(\powerset{\mathcal S}\times\mathcal G)$-measurable, where the implied $\sigma$-algebra $\mathcal G\subset\powerset{\mathcal S^\mathcal S}$ is generated by cylinder sets of the form
		\[B=\{f\in\mathcal S^\mathcal S:\;f(s_1)=s_1',\,\ldots,\,f(s_n)=s_n'\}\in\mathcal G,\]
		
		with $n\in\mathbb N$ and $s_i,s_i'\in\mathcal S$. The product $\sigma$-algebra $\powerset{\mathcal S}\times\mathcal G$ is generated by sets of the form $\{s_0'\}\times B$. Measurability now follows by observing that the preimage
		\[T^{-1}(\{s_0'\}\times B)
		= \bigcup_{s_0\in\mathcal S}\left(\{s_0\}\times\{f\in\mathcal S^\mathcal S:\;f(s_0)=s_0',\,f(s_1)=s_1',\,\ldots,\,f(s_n)=s_n'\}\right)\]
		is itself a countable union of cylinder sets, hence is in $\powerset{\mathcal S}\times\mathcal G$.
		
		By \Cref{thm:equivalence1}, there exists $\Gamma\in\pmeas(\mathcal S^\mathcal S)$ such that $\Phi_\mathrm{matrix}(\cdot, p) = \Phi_\mathrm{random}(\cdot,\Gamma)$. Now, $(\mathcal S^\mathcal S,\mathcal G,\Gamma,T)\in\mathrm{Dyn}(\mathcal S)$ and, since
		\[A(s,s') = \{f\in\mathcal S^\mathcal S:\; \exists f'\in\mathcal S^\mathcal S,\; T(s,f)=(s',f')\} = \{f\in\mathcal S^\mathcal S:\; f(s)=s'\},\]
		
		we have $\Gamma(A(s,s_i)) = p(s,s_i)$. This equality implies, by the same argument used in the forward implication, that $\Phi_\mathrm{matrix}(\cdot, p) = \Phi_\mathrm{determ}(\cdot,\mathcal S^\mathcal S,\mathcal G,\Gamma,T)$.
		
		In the case where $p$ is $m$-sided, the $\Gamma$ that we obtained using \Cref{thm:equivalence1} can be taken to be $m$-sided. That is, $\Gamma$ is the uniform distribution over some multiset $\{f_0,f_1,\ldots,f_{m-1}\}$, which becomes the fair $m$-sided die after encoding each function $f_i$ by its index $i$. In terms of this encoding, the map $T:\mathcal S\times\mathbb Z_m\rightarrow\mathcal S\times\mathbb Z_m$ becomes
		\[T(s,i) := (f_i(s),i).\]
		
		If we also have $p\in\mathrm{DM}(\mathcal S)$, \Cref{thm:equivalence1} allows choosing the functions $f_i$ to be bijections, making $T$ bijective as well.
		
		Finally, we address the difficult case, where $p\in\mathrm{SM}_\pi(\mathcal S)$ for some $\pi\in(\frac 1m\mathbb N)^\mathcal S$. Consider the $m\pi$-duplicated dynamics $p_\mathrm{dup}\in\mathrm{DM}(\mathcal S_\mathrm{dup})$; by \Cref{thm:equivalence1}, its random function presentation $\Gamma$ can be taken to be supported on $\mathrm{Bij}(\mathcal S_\mathrm{dup})$. Now for each $s\in\mathcal S$, split the interval $[0,1)$ into $m\cdot\pi(s)$ equal-sized sub-intervals $I_{s,i}:=[\frac {i}{m\cdot\pi(s)},\frac {i+1}{m\cdot\pi(s)})$, one for each duplicate $(s,i)\in\mathcal S_\mathrm{dup}$ of $s$. For each $(s,i)\in\mathcal S_\mathrm{dup}$ and $f\in\mathrm{Bij}(\mathcal S_\mathrm{dup})$, denote $(s',i'):=f((s,i))$ to define the bijection
		
		\begin{align*}
			&T_{s,i,f}:\{s\}\times I_{s,i}\times\{f\}
			\rightarrow \{s'\}\times I_{s',i'}\times\{f\} \quad\text{by}
			\\&T_{s,i,f}\left(s,\,\frac {i+r}{m\cdot\pi(s)},\,f\right)
			:= \left(s',\; \frac {i'+r}{m\cdot\pi(s')},\; f\right)\;\;\forall r\in\left[0,1\right).
		\end{align*}
		
		The collection of $T_{s,i,f}$'s have disjoint domains and disjoint ranges. By joining them all together, we obtain one $(\pi\times\lambda\times\Gamma)$-measure-preserving bijection
		\[T:\mathcal S\times [0,1)\times\mathrm{Bij}(\mathcal S_\mathrm{dup})\rightarrow\mathcal S\times [0,1)\times\mathrm{Bij}(\mathcal S_\mathrm{dup}).\]
		
		Let $\mathcal B\subset\powerset{[0,1)}$ denote the Borel $\sigma$-algebra generated by the subintervals of $[0,1)$, and let $\mathcal G\subset\powerset{\mathrm{Bij}(S_\mathrm{dup})}$ be generated by the cylinder sets. It remains to show that, for all $\mu\in\pmeas(\mathcal S)$, \[\Phi_\mathrm{matrix}(\mu,\, p) = \Phi_\mathrm{determ}(\mu,\,[0,1)\times\mathrm{Bij}(\mathcal S_\mathrm{dup}),\,\mathcal B\times\mathcal G,\,\lambda\times\Gamma,\,T).\]
		
		As usual, this amounts to verifying the state transition probabilities. For the transition matrix presentation, they are given simply by $p(s,s')$. For the deterministic function presentation, a metaphorical die $(r,f)\in[0,1)\times\mathrm{Bij}(\mathcal S_\mathrm{dup})$ is sampled from $\lambda\times\Gamma$, and plugged into $T$. Then, the state $s$ transitions to $s'$ iff $T(s,r,f) = (s',r',f)$ for some $r'$. The probability of this event is
		\begin{align*}
			&(\lambda\times\Gamma)\left(\left\{(r,f):\; \exists r'\in[0,1),\; T(s,r,f) = (s',r',f)\right\}\right)
			\\&= \sum_{i\in\mathbb Z_{m\pi(s)}} \lambda\left(\Big[\frac{i}{m\pi(s)},\frac{i+1}{m\pi(s)}\Big)\right)\cdot \Gamma\left(\left\{f:\; \exists r'\in[0,1),\; T(s,i/m\pi(s),f) = (s',r',f)\right\}\right)
			\\&= \frac{1}{m\pi(s)}\sum_{i\in\mathbb Z_{m\pi(s)}} \Gamma\left(\left\{f:\; \exists r'\in[0,1),\; T(s,i/m\pi(s),f) = (s',r',f)\right\}\right)
			\\&= \frac{1}{m\pi(s)}\sum_{i\in\mathbb Z_{m\pi(s)}}\sum_{i'\in\mathbb Z_{m\pi(s')}} \Gamma\left(\left\{f:\; T(s,i/m\pi(s),f) = (s',i'/m\pi(s'),f)\right\}\right)
			\\&= \frac{1}{m\pi(s)}\sum_{i\in\mathbb Z_{m\pi(s)}}\sum_{i'\in\mathbb Z_{m\pi(s')}} \Gamma\left(\left\{f:\; f((s,i)) = (s',i')\right\}\right)
			\\&= \frac{1}{m\pi(s)}\sum_{i\in\mathbb Z_{m\pi(s)}}\sum_{i'\in\mathbb Z_{m\pi(s')}} p_\mathrm{dup}((s,i),\,(s',i'))
			\\&= \frac{1}{m\pi(s)} \cdot m\pi(s) \cdot m\pi(s') \cdot \frac{p(s,s')}{m\pi(s')}
			\\&= p(s,s').
		\end{align*}
	\end{proof}
	
	\section{Breaking the Symmetry Between Space and Time}
	\label{sec:appendixspacetime}
	
	An SPCA's deterministic presentation provides microscopic dynamics that are symmetric between past and future, but asymmetric between space and time. Throughout this article, we've taken it for granted that influences travel along the timelike edges of $\mathrm{G}(\mathcal X,\mathcal T)$, and sought only to orient these edges forward or backward by imposing an appropriate Past Hypothesis. There was never any question of a ``sideways'' arrow of time because, by design, influences were not allowed to travel along spacelike paths.
	
	On the other hand, in the theory of general relativity, Einstein's field equations treat spacetime as a four-dimensional Lorentzian manifold. While the roles of space and time are not \emph{completely} symmetric in this theory, the distinction between the two is more subtle. It's natural to wonder why causality, in the Pearlean sense, is timelike.
	
	In this appendix, we propose that the initial condition serves to break the symmetry, not only between past and future, but also between space and time. To demonstrate this, we construct a new symbolic dynamics which, unlike an SPCA's, treats all the axes symmetrically. The time axis is only identified by the initial condition. We will find that the resulting model is Markov, relative to a causal graph whose edges are directed along this axis.
	
	For concreteness, fix the \textbf{dimension} $d\in\mathbb N$ of our discrete spacetime $\mathbb Z^d$. For $i\in\{1,\ldots,d\}$, let $e_i\in\mathbb Z^d$ denote the $i$'th basis vector, i.e., $(e_i)_j := \delta_{ij}$. Let $\mathbf 0$ denote the zero vector. Let
	\[\mathcal N_1 := \{e_1,-e_1,e_2,-e_2,\ldots,e_d,-e_d\},\]
	
	consist of all unit-length vectors in $\mathbb Z^d$, and $\mathcal N := \mathcal N_1\cup\{\mathbf 0\}$. Given a set $\mathcal S$, an \textbf{$\mathcal S$-relation} is any subset $P\subset \mathcal S^\mathcal N$. We say $P$ is \textbf{deterministic} if, for all $n\in\mathcal N_1$ and $g:\mathcal N\setminus \{n\}\rightarrow\mathcal S$, the following equivalent conditions hold:
	\begin{itemize}
		\item There exists a unique $f\in P$ whose restriction to $\mathcal N\setminus \{n\}$ equals $g$.
		\item There exists a unique $s\in\mathcal S$ such that the function $f := g_{n\leftarrow s}$ is in $P$.
	\end{itemize}
	
	For example,
	\begin{equation}
		\label{eq:relmod}
		P_{\mathrm{mod}\,m} := \left\{f\in(\mathbb Z_m)^\mathcal N:\; \sum_{n\in\mathcal N}f(n) \equiv 0\text{ (mod }m)\right\}.
	\end{equation}
	
	is a deterministic $\mathbb Z_m$-relation, because each $f\in P_{\text{mod }m}$ is uniquely determined by specifying $f$ at all but one point of its domain $\mathcal N$. We denote the set of all deterministic $\mathcal S$-relations by $\mathrm{Rel}(\mathcal S)\subset\powerset{\mathcal S^\mathcal N}$. As the name suggests, a deterministic relation can take the role of a system's dynamics, to evolve it uniquely from some initial condition.
	
	Our system is a kind of second-order cellular automaton. Begin by fixing a time axis $i\in\{1,\ldots,d\}$; in other words, for arbitrary spacetime coordinates $x\in\mathbb Z^d$, $x_i$ takes the role of time. We set initial conditions that specify the state at not just one, but two adjacent times, say $x_i=0$ and $x_i=1$. That is, an initial state is specified at each point of the \textbf{initial set}
	\begin{align}
		\label{eq:initialset}
		I:=\{x\in\mathbb Z^d:\; x_i\in\{0,\,1\}\}.
	\end{align}
	
	Given an initial condition $U:I\rightarrow\mathcal S$ and a deterministic relation $P\in\mathrm{Rel}(\mathcal S)$, there exists a unique configuration $\mathbf C:\mathbb Z^d\rightarrow\mathcal S$ satisfying
	\begin{align}
		\label{eq:sym1}
		C_x &= U_x, &&\forall x\in I,
		\\\label{eq:sym2}
		(n\mapsto C_{x+n})&\in P &&\forall x\in\mathbb Z^d.
	\end{align}
	
	To see this, note that \Cref{eq:sym1} determines the state at times $0$ and $1$. Since $P$ is deterministic, \Cref{eq:sym2} uniquely determines $C_{x+e_i}$ as a function of the state at all other neighbors of $x$. Therefore, given the states at times $t-1$ and $t$, applying \Cref{eq:sym2} at all $x\in\mathbb Z^d$ with $x_i=t$ uniquely determines the states at time $t+1$. By induction on $t$, we conclude that for all $x\in\mathbb Z^d$ with $x_i\ge 1$, $C_x$ is determined and \Cref{eq:sym2} is satisfied. The case $x_i \le 0$ is symmetric, since $C_{x-e_i}$ is also uniquely determined as a function of all the other neighbors.
	
	At this point, we make two remarks. Firstly, the only place where $i$ entered \Cref{eq:sym1,eq:sym2} was in the definition of the initial set $I$. Thus, it is the initial condition, not the dynamics, that gave time its distinguished role. Secondly, we can allow $U$ to be randomly distributed. For every realization of $U$, \Cref{eq:sym1,eq:sym2} uniquely determine a realization of $\mathbf C$; thus, the random variable $\mathbf C$ is uniquely determined.
	
	Mirroring our work on SPCAs, let's now enhance the state with microscopic variables $\mathbf R$: at every spacetime coordinate $x\in\mathbb Z^d$, let the complete state at $x$ be given by $(C_x,R_x)\in\mathcal S\times(\mathbb Z_m)^{\mathbb Z^d}$. The initial microscopic variables $R_I=(R_{x,y})_{(x,y)\in I\times\mathbb Z^d}$ are uniform and i.i.d., while the macroscopic variables $C_I$ are drawn from an arbitrary distribution $\mu$.
	
	\begin{definition}
		For fixed $\mathcal S$ and $d,i,m\in\mathbb N$ with $i\le d$, let $I$ be given by \Cref{eq:initialset}. Define the map
		\[\Phi_\mathrm{rel\text-determ}:\pmeas(\mathcal S^I)\times\mathrm{Rel}(\mathcal S\times\mathbb Z_m) \rightarrow \pmeas(\mathcal S^{\mathbb Z^d})\]
		
		as follows. Consider any \textbf{initial condition} $\mu\in\pmeas(\mathcal S^I)$ and \textbf{macroscopic dynamics} $T\in\mathrm{Rel}(\mathcal S\times\mathbb Z_m)$. Let the random variables $(\mathbf C,\mathbf R)=(C_x,R_x)_{x\in\mathbb Z^d}$ be distributed such that
		\begin{align}
			\label{eq:symicro1}
			\prob_{(C_I,\,R_I)} &= \mu\times\left(\frac 1m\sharp\right)^{I\times\mathbb Z^d},
			\\\label{eq:symicro2}
			(n\mapsto R_{x+n,\,y-n})&\in P_{\mathrm{mod}\,m} &&\forall x,y\in\mathbb Z^d,y\ne\mathbf 0,
			\\\label{eq:symicro3}
			(n\mapsto (C_{x+n},\,R_{x+n,\,-n}))&\in T &&\forall x\in\mathbb Z^d,
		\end{align}
		
		with $P_{\mathrm{mod}\,m}\in\mathrm{Rel}(\mathbb Z_m)$ defined by \Cref{eq:relmod}. Then, $\Phi_\mathrm{rel\text-determ}(\mu,\,T):=\prob_{\mathbf C}$.
	\end{definition}
	
	\Cref{eq:symicro2,eq:symicro3} can be combined into a single deterministic relation constraining $(n\mapsto (C_{x+n},R_{x+n}))$, for every $x\in\mathbb Z^d$. Thus, \Cref{eq:symicro1,eq:symicro2,eq:symicro3} are a special case of \Cref{eq:sym1,eq:sym2}, from which it follows that the random variables $(\mathbf C,\mathbf R)$ are uniquely determined.
	
	Note that $\mathbf R$ has two $\mathbb Z^d$-valued subscripts: the first is a spacetime coordinate, while the second is a $d$-dimensional generalization of the SPCA shift index. $P_{\mathrm{mod}\,m}$ takes a role analogous to the shift map, putting fresh ``dice'' into place as needed. Just as in \Cref{sec:markovdet}, indices closer to $\mathbf 0$ should be considered more significant, with $T$ acting only on the most significant indices.
	
	Alternatively, we can do away with the shifting by substituting $R'_{x,\,y+x} := R_{x,\,y}$. \Cref{eq:symicro1,eq:symicro2,eq:symicro3} then reduce to the equivalent form:
	\begin{align}
		\label{eq:symicro1p}
		\prob_{(C_I,\,R'_I)} &= \mu\times\left(\frac 1m\sharp\right)^{I\times\mathbb Z^d},
		\\\label{eq:symicro2p}
		(n\mapsto R'_{x+n,\,y})&\in P_{\mathrm{mod}\,m} &&\forall x,y\in\mathbb Z^d,x\ne y,
		\\\label{eq:symicro3p}
		(n\mapsto (C_{x+n},\,R'_{x+n,\,x}))&\in T &&\forall x\in\mathbb Z^d.
	\end{align}
	
	The macroscopic component $\mathbf C$ has an arrow of time, which we will express in terms of a sparse causal graph $G$ satisfying the following:
	\begin{itemize}
		\item The vertex set is $\{O\}\cup\mathbb Z^d$, where we identify $O$ with $C_I$, and each $x\in\mathbb Z^d$ with $C_x$.
		\item With these identifications, $\mathbf C$ is Markov relative to $\mathrm G$.
		\item On each side of $I$, $\mathrm G$ and its conditional probabilities are homogeneous in spacetime.
	\end{itemize}
	
	By symmetry, we can focus on the side with positive times: that is, restrict the vertices and spacetime indices $x\in\mathbb Z^d$ to $x_i\ge 0$, and restrict \Cref{eq:symicro2,eq:symicro3} to $x_i\ge 1$. For $t,u\in\mathbb Z\cup\{-\infty,\,\infty\}$, let
	\begin{align*}
		I(t,\,u)&:=\{x\in\mathbb Z^d:\; t\le x_i\le u\},
	\end{align*}
	
	so that $I=I(0,\,1)$. We now prove that, for all $t,u\in\mathbb Z^+$ with $t<u$,
	\begin{equation}
		\label{eq:symiid}
		\prob_{\left(C_{I(0,\,u)},\,R'_{I(t,\,t+1),\,I(u,\,\infty)}\right)}
		= \prob_{C_{I(0,\,u)}} \times
		\left(\frac 1m\sharp\right)^{I(t,\,t+1)\times I(u,\,\infty)}.
	\end{equation}
	
	In other words, the random variables $R'_{x,\,y}$, with $x_i\in\{t,\,t+1\}$ and $y_i\ge u$, are uniform on $\mathbb Z_m$, i.i.d., and jointly independent of $C_{I(0,\,t)}$.
	
	Independence of $R'_{I(t,\,t+1),\,I(u,\,\infty)}$ from $C_{I(0,\,u)}$ follows from the fact that they are functions of disjoint sets of independent variables: specifically, the former is a function of $R'_{I,\,I(u,\,\infty)}$, whereas the latter is a function of $(C_I,\,R'_{I,\,I(0,\,u-1)})$. It only remains to show that
	\begin{equation}
		\label{eq:symiid2}
		\prob_{R'_{I(t,\,t+1),\,I(u,\,\infty)}}
		= \left(\frac 1m\sharp\right)^{I(t,\,t+1)\times I(u,\,\infty)}.
	\end{equation}
	
	Fix $u\ge 1$, and proceed by induction on $t$. The base case $t=0$ is an immediate consequence of \Cref{eq:symicro1p}. For the inductive case, suppose \Cref{eq:symiid2} holds for $t-1$. For all $(x,y)\in I(t,t)\times I(u,\infty)$, we have $x_i = t < u \le y_i$, hence $x\ne y$. Therefore, \Cref{eq:symicro2p} applies at $(x,y)$. Substituting into \Cref{eq:relmod}, this becomes
	\[\sum_{n\in\mathcal N}R'_{x+n,\,y} \equiv 0\text{ (mod }m).\]
	
	Solving for the new term at time coordinate $t+1$,
	\[R'_{x+e_i,\,y} \equiv -R'_{x-e_i,\,y} - \sum_{n\in\mathcal N\setminus\{e_i,-e_i\}}R'_{x+n,\,y}\text{ (mod }m).\]
	
	Since every term in the sum has time coordinate $(x+n)_i=x_i+n_i=x_i=t$, these terms are included in $R'_{I(t,t)\times I(u,\infty)}$. Conditioning on the latter, the new term $R'_{x+e_i,\,y}$ at time $t+1$ is a bijective function of the corresponding term $R'_{x-e_i,\,y}$ at time $t-1$. By the inductive hypothesis, all the terms at time $t-1$ are uniform, i.i.d., and independent of the terms at time $t$; hence, the same is true of all the terms at time $t+1$. Therefore, \Cref{eq:symiid2} holds for $t$. This completes the induction, and the proof of \Cref{eq:symiid}.
	
	To see why this matters, recall that \Cref{eq:symicro3p} lets us uniquely compute the macrostate $C_{x+e_i}$ as a function of $(C_{x+n},R'_{x+n,\,x})_{n\in\mathcal N\setminus\{e_i\}}$, one time step at a time. At time step $t$, the microscopic variables in question have subscripts
	\[\{(x+n,\,x):\; x\in I(t,\,t),\; n\in\mathcal N\setminus\{e_i\}\}
	\subset I(t-1,\,t) \times I(t,\,t).\]
	
	By \Cref{eq:symiid}, these variables are uniform on $\mathbb Z_m$, i.i.d., and independent of all prior macrostates. Therefore, $C_{x+e_i}$ has homogeneous transition probabilities in terms of its Markovian parents $(C_{x+n})_{n\in\mathcal N\setminus\{e_i\}}$. Explicitly, the required causal graph $\mathrm G$ has the initial edges $\{(O,\,x):\; x\in I(0,1)\}$, along with the additional edges
	\[\{(x+n,\,x+e_i):\; x\in I(1,\,\infty),\;n\in\mathcal N\setminus\{e_i\}\}.\]
	
	To include the negative times, simply add a symmetric set of edges:
	\[\{(x-n,\,x-e_i):\; x\in I(-\infty,\,0),\;n\in\mathcal N\setminus\{e_i\}\}.\]
	
	The edges of $\mathrm G$ characterize the arrow of time: they point away from the initial set, forward along the time axis $i$ at positive times, and backward at negative times.

	\bibliography{time}
	
\end{document}